\documentclass[12pt]{article}

\usepackage[margin=1in]{geometry}

\usepackage[numbers]{natbib}
\usepackage[T1]{fontenc}    
\usepackage{hyperref}       
\usepackage{url}            
\usepackage{booktabs}       
\usepackage{amsfonts}       
\usepackage{amsmath, amssymb, amsthm}  
\usepackage{nicefrac}       
\usepackage{microtype}      
\usepackage[dvipsnames, table]{xcolor}      
\usepackage{pifont}
\usepackage{multirow}
\usepackage{nicematrix}
\usepackage{mathtools}
\usepackage{tikz-cd}
\usepackage[ruled, linesnumbered]{algorithm2e}
\usepackage{enumitem}   
\usepackage{siunitx}
\usepackage{makecell}

\theoremstyle{definition}
\newtheorem{definition}{Definition}[section]
\newtheorem{theorem}[definition]{Theorem}

\newtheorem{lemma}[definition]{Lemma}
\newtheorem{corollary}[definition]{Corollary}

\newtheorem{remark}[definition]{Remark}

\SetKwInput{KwParam}{Parameters}

\let\oldnl\nl 
\newcommand{\nonl}{\renewcommand{\nl}{\let\nl\oldnl}}

\newcommand{\bluecode}[1]{\small\texttt{\color{blue}#1}}

\newcommand{\cmark}{\color{ForestGreen} \ding{51}}
\newcommand{\xmark}{\color{BrickRed} \ding{55}}

\newcommand{\dimp}{p}

\DeclareMathOperator*{\argmin}{arg\,min}
\DeclareMathOperator{\rank}{rank}
\DeclareMathOperator{\diag}{diag}

\definecolor{blue1}{HTML}{6BC4FF}
\definecolor{blue2}{HTML}{5194FF}
\definecolor{blue3}{HTML}{3859FF}

\definecolor{my_orange}{HTML}{FE6100}
\definecolor{my_yellow}{HTML}{FFB000}

\makeatletter
\newtheorem*{rep@theorem}{\rep@title}
\newcommand{\newreptheorem}[2]{%
	\newenvironment{rep#1}[1]{%
		\def\rep@title{#2 \ref{##1}}%
		\begin{rep@theorem}}%
		{\end{rep@theorem}}}
\makeatother

\newreptheorem{definition}{Definition}
\newreptheorem{theorem}{Theorem}
\newreptheorem{proposition}{Proposition}
\newreptheorem{corollary}{Corollary}
\newreptheorem{lemma}{Lemma}

\title{Online Changepoint Detection via \\ Dynamic Mode Decomposition}

\author{Victor K. Khamesi  \and Niall M. Adams \and Dean A. Bodenham \and Edward A. K. Cohen}

\date{{\normalsize Department of Mathematics, Imperial College London,} \\
 {\normalsize South Kensington Campus, London SW7 2AZ, U.K.} \\
{\small \texttt{victor.khamesi21@imperial.ac.uk}, 
\quad \texttt{n.adams@imperial.ac.uk},
\quad \texttt{dean.bodenham@imperial.ac.uk},
\quad \texttt{e.cohen@imperial.ac.uk}}
}

\begin{document}

\maketitle

\begin{abstract}
	Detecting changes in data streams is a vital task in many applications. There is increasing interest in changepoint detection in the online setting, to enable real-time monitoring and support prompt responses and informed decision-making. Many approaches assume stationary sequences before encountering an abrupt change in the mean or variance. Notably less attention has focused on the challenging case where the monitored sequences exhibit trend, periodicity and seasonality. Dynamic mode decomposition is a data-driven dimensionality reduction technique that extracts the essential components of a dynamical system. We propose a changepoint detection method that leverages this technique to sequentially model the dynamics of a moving window of data and produce a low-rank reconstruction. A change is identified when there is a significant difference between this reconstruction and the observed data, and we provide theoretical justification for this approach. Extensive simulations demonstrate that our approach has superior detection performance compared to other methods for detecting small changes in mean, variance, periodicity, and second-order structure, among others, in data that exhibits seasonality. Results on real-world datasets also show excellent performance compared to contemporary approaches.
\end{abstract}

\section{Introduction}

Streaming data is ubiquitous in applications as diverse as cybersecurity \cite{cybersecurity}, speech recognition \cite{speech_recognition}, continual learning \cite{continual_learning}, finance \cite{finance}, meteorology \cite{meteorology}, biology \cite{biology}, and medicine \cite{medicine}. Detecting when the data generating mechanism changes is a crucial task in these applications. However, the time series data monitored in these applications often exhibit seasonality, such as higher values at certain times of the day, week, or year. Most changepoint detection methods are developed using the assumption that the data being analysed is piecewise-stationary. Such methods would be incapable of differentiating between changes in the distribution of the data and any seasonal effects.

In this paper, we propose a novel online changepoint detection method that can identify changepoints in streaming data while accounting for any seasonality in the time series. Our approach utilises a technique from dynamical systems called dynamic mode decomposition (DMD) \cite{dmdschmid, schmid_annual_reviews}, which can decompose a $p$-dimensional time series into spatial \emph{modes} and their associated \emph{dynamics}, encapsulating their frequency and growth rate. Focusing on the $r < p$ dominant modes, this provides a low-rank factorisation of the time series used to construct a statistic which indicates whether or not a change has occurred. We demonstrate the effectiveness of this approach on simulated data with seasonality that have changes in mean, variance, periodicity, and second-order structure, and on several real-world datasets. The method has several advantages: it is nonparametric, so the type of change does not need to be specified in advance; it has few parameters and is shown experimentally to be robust to choices for these parameters; it is unsupervised and does not require offline training. 

The rest of the paper is structured as follows: in Section \ref{sec:background}, we describe background for changepoint detection and DMD; in Section \ref{sec:methodology}, we introduce our novel method named \textbf{C}hange\textbf{P}oint \textbf{D}etection via \textbf{D}ynamic \textbf{M}ode \textbf{D}ecomposition (\textbf{CPDMD}); in Section \ref{sec:experiments_simulations}, we show that CPDMD achieves state-of-the-art performance on both synthetic and real-world data; in Section \ref{sec:conclusion}, we briefly summarise our conclusions. The experiments on synthetic data focus on the case where the data is univariate and has a single change, since we are already considering several different types of changes, but the real-world experiments cover the case of multivariate data with multiple changepoints.


\section{Background}
\label{sec:background}

We first provide an overview of the changepoint detection problem and dynamic mode decomposition algorithm to establish the necessary context and background.

\subsection{Changepoint detection}

Consider a stream of multivariate observations $x_1, x_2, \dots$ where $x_t \in \mathbb{R}^p$, $p \in \mathbb{N}$, denotes the \hbox{$p$-dimensional} observation at regularly spaced and discrete time $t$, sampled from i.i.d. random variables $X_1, X_2, \dots$ with changepoints $\tau_1, \tau_2, \dots$ such that
\begin{equation*}
	\begin{aligned}
		X_1, X_2, \dots, X_{\tau_1} &\sim F_1, \\
		X_{\tau_1+1}, X_{\tau_1+2}, \dots, X_{\tau_2} &\sim F_2, \\
		X_{\tau_2+1}, X_{\tau_2+2}, \dots, X_{\tau_3} &\sim F_3, \text{ etc.},
	\end{aligned}
\end{equation*}
where $F_1, F_2, \dots$ is a sequence of \textit{unknown} probability distributions such that $F_k \neq F_{k+1}$ for all $k \in \mathbb{N}$. Changepoint detection seeks to estimate the changepoint locations $\tau_1, \tau_2, \dots$. In other words, the problem consists in dividing the sequence of observations $x_1, x_2, \dots$ into piecewise segments with the same data generation process. In the statistics and machine learning literature, changepoint detection is also referred to as distribution shift detection \cite{distribution_shift} or temporal segmentation \cite{temporal_segmentation}. While many approaches consider the \emph{offline} setting \cite{pelt, binseg, ecp}, where the dataset is static and fully available, we are interested in the \emph{online} setting, where observations are processed sequentially and changepoints need to be detected as soon as possible after occurring. 

\subsubsection{Related work}
\label{sec:relatedwork}

The changepoint literature is vast, so we only highlight a few notable approaches. Two early online approaches that are still widely used are CUSUM \cite{cusum} and EWMA \cite{ewma}. In recent years, an online Bayesian method \cite{bocpd} and its extensions (e.g. \cite{bocpdms}) have become popular, although they can be computationally expensive for large sequences. Another class of approaches based on density ratio estimation includes KLIEP \cite{kliep}, uLSIF \cite{ulsif}, and RuLSIF \cite{rulsif}, with RuLSIF shown to have best performance among these three. Singular spectrum analysis has been adapted for changepoint detection in both univariate \cite{ssa_cpd} and multivariate contexts \cite{mssa}. Previous studies \cite{dmd_changepoint_1, dmd_changepoint_2} have considered using dynamic mode decomposition (DMD) for changepoint detection, although this was suggested as a secondary application after developing general DMD methodology, and only brief examples were provided without a comprehensive performance assessment.

\subsection{Dynamic mode decomposition}

Dynamic mode decomposition (DMD)  \cite{dmdschmid, dmdbook, schmid_annual_reviews} is a data-driven model reduction algorithm that aims to describe high-dimensional dynamical systems by discovering a low-dimensional subspace capturing dominant dynamics. 

\subsubsection{Continuous and discrete time}

\paragraph{Continuous time.} 
DMD assumes data is collected from an underlying dynamical system defined by $\frac{\mathrm{d} {x}}{\mathrm{d} t} = {f}({x}, t; \theta)$ where ${x}(t) \in \mathbb{R}^{\dimp}$ is the $p$-dimensional observation of the system's state at time $t$, $\theta$ contains the system's parameters, and ${f} : \mathbb{R}^p \times \mathbb{R} \longrightarrow \mathbb{R}^p$ characterises the \textit{unknown} dynamics. DMD linearly approximates the local dynamics using an operator $\mathcal{A} \in \mathbb{R}^{p \times p}$ such that $\frac{\mathrm{d} {x}}{\mathrm{d} t} \simeq \mathcal{A} {x}$. Given an initial condition ${x}(0)$, this differential equation has closed-form solution ${x}(t) = \sum_{j=1}^{p} \boldsymbol{\phi}_j e^{\omega_j t} b_j$, where $\omega_j \in \mathbb{C}$ and $\boldsymbol{\phi}_j \in \mathbb{C}^p$ are respectively the eigenvalues and eigenvectors of the matrix $\mathcal{A}$ and $b_j \in \mathbb{C}$ are the coordinates of the initial condition ${x}(0)$ in the eigenvectors basis, for $j \in \{ 1, 2, ..., p \}$. 

\paragraph{Discrete time.} In practice, one can only collect a discrete set of observations from the studied dynamical system. Therefore, DMD equations may be discretised considering equally time-spaced by $\Delta t$ observations ${x}_k = {x}(k \Delta t) \in \mathbb{R}^p$, $k \in \mathbb{N}$, and the dynamical evolution is now described by ${x}_{k+1} = {F}({x}_k)$ for some \textit{unknown} function ${F} : \mathbb{R}^p \longrightarrow \mathbb{R}^p$ representing the dynamics. Similar to continuous time, one may approximate the local dynamics such that ${x}_{k+1} \simeq {A} {x}_k$, where ${A}~=~\exp{(\mathcal{A} \Delta t)} \in \mathbb{R}^{p \times p}$. The discretised system has solution ${x}_{k+1} = \sum_{j=1}^{p} \boldsymbol{\phi}_j \lambda_j^{k} b_j$, where $\lambda_j \in \mathbb{C}$ and $\boldsymbol{\phi}_j \in \mathbb{C}^p$ are respectively the eigenvalues (dynamics) and eigenvectors (modes) of the matrix $A$ and $b_j \in \mathbb{C}$ are the coordinates of the initial condition ${x}_1$ in the eigenvectors basis, for $j \in \{ 1, 2, ..., p \}$. The analogy between continuous and discrete time DMD is illustrated in Figure \ref{fig:dmd-continuous-discrete}.

\paragraph{Computing the modes and dynamics.} There are several approaches to extract the DMD modes and dynamics from a sequence of snapshots \cite{dmdschmid, low_rank_dmd, schmid_annual_reviews}. We use the standard approach from \cite{dmdschmid}, which relies on singular value decomposition;  details are provided in Appendix \ref{sec:pseudo_code_dmd}. 

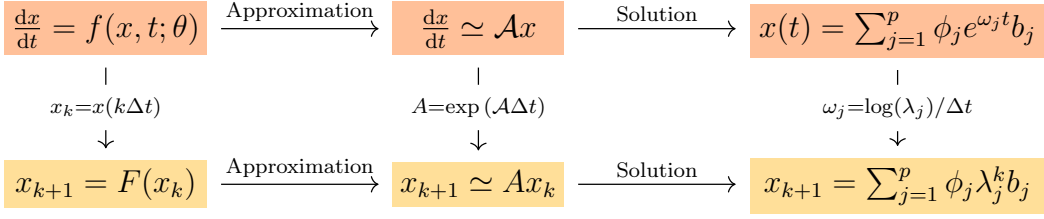
\begin{figure}[t]
	\centering
	\begin{tikzcd}[column sep=5pc, row sep=0.5pc]
		\colorbox{my_orange!40}{\makebox[6em][c]{$\frac{\mathrm{d} {x}}{\mathrm{d} t} = {f}({x}, t; \theta)$}} \arrow[dash]{d} \arrow{r}{\text{Approximation}} & \colorbox{my_orange!40}{\makebox[5em][c]{$\frac{\mathrm{d} {x}}{\mathrm{d} t} \simeq \mathcal{A} {x}$}} \arrow[dash]{d} \arrow{r}{\text{Solution}} & \colorbox{my_orange!40}{\makebox[9em][c]{${x}(t) = \sum_{j=1}^{p} \mathbf{\phi}_j e^{\omega_j t} b_j $}} \arrow[dash]{d} \\
		\scriptstyle {x}_k = {x}(k \Delta t) \arrow{d} & \scriptstyle {A} = \exp{( \mathcal{A} \Delta t )} \arrow{d} & \scriptstyle \omega_j = \log ( \lambda_j ) / \Delta t  \arrow{d} \\
		\colorbox{my_yellow!40}{\makebox[6em][c]{${x}_{k+1} = {F}({x}_k)$}} \arrow{r}{\text{Approximation}} & \colorbox{my_yellow!40}{\makebox[5em][c]{${x}_{k+1} \simeq {A} {x}_k$}} \arrow{r}{\text{Solution}} & \colorbox{my_yellow!40}{\makebox[9em][c]{${x}_{k+1} = \sum_{j=1}^{p} \mathbf{\phi}_j \lambda_j^{k} b_j $}}
	\end{tikzcd}
	\caption{DMD analogy between continuous time (top) and discrete time (bottom) settings.}
	\label{fig:dmd-continuous-discrete}
\end{figure}


\section{Proposed changepoint detection algorithm}
\label{sec:methodology}

We start by describing our approach to using dynamic mode decomposition (DMD) to 
extract the dynamics from a multivariate stream, then used to construct a statistic which reflects changes.

\subsection{Space-time transformation of the data}

Consider a time series $x_1, x_2, \dots$ where $x_t = \begin{bNiceMatrix}[margin]
	\Block[fill=blue1!40,rounded-corners]{1-1}{} x_t^{(1)} & \Block[fill=blue2!40,rounded-corners]{1-1}{} x_t^{(2)} & \cdots & \Block[fill=blue3!40,rounded-corners]{1-1}{} x_t^{(p)}
\end{bNiceMatrix}^{\intercal} \in \mathbb{R}^p$ denotes the $p$-dimensional observation at time $t \in \mathbb{N}$, $p \in \mathbb{N}$. Given a fixed window length $w \in \mathbb{N}$, let 
\begin{equation}
	X_t \vcentcolon = \left[ x_{t-w+1}, x_{t-w+2}, \dots, x_t \right] = \begin{bNiceMatrix}[margin]
		\Block[fill=blue1!40,rounded-corners]{1-4}{} x_{t-w+1}^{(1)} & x_{t-w+2}^{(1)} & \cdots & x_{t}^{(1)} \\
		\Block[fill=blue2!40,rounded-corners]{1-4}{}
		x_{t-w+1}^{(2)} & x_{t-w+2}^{(2)} & \cdots & x_{t}^{(2)} \\
		\vdots & \vdots & \ddots & \vdots \\
		\Block[fill=blue3!40,rounded-corners]{1-4}{}
		x_{t-w+1}^{(p)} & x_{t-w+2}^{(p)} & \cdots & x_{t}^{(p)} \\
	\end{bNiceMatrix} \vcentcolon = \begin{bNiceMatrix}[margin]
		\Block[fill=blue1!40,rounded-corners]{1-1}{} X_t^{(1)} \\
		\Block[fill=blue2!40,rounded-corners]{1-1}{} X_t^{(2)} \\
		\vdots \\
		\Block[fill=blue3!40,rounded-corners]{1-1}{} X_t^{(p)} \\
	\end{bNiceMatrix} 
	\label{eq:window}
\end{equation}
denote the concatenation of the last $w$ snapshots at time $t$, reflecting the dynamics of the signal within this context window, where $X_t^{(j)} \vcentcolon = \begin{bNiceMatrix}[margin]
	x_{t-w+1}^{(j)} & x_{t-w+2}^{(j)} & \cdots & x_{t}^{(j)}
\end{bNiceMatrix} \in \mathbb{R}^{1 \times w}$ for $j~\in~\{ 1, 2, \dots, p \}$, and 
$X_t \in \mathbb{R}^{p \times w}$. In the following, elements belonging to the same component are coloured identically.

Furthermore, at time $t$, given a fixed auto-regressive order $d \in \{ 1, 2, \dots, w \}$, one may rearrange the finite set of observations in the current window $X_t$ and construct each Hankel matrix \cite{hodmd, hankel_dmd_koopman} 
$\mathcal{X}_t^{(j)} \in \mathbb{R}^{d \times (w-d+1)}$  from $X_t^{(j)}$, along with their concatenation $\mathcal{X}_t \in \mathbb{R}^{pd \times (w-d+1)}$ as
\begin{equation}
	\forall j \in \{1, 2, \dots, p\}, ~
	\mathcal{X}_t^{(j)} \vcentcolon = \left[ \begin{array}{cccc}
		x_{t-w+1}^{(j)} & x_{t-w+2}^{(j)} & \cdots & x_{t-d+1}^{(j)} \\
		x_{t-w+2}^{(j)} & x_{t-w+3}^{(j)} & \cdots & x_{t-d+2}^{(j)} \\
		\vdots & \vdots & \ddots & \vdots \\
		x_{t-w+d}^{(j)} & x_{t-w+d+1}^{(j)} & \cdots & x_{t}^{(j)} \\
	\end{array} \right] 
	\Rightarrow   \mathcal{X}_t \vcentcolon = \begin{bNiceMatrix}[margin]
		\Block[fill=blue1!40,rounded-corners]{1-1}{}
		\mathcal{X}_t^{(1)} \\
		\Block[fill=blue2!40,rounded-corners]{1-1}{}
		\mathcal{X}_t^{(2)} \\
		\vdots \\
		\Block[fill=blue3!40,rounded-corners]{1-1}{}
		\mathcal{X}_t^{(p)} \\
	\end{bNiceMatrix}.  
	\label{eq:block_hankel}
\end{equation}
Formatting data within a Hankel matrix has been suggested by previous works in the dynamical systems \cite{hodmd, hankel_dmd_koopman} and changepoint detection \cite{ssa_cpd, subspace_identification} literature. The concatenated $\mathcal{X}_t$ is then used to learn the dynamical properties of the multivariate signal at time $t$ within the current sliding window.

\subsection{Learning and reconstructing the stream's dynamics}

Our method leverages DMD to learn the dynamical properties of the input stream and flag times at which the underlying generating process of the dynamical system appears to have changed. Indeed, DMD provides a linear state-space model approximation of a nonlinear dynamical system by operating on delay-embedded data snapshots. 

After computing $\mathcal{X}_t  \in \mathbb{R}^{pd \times (w-d+1)}$ in Equation~\eqref{eq:block_hankel},
we then perform DMD on $\mathcal{X}_t$ by considering its $(w-d+1)$ columns as $pd$-dimensional snapshots. Incorporating time-delay embeddings is a typical technique in signal processing \cite{ssa_cpd, mssa} and allows one to express ergodic attractors of non-linear dynamical systems \cite{schmid_annual_reviews}. Performing DMD on time-delay embeddings rather than on the original observations is also known as Hankel DMD \cite{hankel_dmd_koopman} or Higher-Order DMD \cite{hodmd}. In particular, we use this variant of DMD to be more flexible and account for signals with \emph{low} spatial dependencies. Furthermore, DMD on time-delay embeddings has been shown to converge to the true eigenfunctions of the Koopman operator \cite{hankel_dmd_koopman}. Hence, given a rank $r \leq \min \{pd, w-d+1\}$, we apply DMD on $\mathcal{X}_t$ as described in Algorithm \ref{alg:dmd} and obtain spatial modes $\mathbf{\Phi}_t \in \mathbb{C}^{pd \times r}$, dynamics $\mathbf{\Omega}_t \in \mathbb{C}^{r \times r}$ and low-rank reconstruction of the Hankel batch $\widehat{\mathcal{X}}_t \in \mathbb{R}^{pd \times (w-d+1)}$.

The DMD reconstruction of the Hankel batch $\widehat{\mathcal{X}}_t$ produces multiple different estimates of the same original observations due to time-delay embeddings, as illustrated in Equation~\eqref{eq:reconstruction_hankel_component}. Therefore, from the reconstructed delay-embedded snapshots $\widehat{\mathcal{X}}_t$, one can recover many different reconstructions of the windowed batch $X_t$. For each component $j \in \{1, 2, \dots, p\}$, our reconstruction is defined as
\begin{equation}
	\widehat{\mathcal{X}}_t^{(j)} \vcentcolon = \begin{bNiceMatrix}[margin]
		\Block[fill=gray!30,rounded-corners]{1-4}{} 
		\widehat{x}_{t-w+1}^{(j), 1} & \widehat{x}_{t-w+2}^{(j), 2} & \cdots & \Block[fill=gray!30,rounded-corners]{4-1}{} \widehat{x}_{t-d+1}^{(j), w-d+1} \\
		\widehat{x}_{t-w+2}^{(j), 1} & \widehat{x}_{t-w+3}^{(j), 2} &  \cdots & \widehat{x}_{t-d+2}^{(j), w-d+1} \\
		\vdots & \color{gray} \vdots & \ddots & \vdots \\
		\widehat{x}_{t-w+d}^{(j), 1} & \widehat{x}_{t-w+d+1}^{(j), 2} & \cdots & \widehat{x}_{t}^{(j), w-d+1} \\
	\end{bNiceMatrix}  \Rightarrow \widehat{X}_t^{(j)} \vcentcolon = \begin{bNiceMatrix}[margin]
		\Block[fill=gray!30,rounded-corners]{7-1}{}
		\widehat{x}_{t-w+1}^{(j), 1} \\
		\widehat{x}_{t-w+2}^{(j), 2} \\
		\vdots \\
		\widehat{x}_{t-d+1}^{(j), w-d+1} \\
		\widehat{x}_{t-d+2}^{(j), w-d+1} \\
		\vdots \\
		\widehat{x}_{t}^{(j), w-d+1} \\
	\end{bNiceMatrix}^{\intercal},
	\label{eq:reconstruction_hankel_component}
\end{equation}
i.e. by traversing the path of $\widehat{\mathcal{X}}_t^{(j)} \in \mathbb{R}^{d \times (w-d+1)}$ first along its first row and then along the last column to obtain $\widehat{X}_t^{(j)} \in \mathbb{R}^{1 \times w}$. Applying this block reconstruction to each component of the Hankel batch leads to a reconstruction of the original windowed batch again using concatenation 
\begin{equation}
	\widehat{\mathcal{X}}_t \vcentcolon = \begin{bNiceMatrix}[margin]
		\Block[fill=blue1!40,rounded-corners]{1-1}{}
		\widehat{\mathcal{X}}_t^{(1)} \\
		\Block[fill=blue2!40,rounded-corners]{1-1}{}
		\widehat{\mathcal{X}}_t^{(2)} \\
		\vdots \\
		\Block[fill=blue3!40,rounded-corners]{1-1}{}
		\widehat{\mathcal{X}}_t^{(p)} \\
	\end{bNiceMatrix} \in \mathbb{R}^{pd \times (w-d+1)} \Rightarrow \widehat{X}_t \vcentcolon = \begin{bNiceMatrix}[margin]
		\Block[fill=blue1!40,rounded-corners]{1-1}{}
		\widehat{X}_t^{(1)} \\
		\Block[fill=blue2!40,rounded-corners]{1-1}{}
		\widehat{X}_t^{(2)} \\
		\vdots \\
		\Block[fill=blue3!40,rounded-corners]{1-1}{}
		\widehat{X}_t^{(p)} \\
	\end{bNiceMatrix} \in \mathbb{R}^{p \times w}.
	\label{eq:reconstruction_hankel}
\end{equation}
While this reconstruction is not unique, the one we consider in Equations~\eqref{eq:reconstruction_hankel_component} and \eqref{eq:reconstruction_hankel} may be considered as \textit{the most causal} since it only requires the first column of the windowed batch $X_t$, i.e. the true observation $x_{t-w+1}$, to compute the remaining $w-1$ columns of the windowed batch $\widehat{x}_{t-w+2}, \widehat{x}_{t-w+3}, \dots, \widehat{x}_{t}$ similarly to classical DMD settings. Our sequential transformation of the data and reconstruction of its dynamics is illustrated in Figure \ref{fig:summary_dmd}. 

\begin{figure}[t]
	\centering
	\includegraphics[width=\linewidth]{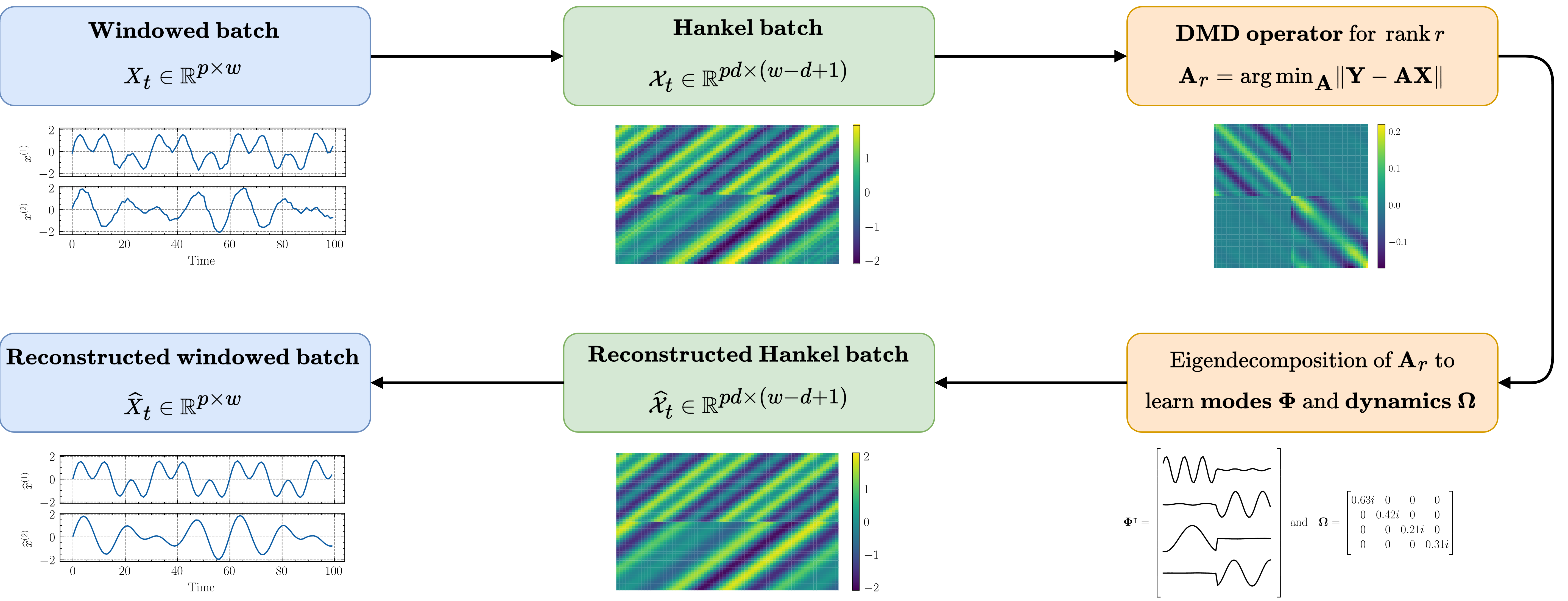}
	\caption{Illustration of the sequential data stream processing and reconstruction via DMD.}
	\label{fig:summary_dmd}
\end{figure}

\subsection{Detecting changepoints based on reconstruction error}

After obtaining a low-rank reconstruction of the windowed batch $\widehat{X}_t$, our method consists in detecting shifts in the data generating mechanism by monitoring the reconstruction quality \cite{reconstruction_error_anomaly}. Let $\varepsilon_t$ denote the reconstruction error on the windowed batches, and let $\delta_t$ be the error increments at time $t$, where
\begin{equation}
	\varepsilon_t \vcentcolon = \frac{1}{p w} \lVert X_t - \widehat{X}_t \rVert^2_F \in \mathbb{R}^{+}, \qquad \delta_t \vcentcolon = \varepsilon_t - \varepsilon_{t-1} \in \mathbb{R}.
	\label{eq:reconstruction_error}
\end{equation} 
Distribution shifts in the data are expected to manifest as perturbations in the reconstruction error over time.
While consecutive reconstruction error terms may be correlated due to their overlapping set of considered observations, their increments would be better modelled as stationary. Therefore, changepoints in the data can be detected by analysing the gradients of the reconstruction error. In discrete time settings with uniformly spaced observations, this approach is equivalent to monitoring the error increments, as defined in Equation~\eqref{eq:reconstruction_error}. To detect changes in the gradient of the reconstruction error, we use an adaptive version of EWMA \cite{ewma}, detailed in Appendix \ref{sec:background}. This distribution-free method has two hyperparameters, a learning rate $\lambda \in [0, 1]$ and a control limit $L \in \mathbb{R}^{*}_{+}$.

\subsection{Algorithm summary}

In summary, our proposal translates the task of changepoint detection in multivariate data, involving potentially different types of change, into the problem of univariate changepoint detection in the reconstruction error, which is computed efficiently at each time step within sequential context windows. The reconstruction method is based on the DMD algorithm applied to time-delay embeddings of the original data, extracting accurate representations of the dynamics. Changes in the input stream may be detected as perturbations of the reconstruction error increments using an adaptive EWMA algorithm. This approach is summarised in Algorithm \ref{alg:single_changepoint} for single changepoint detection. 

\begin{algorithm}[t]
	\caption{$\widehat{\tau} \gets\texttt{SingleCP}(x_1, x_2, \dots, x_T \mid T_0, w, d, r, \lambda, L)$}\label{alg:single_changepoint}
	
	\nonl \bluecode{/* Single changepoint detection via dynamic mode decomposition */}
	
	\KwIn{Sequence of observations $x_1, x_2, \dots, x_T \in \mathbb{R}^p$}
	
	\KwOut{First detected changepoint in the sequence $\widehat{\tau}$}
	
	\KwParam{Burn-in period $T_0 \in \mathbb{N}$, window length $w \in \mathbb{N}$, auto-regressive order $d \leq w$, SVD rank $r \leq \min \{ pd, w-d+1 \}$, adaptive EWMA learning rate $\lambda \in [0, 1]$ and limit $L \in \mathbb{R}^{\ast}_{+}$}
	
	\For{$t \in \{ w, w+1, \dots, T \}$}{
		$X_t \gets [ x_{t-w+1}, x_{t-w+2}, \dots, x_t ]$ \bluecode{\hspace{0px} /* Windowed batch (Equation \eqref{eq:window}) */} \\
		$\mathcal{X}_t \gets \texttt{Hankel}(X_t \mid d)$ \bluecode{\hspace{0px} /* Hankel batch (Equation \eqref{eq:block_hankel}) */} \\
		
		$\{ \widehat{\mathcal{X}}_t, \mathbf{\Phi}_t, \mathbf{\Omega}_t \} \gets \texttt{DMD}(\mathcal{X}_t \mid r)$ \bluecode{\hspace{0px} /* DMD on Hankel batch (Algorithm \ref{alg:dmd})  */} \\
		$\widehat{X}_t \gets \texttt{Unroll}(\widehat{\mathcal{X}}_t \mid d)$ 	\bluecode{\hspace{0px} /* Windowed batch reconstruction (Equations \eqref{eq:reconstruction_hankel_component} and \eqref{eq:reconstruction_hankel}) */} \\
		$\varepsilon_t \gets \lVert X_t - \widehat{X}_t \rVert^2_F$	\bluecode{\hspace{0px} /* Reconstruction error (Equation \eqref{eq:reconstruction_error}) */} \\
		$\delta_t \gets \varepsilon_t - \varepsilon_{t-1}$ \\
		$Z_t \gets (1 - \lambda) Z_{t-1} + \lambda \delta_t$ \bluecode{\hspace{0px} /* Adaptive EWMA on increments (Definition \ref{def:adaptive_ewma}) */} \\
		\If{$t > T_0 ~ \mathbf{and} ~ \left(Z_t > \mu_t + L \sigma_{Z_t} ~ \mathbf{or} ~ Z_t < \mu_t - L \sigma_{Z_t}\right)$}{
			\Return{$\widehat{\tau} = t$} \bluecode{\hspace{0px} /* Detected changepoint */}
		}
	}
	\Return{$\widehat{\tau} = \mathtt{None}$} \bluecode{\hspace{0px} /* No detected changepoint */} \\
	
\end{algorithm}

\subsection{Hyperparameter selection}
\label{sec:model_selection}

Algorithm \ref{alg:single_changepoint} requires specifying hyperparameters such as the window size $w$, the auto-regressive order $d$, and the rank $r$. To address the lack of prior knowledge for setting these parameters, we propose an approach that determines $w$, $d$, and $r$ based on a burn-in period $T_0 \in \mathbb{N}$ in an unsupervised and data-driven manner.

Assuming a fixed burn-in period $T_0$, we define a grid $\Theta$ satisfying the parameters constraints
\begin{equation*}
	\Theta \vcentcolon = \left\{ (w, d, r) \in \mathbb{N}^{3} \mid w \leq T_0, d \leq w, r \leq \min{ \{ pd, w-d+1 \}} \right\},
	\label{eq:parameters_constraints}
\end{equation*}
where $p$ is the dimension of the input stream. Running parallel competing models \cite{bocpdms}, hyperparameter selection is achieved by minimising the average reconstruction error, i.e.
\begin{equation*}
	\theta^{\ast} \vcentcolon = \argmin_{(w, d, r) \in \Theta} \frac{1}{T_0-w+1} \sum_{t=w}^{T_0-w+1} \varepsilon_t.
\end{equation*}
This approach is summarised in Algorithm \ref{alg:model_selection} and discussed further in Appendix \ref{sec:appendix_model_selection}. The remaining adaptive EWMA parameters $\lambda$ and $L$ are set by the user and control the overall sensitivity, although suggested empirical values from the EWMA literature \cite{ewma_suggested_params} are $\lambda \in [0.05, 1]$ and $L \in [2.4, 3]$.


\subsection{Theoretical analysis}
\label{sec:theory}

\begin{theorem}
	\label{th:bound_eigendecomposition}
	Consider a $p$-dimensional stream of observations monitored via Algorithm \ref{alg:single_changepoint}. Let $\mathbf{A} \in \mathbb{C}^{pd \times pd}$ be the DMD operator at time $t$ and let $\mathbf{\Tilde{A}} \in \mathbb{C}^{pd \times pd}$ denote the DMD operator at time $t+1$. Then, there exists a perturbation matrix $\mathbf{E} \in \mathbb{C}^{pd \times pd}$ such that $\mathbf{\Tilde{A}} = \mathbf{A} + \mathbf{E}$, where the closed-form of $\mathbf{E}$ is detailed in Appendix \ref{sec:theoretical_analysis_eig}. 
	Furthermore, assuming that the DMD operator $\mathbf{A}$ is diagonalisable, then the following results hold:
	\begin{enumerate}[label=(\roman*)]
		\item Eigenvalues (dynamics): let $\Tilde{\lambda} \in \sigma(\mathbf{\Tilde{A}})$, i.e. $\Tilde{\lambda}$ is an eigenvalue of $\mathbf{\Tilde{A}}$, then 
		\begin{equation*}
			\min_{\lambda \in \sigma(\mathbf{A})} \lvert \Tilde{\lambda} - \lambda \rvert \leq \kappa_2(\mathbf{\Phi}) \lVert \mathbf{E} \rVert_2,
		\end{equation*} 
		where $\lVert \mathbf{E} \rVert_2$ denotes the matrix $2$-norm of $\mathbf{E}$, $\kappa_2(\mathbf{\Phi}) = \lVert \mathbf{\Phi} \rVert_2 \lVert \mathbf{\Phi}^{-1} \rVert_2$ is the condition number of $\mathbf{\Phi} \in \mathbb{C}^{pd \times pd}$ and $\mathbf{\Phi}$ is the eigenvectors matrix of $\mathbf{A}$. 
		\item Eigenvectors (modes): let $\boldsymbol{\Tilde{\phi}}_k \in \mathbb{C}^{pd}$ and $\boldsymbol{\phi}_k \in \mathbb{C}^{pd}$ denote normalised eigenvectors of $\mathbf{\Tilde{A}}$ and $\mathbf{A}$ respectively for $k \in \{ 1, 2, \dots, pd\}$, then
		\begin{equation*}
			\lVert \boldsymbol{\Tilde{\phi}}_k - \boldsymbol{\phi}_k \rVert_2 \leq \frac{\lVert \mathbf{E} \rVert_2}{\min_{j \neq k} \lvert \lambda_k - \lambda_j \rvert} + \mathcal{O}\left(\lVert \mathbf{E} \rVert_2^2\right), \forall k \in \{1, 2, \dots, pd\}.
		\end{equation*}
	\end{enumerate}
\end{theorem}

\paragraph{Consequences.} 
Theorem \ref{th:bound_eigendecomposition} shows that, following our approach, within consecutive time steps in which we model the modes and dynamics of the input stream as pairs of eigenvalues and eigenvectors of the DMD operator, the recently estimated modes and dynamics are \textit{close} to the previous ones provided that the perturbation due to the sliding window and the new observed datum is \textit{small}. Indeed, the closed-form of $\mathbf{E}$ in Appendix \ref{sec:theoretical_analysis_eig} shows that the perturbation is proportional to the predictive error of the previous DMD model on the new snapshot. However, if the underlying generating process has not changed, the predictive accuracy remains \emph{high}, provided that parameters are set to minimise the error as discussed in Section \ref{sec:model_selection}. Therefore, when monitoring a sequence with no change in the underlying generating process, modes and dynamics extracted by the sequential DMD models are similar, leading to a \textit{stable} reconstruction error.

\subsection{Computational complexity}
\label{sec:complexity}

\begin{theorem}
	The computational complexity of Algorithm \ref{alg:single_changepoint} in the context of single changepoint detection is $\mathcal{O}(pd(w-d) \cdot \min \{ pd, w-d \})$, where $p$ is the dimension of the input stream, $w$ is the window length and $d$ denotes the auto-regressive order. 
	\label{th:complexity}
\end{theorem}

Theorem \ref{th:complexity} shows that complexity scales linearly with the dimension, in addition to that our proposed approach is efficient and its complexity is simply determined by its parameters. Figure \ref{fig:complexity} in Appendix \ref{sec:theoretical_analysis_complexity} illustrates how the complexity of our proposed method for single changepoint detection scales with varying dimension, window length and auto-regressive order.  Further discussion on complexity can be found in Appendix \ref{sec:theoretical_analysis_complexity}, particularly Corollary \ref{th:complexity_model_selection} on hyperparameter selection.


\section{Experiments and simulations}
\label{sec:experiments_simulations}

Algorithm \ref{alg:multiple_changepoints_with_model_selection} in Appendix \ref{sec:appendix_cpdmd}  describes our proposed method with auto-adaptive parameters combining Algorithms \ref{alg:single_changepoint} and \ref{alg:model_selection} to process data streams with multiple changepoints. We conduct a comparative evaluation of the CPDMD algorithm against six established changepoint detection methods: EWMA \cite{ewma}, EWMVar \cite{ewmvar}, RuLSIF \cite{rulsif}, BOCPDMS \cite{bocpdms}, mSSA and its variant mSSA-MW \cite{mssa}. A comparison of these selected algorithms capabilities can be found in Table \ref{tab:methods_capabilities} in Appendix \ref{sec:synthetic_models_parameters}. We do not compare to \cite{dmd_changepoint_1, dmd_changepoint_2} cited in Section~\ref{sec:relatedwork}; those works focused on developing DMD methodology in which changepoint detection was only considered as a secondary application, without a comprehensive performance assessment, and implementations of these methods are unavailable.

The simulation study presented in Section \ref{sec:experiments_synthetic} focuses on online changepoint detection in synthetic univariate time series that may exhibit trend and seasonality, with a single change in the data generating mechanism. We focus our simulation study on the single changepoint setting, which is a common approach \cite{ross_cpm, single_cp_chapter}, as poor performance in this simpler case is indicative of potential struggles with the more challenging multiple changepoint scenario, and the latter may be addressed, among other approaches, by iteratively restarting a performing single changepoint detector. 
We reserve comprehensive multivariate simulations for future work in the interest of  conciseness, since 
changepoint detection of multivariate data is more complex with a wider variety of change types. However, since the problem of multiple changepoint detection in multivariate data remains important across various application domains, we demonstrate the effectiveness of CPDMD on real-world datasets in Section~\ref{sec:experiments_real_world}.

\begin{figure}[t]
	\centering
	\includegraphics[width=\linewidth]{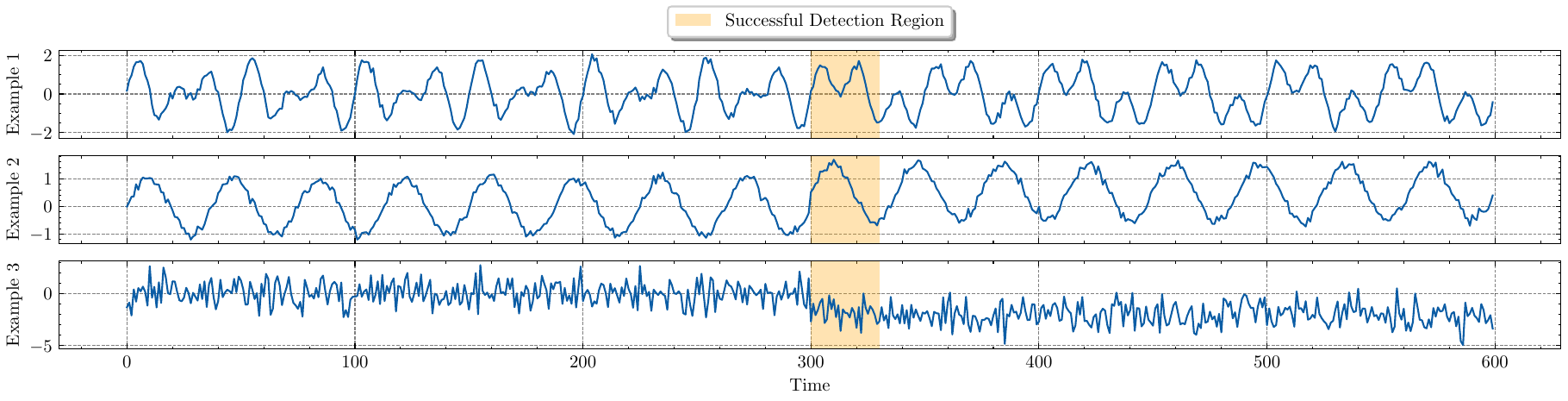}
	\caption{Three examples of sequences considered in synthetic data simulations.}
	\label{fig:illustration_double_proxy}
\end{figure}

\subsection{Synthetic data}
\label{sec:experiments_synthetic}

\paragraph{Types of change.}
We evaluate algorithms performance on seasonal univariate streams containing a specific type of change: periodicity, location, amplitude, trend, and double periodicity. We also consider streams that do not present any seasonality, with changes in the mean or variance. See Figure~\ref{fig:illustration_double_proxy} for examples of these sequence types and Appendix \ref{sec:data_generation_process} for additional details.

\paragraph{Simulation settings.}
Each algorithm requires parameters to be set. To reduce any bias induced by parameters specification, we follow the recommendations from \cite{cpd_evaluation} by reporting the performance obtained with default parameters, along with the best performance from a grid of parameters which is defined for each selected algorithm in Appendix \ref{sec:synthetic_models_parameters}. 

\paragraph{Performance metrics.}
We assess performance via well-established online changepoint detection performance metrics: precision $P$, recall $R$, $F1$-Score and Average Run Lengths $\text{ARL}_0$ and $\text{ARL}_1$. We provide definitions of these metrics in Appendix \ref{sec:background}. 
The standard errors for $P$, $R$, and $F_1$, computed via a bootstrap procedure, and $\text{SDRL}_0$ and $\text{SDRL}_1$ are shown in parentheses.
An effective algorithm would aim to maximise the $F1$-Score and $\text{ARL}_0$, while minimising $\text{ARL}_1$.

\paragraph{Results.}
Table \ref{tab:performance_synthetic_without_bocpdms} shows that CPDMD significantly outperforms every other considered method in terms of precision, recall and  $F_1$-Score for both best and default parameters, while the corresponding values of $\text{ARL}_0$ and $\text{ARL}_1$ are often ranked among the two best ones. However, since no other method achieves similar $F_1$-Scores, the interpretation is more complex and strictly comparing raw values would make the comparison of $\text{ARLs}$ biased. This is indeed well shown and complemented by Figure \ref{fig:simulations_global_without_bocpdms} which provides a better overview of the performance of each algorithm across all tested sets of parameters. CPDMD reaches the highest values of $F_1$-Score for the lowest values of $\text{ARL}_1$. Moreover, given a desired sensitivity defined by a desired $\text{ARL}_0$, simulations show that there always exists a CPDMD model which obtains the best $F_1$-Score. Another important aspect of changepoint detection algorithms is the scaling law between $\text{ARL}_0$ and $\text{ARL}_1$, and our proposed approach shows among the best scaling laws, i.e. given any desired value of $\text{ARL}_0$, it has among the lowest $\text{ARL}_1$. CPDMD is also shown to be robust to parameters selection, with close performance of default and best parameters, along with no widely scattered points in Figure \ref{fig:simulations_global_without_bocpdms}. Our proposed approach performance is followed by EWMA and  RuLSIF, while the latter was originally proposed in retrospective changepoint detection and could be adapted to sequential context using an \emph{a priori unknown} threshold as the decision rule. Note that due to its increased computational complexity on long sequences, the $\text{ARL}_0$ of BOCPDMS is not estimated: however, precision, recall, $F_1$-Score and $\text{ARL}_1$ of the BOCPDMS algorithm on synthetic data is detailed in Appendix \ref{sec:global_with_bocpdms}. Results per change type are provided in Appendix \ref{sec:experiment_details_synthetic}.

\begin{table}[t]
	\centering
	\smaller
	\caption{Evaluation of selected changepoint detection algorithms on synthetic data, comparing performance achieved with the parameter set yielding highest $F_1$-Score (\textit{best}) against default parametrisation (\textit{default}). We highlight the top two performance metrics for each set of parameters in bold.}
	\begin{tabular}{l l r@{\hskip\tabcolsep} >{\smaller}r@{\hskip\tabcolsep} r@{\hskip\tabcolsep} >{\smaller}r@{\hskip\tabcolsep} r@{\hskip\tabcolsep} >{\smaller}r@{\hskip\tabcolsep} r@{\hskip\tabcolsep} >{\smaller}r@{\hskip\tabcolsep} r@{\hskip\tabcolsep} >{\smaller}r@{\hskip\tabcolsep}}
		\toprule
		\textbf{Algorithm} & \textbf{Params.} & \multicolumn{2}{c}{$\boldsymbol{P}$} & \multicolumn{2}{c}{$\boldsymbol{R}$} & \multicolumn{2}{c}{$\boldsymbol{F_1}$} & \multicolumn{2}{c}{$\mathbf{ARL_1}$} & \multicolumn{2}{c}{$\mathbf{ARL_0}$} \\
		\midrule
		\multirow{2}{*}{EWMA} & Best & \textbf{.669} & (.004) & \textbf{.474}  & (.003) & \textbf{.555} & (.003) & 9.78 & (7.44) & \textbf{71085.25} & (41976.94) \\
		& Default & .188 & (.003) & .181 & (.003) & .185 & (.003) & \textbf{5.88} & (5.11) & 4261.15 & (18156.77) \\
		\hline
		\multirow{2}{*}{EWMVar} & Best & .275 & (.005) & .131 & (.002) & .178 & (.003) & 10.96 & (6.86) & \textbf{71202.77} & (45092.24) \\
		& Default & \textbf{.452} & (.009) & .089 & (.002) & .148 & (.003) & 12.36 & (7.40) & \textbf{99900.00} & (0.00) \\
		\hline
		\multirow{2}{*}{RuLSIF} & Best & .537 & (.013) & .414 & (.010) & .468 & (.011) & 19.01 & (6.80) & 41934.74 & (49177.96) \\
		& Default & .448 & (.011) & \textbf{.364} & (.010) & \textbf{.401} &(.010)  & 23.58 & (4.61) & 57147.14 & (49746.53) \\
		\hline
		\multirow{2}{*}{mSSA} & Best & .475 & (.005) & .266 & (.003) & .341 & (.004) & \textbf{6.91} & (6.09) & 43786.39 & (47153.54) \\
		& Default & .304 & (.003) & .210 & (.002) & .248 & (.003) & \textbf{6.28} & (5.80) & 27023.57 & (42587.84) \\
		\hline
		\multirow{2}{*}{mSSA-MW} & Best & .443 & (.005) & .313 & (.004)& .367 & (.004) & 8.79 & (5.87) & 13834.54 & (32965.67) \\
		& Default & .043 & (.001) & .041 & (.001) & .042 & (.001) & 9.11 & (7.29) & 8210.28 & (20999.07) \\
		\hline
		\multirow{2}{*}{\textbf{CPDMD}} & Best & \textbf{.960} & (.001) & \textbf{.902} & (.002) & \textbf{.930} & (.002) & \textbf{7.11} & (6.14) & 4034.30 & (3997.10) \\
		& Default & \textbf{.978} & (.001) & \textbf{.807} & (.003) & \textbf{.884} & (.002) & 9.53 & (6.47) & \textbf{58341.28} & (38161.07) \\
		\bottomrule
	\end{tabular}
	\label{tab:performance_synthetic_without_bocpdms}
\end{table}

\begin{figure}[t]
	\centering
	\includegraphics[width=\linewidth]{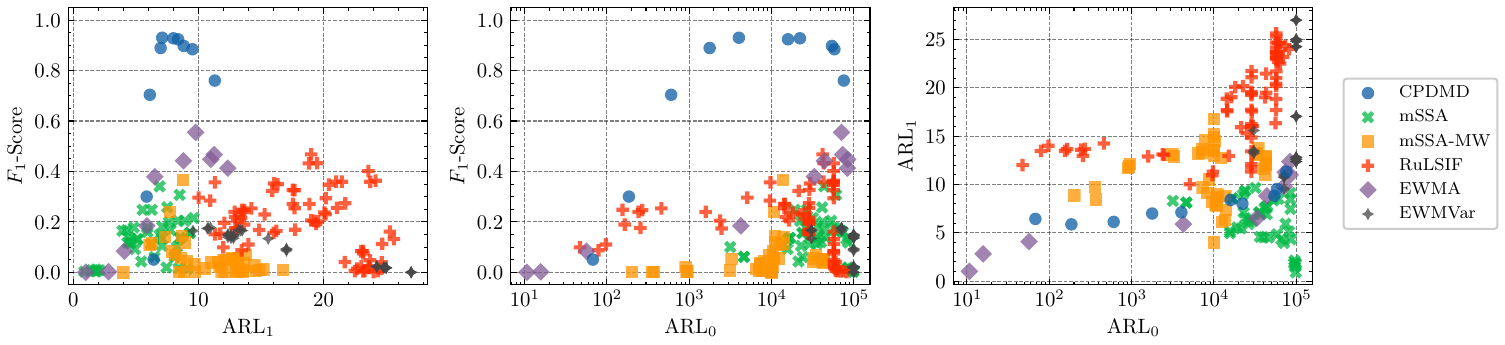}
	\caption{Performance comparison of changepoint detection algorithms on synthetic data across all scenarios. Each dot corresponds to an algorithm with a specific choice of parameter values.}
	\label{fig:simulations_global_without_bocpdms}
\end{figure}

\subsection{Real-world data}
\label{sec:experiments_real_world}

\paragraph{Datasets.}
We consider three real-world datasets from various application domains. The HASC dataset \cite{hasc} is an activity recognition dataset, consisting of three-dimensional acceleration measurements from wearable devices, in which the goal is to detect when a user switches from an activity to another. The Digits dataset consists of sequences of $8 \times 8$ greyscale images extracted from \cite{digits}, where images of the same digit are displayed for a specific duration before switching to a different digit. The Yahoo Webscope S5 dataset \cite{yahoo} is a collection of univariate time series extracted from real-world web traffic. We provide additional description of these datasets in Appendix \ref{sec:datasets}.

\paragraph{Simulation settings.}
Similarly to simulations on synthetic data, we evaluate the performance of each selected algorithm on a grid of parameters with consistency across burn-in periods of each algorithm and dataset. This grid is detailed in Appendix \ref{sec:appendix_real_world_parameters}. Note that both EWMA and EWMVar are univariate methods, and there may exist many different ways to use those on multivariate streams. In this work, we adapt them to the multivariate datasets by running separate instances per component and taking the union of detected changepoints across all components.

\paragraph{Performance metrics.}
We assess the empirical performance of each algorithm via $F_1$-Score and covering metric $\mathcal{C}$, both defined in Appendix \ref{sec:background}. Since changepoint locations in real-world data  are the result of human annotations, shifts may be detected slightly before their location times, making the computation of $\text{ARL}_1$ intractable. Instead, we use the covering metric, often used in offline changepoint detection as a good measure of the quality of a segmentation with multiple changepoints.

\paragraph{Results.}
Table \ref{tab:performance_real_world} demonstrates superior overall performance of CPDMD compared to other methods. CPDMD consistently ranks among the top two methods across the two considered metrics. Similar to the synthetic data simulations, this experiment suggests that CPDMD is robust to parameter selection, achieving near-optimal $F_1$-Scores with default parameters, particularly on the HASC and Yahoo datasets. On the Digits dataset, CPDMD outperforms all other algorithms, indicating its suitability for high-dimensional datasets. While some methods, such as mSSA-MW on the HASC dataset and mSSA on the Yahoo dataset, show comparable performance, no single method matches CPDMD's overall effectiveness. RuLSIF's performance cannot be computed for the Digits dataset due to matrix inversion issues, likely due to high dimensionality and low variability of certain pixels. Overall, the performance gap between CPDMD and other methods is smaller on real-world data compared to synthetic data. This may be due to CPDMD's sensitivity to subtle changes, which might not be labelled as changepoints in the datasets, increasing false positives. Additionally, annotation fluctuations, especially in the HASC dataset, might cause delays in labelling actual changes. Lastly, results on the Yahoo dataset demonstrate CPDMD's efficiency for both change and anomaly detection.

\begin{table}[t]
	\centering
	\smaller
	\caption{Evaluation of selected changepoint detection algorithms on real-world data, comparing performance achieved with the parameter set yielding highest $F_1$-Score (\textit{best}) against default parametrisation (\textit{default}). We highlight the top two performance metrics for each set of parameters in bold.}
	\begin{tabular}{l l r@{\hskip\tabcolsep} >{\scriptsize}r@{\hskip\tabcolsep} r@{\hskip\tabcolsep} >{\scriptsize}r@{\hskip\tabcolsep} r@{\hskip\tabcolsep} >{\scriptsize}r@{\hskip\tabcolsep} r@{\hskip\tabcolsep} >{\scriptsize}r@{\hskip\tabcolsep} r@{\hskip\tabcolsep} >{\scriptsize}r@{\hskip\tabcolsep} r@{\hskip\tabcolsep} >{\scriptsize}r@{\hskip\tabcolsep}}
		\toprule
		\multirow{2}{*}{\textbf{Algorithm}} & \multirow{2}{*}{\textbf{Params.}} & \multicolumn{4}{c}{\textbf{HASC}} & \multicolumn{4}{c}{\textbf{Digits}} &  \multicolumn{4}{c}{\textbf{Yahoo}} \\
		& & \multicolumn{2}{c}{$\boldsymbol{F_1}$} & \multicolumn{2}{c}{$\boldsymbol{\mathcal{C}}$}  & \multicolumn{2}{c}{$\boldsymbol{F_1}$} & \multicolumn{2}{c}{$\boldsymbol{\mathcal{C}}$} & \multicolumn{2}{c}{$\boldsymbol{F_1}$} & \multicolumn{2}{c}{$\boldsymbol{\mathcal{C}}$} \\ 
		\midrule
		\multirow{2}{*}{EWMA} & Best & .347 & (.086) & .470 & (.081) & .067 & (.005) & .329 & (.044) & .454 & (.300) & .593 & (.237)  \\
		& Default & .216 & (.030) & .400 & (.072) & .025 & (.001) & .129 & (.017) & .240 & (.248) & .428 & (.183)  \\
		\midrule
		\multirow{2}{*}{EWMVar} & Best & .460 & (.069) & \textbf{.638} & (.077) & .091 & (.006) & .262 & (.029) & .442 & (.365) & .700 & (.240)  \\
		& Default & .384 & (.039) & \textbf{.615} & (.059) & .053 & (.004) & .183 & (.018) & .380 & (.358) & .642 & (.255)  \\
		\hline
		\multirow{2}{*}{BOCPDMS} & Best & .209 & (.121) & .466 & (.096) & .037 & (.001) & .036 & (.002) & .096 & (.217) & .616 & (.273)  \\
		& Default & .194 & (.114) & .456& (.089) & .013 & (.002) & .032 & (.004) & .080 & (.195) & .610 & (.273)  \\
		\hline
		\multirow{2}{*}{RuLSIF} & Best & .420 & (.080) & .512 & (.055) & N/A & (N/A) & N/A & (N/A) & .247 & (.352) & .573 & (.221)  \\
		& Default & .267 & (.095) & \textbf{.563} & (.073) & N/A & (N/A) & N/A & (N/A) & .232 & (.361) & .614 & (.234)  \\
		\hline
		\multirow{2}{*}{mSSA} & Best & .213 & (.169) & .260 & (.145) & .073 & (.162) & .265 & (.096) & \textbf{.540} & (.325) &\textbf{.710} & (.212)  \\
		& Default & .183 & (.126) & .239 & (.137) & .073 & (.162) & .267 & (.090) & \textbf{.540} & (.325) & \textbf{.710} & (.212)  \\
		\hline
		\multirow{2}{*}{mSSA-MW} & Best & \textbf{.486} & (.096) & .586 & (.046) & \textbf{.206} & (.200) & \textbf{.555} & (.121) & .403 & (.298) & .578 & (.211)  \\
		& Default & \textbf{.403} & (.086) & .557 & (.063) & \textbf{ .145} & (.130) & \textbf{.592} & (.106) & .399 & (.296) & .548 & (.224)  \\
		\hline
		\multirow{2}{*}{\textbf{CPDMD}} & Best & \textbf{.499} & (.077) & \textbf{.608} & (.065) & \textbf{.645} & (.274) & \textbf{.657} & (.148) & \textbf{.526} & (.332) & \textbf{.726} & (.199)  \\
		& Default & \textbf{.476} & (.109) & .539 & (.092) & \textbf{.483} & (.296) & \textbf{.484} & (.200) & \textbf{.526} & (.332) & \textbf{.726} & (.199)  \\
		\bottomrule
	\end{tabular}
	\label{tab:performance_real_world}
\end{table}


\section{Conclusion}
\label{sec:conclusion}
In this paper we propose CPDMD, a nonparametric data-driven approach that leverages DMD to detect changepoints in multivariate streaming data for a variety of change types. This method has excellent performance on both synthetic and real-world datasets.

\paragraph{Limitations.} Since CPDMD detects a change in reconstruction error, it is unclear what type of change actually occurs; monitoring the modes and dynamics directly could provide insight into the type of change. The hyperparameter selection approach is based on a grid search, which is performant but could potentially be improved. One potential improvement could be to replace the adaptive EWMA with a method better suited to detect changes in the reconstruction error, although we did not find a better method from empirical results. These research areas are left for future work.

\section{Acknowledgements}
Victor Khamesi is funded by a Roth Scholarship from the Department of Mathematics, Imperial College London. Ed Cohen acknowledges funding from the EPSRC, grant number EP/X002195/1.

\clearpage


\normalsize
\appendix

\newpage
\section{Background}
\label{sec:background_appendix}

\begin{definition}[EWMA algorithm \cite{ewma}]
	Consider a univariate stream $y_1, y_2, \dots$ with mean $\mu$ and variance $\sigma^2$. Let $\lambda \in [0, 1]$ denote the learning rate and $L \in \mathbb{R^{\ast}_{+}}$ a control parameter for the algorithm's sensitivity. Let $Z_0 = \mu$ and, for $t \in \mathbb{N}^{\ast}$
	\begin{equation*}
		Z_t \vcentcolon = (1 - \lambda) Z_{t-1} + \lambda y_t.
	\end{equation*}
	Then, the standard deviation of $Z_t$ is 
	\begin{equation*}
		\sigma_{Z_t} = \sigma \sqrt{\frac{\lambda}{2-\lambda} \left( 1 - (1 - \lambda)^{2t} \right)},
	\end{equation*}
	and	a change is detected when either $Z_t > \mu + L \sigma_{Z_t}$ or $Z_t < \mu - L \sigma_{Z_t}$.
	\label{def:ewma}
\end{definition}

We introduce in Definition \ref{def:adaptive_ewma} a slightly modified version of the original EWMA algorithm introduced by \cite{ewma} where the mean and variance of the input stream are sequentially updated rather than being set to estimates on historical data. 

\begin{definition}[Adaptive EWMA algorithm]
	Consider a univariate stream $y_1, y_2, \dots$ with both unknown mean $\mu$ and variance $\sigma^2$. Let $\lambda \in [0, 1]$ denote the learning rate and $L \in \mathbb{R^{\ast}_{+}}$ a control parameter for the algorithm's sensitivity. Let $Z_1 = y_1$ and, for $t \geq 2$,
	\begin{equation*}
		Z_t \vcentcolon = (1 - \lambda) Z_{t-1} + \lambda y_t.
	\end{equation*}
	Let $\mu_t$ and $\sigma_t^2$ denote sequential estimates of the input mean and variance respectively computed as in Lemma \ref{prop:efficient_update}. The standard deviation of $Z_t$ is computed as
	\begin{equation*}
		\sigma_{Z_t} =  \sigma_t \sqrt{\frac{\lambda}{2-\lambda} \left( 1 - (1 - \lambda)^{2t} \right)},
	\end{equation*}
	and a change is detected when either $Z_t > \mu_t + L \sigma_{Z_t}$ or $Z_t < \mu_t - L \sigma_{Z_t}$.
	\label{def:adaptive_ewma}
\end{definition}

\begin{definition}[Average run lengths \cite{cusum}]
	Consider a stream with a changepoint located at time $\tau$, and assume the stream is monitored by an online changepoint detection algorithm which infers the first changepoint of the sequence, denoted as $\widehat{\tau}$. Then, the Average Run Length 0 ($\text{ARL}_0$) is computed as the average number of observations until a changepoint is detected, when the algorithm is run over a sequence with no changepoint, i.e.
	\begin{equation*}
		\text{ARL}_0 \vcentcolon = \mathbb{E}\left[ \widehat{\tau} \mid \tau \to \infty \right]. 
	\end{equation*}
	The Average Run Length 1 ($\text{ARL}_1$) is computed as the average number of observations between a changepoint location and it is being detected, i.e.
	\begin{equation*}
		\text{ARL}_1 \vcentcolon = \mathbb{E}\left[ \widehat{\tau} - \tau \mid \widehat{\tau} \geq \tau \right]. 
	\end{equation*}
	Hence, an effective algorithm should aim to maximise the $\text{ARL}_0$ while minimising the $\text{ARL}_1$.
\end{definition}

In practice, for a given algorithm, we estimate the $\text{ARL}_1$ by running simulations on multiple data streams with a changepoint and obtaining for each sequence the detection delay, conditioning on correct detections and therefore excluding false detections from the computation. Thus, the $\text{ARL}_1$ is computed by averaging these delay values, and the corresponding standard deviation $\text{SDRL}_1$ (standard deviation run length) is obtained by taking the sample standard deviation of these run lengths. Similarly, we estimate the $\text{ARL}_0$ by running a given algorithm on multiple \textit{sufficiently long} sequences without changepoints and obtain the corresponding run lengths as the first time to false detection. Then, the $\text{ARL}_0$ is computed by averaging these run lengths and its corresponding $\text{SDRL}_0$ by estimating the standard deviation of the run lengths. In practice, we may need to truncate the time series length for some algorithms with parameters leading to low sensitivity (see Table \ref{tab:simulations_sequence_samples}).

\begin{definition}[Precision, recall, and $F_1$-Score \cite{pelt, cpd_evaluation}] \label{def:f1}
	Consider a sequence with $n$ changepoints denoted as $\tau_1, \tau_2, \dots, \tau_n$ and $m$ detected changepoints $\widehat{\tau}_1, \widehat{\tau}_2, \dots, \widehat{\tau}_m$. The set of correctly detected changepoints, or set of true positives, is defined as 
	\begin{equation*}
		\mathcal{TP} \vcentcolon = \left\{ \tau_i \mid \exists ~ \widehat{\tau}_j : \mu_l \leq \widehat{\tau}_j - \tau_i \leq \mu_r \right\},
	\end{equation*}
	where $\mu_l \geq 0$ is a left margin and $\mu_r \geq 0$ is a right margin, defining the detection acceptance region around the true changepoint. 
	
	As in classical binary classification settings, the performance of a changepoint detection algorithm may be expressed in terms of precision $P$, defined as the ratio of correctly detected changepoints over the number of detected changepoints, and recall $R$, defined as the ratio of correctly detected changepoints over the number of true changepoints, i.e.
	\begin{equation*}
		P \vcentcolon = \frac{\operatorname{card}(\mathcal{TP})}{m} \text{ and } R \vcentcolon = \frac{\operatorname{card}(\mathcal{TP})}{n}.
	\end{equation*}
	The corresponding $F_1$-Score is defined as
	\begin{equation*}
		F_1 \vcentcolon = \frac{2 P R}{P + R}
	\end{equation*}
\end{definition}

\begin{remark}
	The left margin accomodates the detection of a changepoint preceeding the closest true annotated changepoint, particularly to account for fluctuating annotations decisions in real-world data. The right margin extends the acceptance region on the right, defining the acceptable range within which a detected change can be reasonably associated with the most recent true change. 
\end{remark}

\begin{definition}[Covering \cite{covering, cpd_evaluation}]
	Consider a sequence of length $T$ with true changepoints defined by the ordered set $\mathcal{T} = \{ \tau_1, \tau_2, \dots, \tau_n \}$ with $\tau_i \in [1, T]$ for $i \in \{ 1, 2, \dots, n \}$ and $\tau_i < \tau_j$ for $i < j$. Then, $\mathcal{T}$ implies a partition $\mathcal{G}$ of the interval $[1, T]$ into disjoint sets $\mathcal{A}_j$, where $\mathcal{A}_j$ is the segment from $\tau_{j-1}$ to $\tau_j - 1$, for $j \in \{ 1, 2, \dots, n+1 \}$, with convention $\tau_0 = 1$ and $\tau_{n+1} = T+1$. 
	
	For two sets $\mathcal{A}, \mathcal{A}^{\prime} \subseteq [1, T]$, the Jaccard index also known as Intersection over Union, is given by
	\begin{equation*}
		J(\mathcal{A}, \mathcal{A}^{\prime}) \vcentcolon = \frac{\operatorname{card}(\mathcal{A} \cap \mathcal{A}^{\prime})}{\operatorname{card}(\mathcal{A} \cup \mathcal{A}^{\prime})}
	\end{equation*} 
	For a sequence of length $T$, given a ground truth partition $\mathcal{G}$ and a predicted partition $\mathcal{G}^{\prime}$, the covering metric is defined as 
	\begin{equation*}
		\mathcal{C}(\mathcal{G}, \mathcal{G}^{\prime}) \vcentcolon = \frac{1}{T} \sum_{\mathcal{A} \in \mathcal{G}} \operatorname{card}(\mathcal{A}) \cdot \max_{\mathcal{A}^{\prime} \in \mathcal{G}^{\prime}} J(\mathcal{A}, \mathcal{A}^{\prime}). 
	\end{equation*}
\end{definition}

We rely on coding implementations of $F_1$-Score and covering provided by the Alan Turing Institute GitHub repository \href{https://github.com/alan-turing-institute/TCPDBench}{\small \texttt{https://github.com/alan-turing-institute/TCPDBench}} and detailed in \cite{cpd_evaluation}.

\newpage
\section{Algorithms pseudo-code}
\label{sec:pseudo_code}

\subsection{Dynamic mode decomposition}
\label{sec:pseudo_code_dmd}

We provide in Algorithm \ref{alg:dmd} a pseudo-code for the dynamic mode decomposition algorithm as introduced by \cite{dmdschmid} based on singular value decomposition. 

\begin{algorithm}[H]
	\caption{$\{ \mathbf{\widehat{x}}, \mathbf{\Phi}, \mathbf{\Omega} \} \gets \texttt{DMD}(\mathbf{x} \mid r)$}\label{alg:dmd}
	
	\nonl \bluecode{/* Dynamic Mode Decomposition */} 
	
	\KwIn{Snapshots $\mathbf{x} = [ \mathbf{x}_1, \mathbf{x}_2, \dots, \mathbf{x}_m ] \in \mathbb{R}^{p \times m}$} 
	\KwOut{Low-rank reconstruction $\mathbf{\widehat{x}} \in \mathbb{R}^{p \times m}$, modes $\mathbf{\Phi} \in \mathbb{C}^{p \times r}$, dynamics $\mathbf{\Omega} \in \mathbb{C}^{r \times r}$}
	\KwParam{SVD rank $r \leq \min{(p, m)}$}
	
	\bluecode{/* Lagged matrices (Equation \eqref{eq:snapshots_lagged}) */} \\
	$\mathbf{X} \gets [ \mathbf{x}_1, \mathbf{x}_2, \dots, \mathbf{x}_{m-1} ]$ \\
	$\mathbf{Y} \gets [ \mathbf{x}_2, \mathbf{x}_3, \dots, \mathbf{x}_{m} ]$ \\
	\bluecode{/* Singular Value Decomposition */} \\
	$\mathbf{U}, \mathbf{\Sigma}, \mathbf{V} \gets \texttt{SVD}(\mathbf{X})$ where  $\mathbf{X} \simeq \mathbf{U} \mathbf{\Sigma} \mathbf{V}^{\ast}$, $\mathbf{U} \in \mathbb{C}^{p \times r}$, $\mathbf{\Sigma} \in \mathbb{C}^{r \times r}$ and $\mathbf{V} \in \mathbb{C}^{(m - 1) \times r}$\\
	\bluecode{/* Operator projected onto principal modes */} \\
	$\mathbf{\Tilde{A}}_r \gets \mathbf{U}^{\ast} \mathbf{Y} \mathbf{V} \mathbf{\Sigma}^{-1}$ \\
	\bluecode{/* Eigendecomposition */} \\
	$ \mathbf{\Lambda}, \mathbf{W} \gets \texttt{eig}(\mathbf{\Tilde{A}}_r)$ where $\mathbf{\Tilde{A}}_r \mathbf{W} = \mathbf{W} \mathbf{\Lambda}$, $\mathbf{W} \in \mathbb{C}^{r \times r}$ and $\mathbf{\Lambda} = \diag(\lambda_1, \lambda_2, ..., \lambda_r) \in \mathbb{C}^{r \times r}$ \\
	\bluecode{/* Dynamics, modes and amplitudes */} \\
	$\mathbf{\Omega} \gets  \diag(\omega_1, \omega_2, ..., \omega_r)$ where $\omega_j = \log{(\lambda_j)} / \Delta t$ for $j \in \{ 1, 2, \dots, r \}$ \\
	$\mathbf{\Phi} \gets \mathbf{Y} \mathbf{V} \mathbf{\Sigma}^{-1} \mathbf{W}$ \\
	$\mathbf{b} \gets \mathbf{\Phi}^{\dagger} \mathbf{x}_1$ \\
	\bluecode{/* Reconstruction */} \\
	\For{$k \in \{ 1, 2, \dots, m\}$}{
		$\mathbf{\widehat{x}}_k \gets \mathbf{\Phi} \exp{(\mathbf{\Omega} (k-1) \Delta t)} \mathbf{b}$
	}
	$\mathbf{\widehat{x}} \gets [ \widehat{\mathbf{x}}_1, \widehat{\mathbf{x}}_2, \dots, \widehat{\mathbf{x}}_m ]$ \\
	\Return{$\{ \mathbf{\widehat{x}}, \mathbf{\Phi}, \mathbf{\Omega} \}$}
	
\end{algorithm}

\paragraph{Dynamic mode decomposition algorithm.}

Assume one monitors a dynamical system and collects $m$ observations $\mathbf{x}_1, \mathbf{x}_2, \dots, \mathbf{x}_m \in \mathbb{R}^p$. Observations are compiled in two distinct lagged data sets $\mathbf{X}, \mathbf{Y} \in \mathbb{R}^{p \times (m-1)}$
\begin{equation}
	\mathbf{X} = \left[ \begin{array}{cccc}
		\mathbf{x}_1 & \mathbf{x}_2 & \cdots & \mathbf{x}_{m-1} \\
	\end{array} \right], ~ \mathbf{Y} = \left[ \begin{array}{cccc}
		\mathbf{x}_2 & \mathbf{x}_3 & \cdots & \mathbf{x}_m \\
	\end{array} \right].
	\label{eq:snapshots_lagged}
\end{equation}
Then, DMD consists in finding the \emph{best-fit} linear operator $\mathbf{A} \in \mathbb{R}^{p \times p}$ such that $\mathbf{Y} = \mathbf{A} \mathbf{X}$. In fact, the optimal solution is defined as $\mathbf{\Tilde{A}} = \mathbf{Y} \mathbf{X}^{\dagger}$ where $\mathbf{X}^{\dagger}$ denotes the Moore-Penrose pseudo-inverse of $\mathbf{X}$. This regression problem may be solved using classical least-squares optimisation. However, DMD rather estimates a low-rank eigendecomposition of the linear operator $\mathbf{\Tilde{A}}$, since the complete estimation of $\mathbf{\Tilde{A}} \in \mathbb{R}^{p \times p}$ may become intractable for values of $p \gg 1$ \cite{dmdschmid}. Hence, DMD solves the following optimisation problem
\begin{equation*}
	\mathbf{\Tilde{A}}_r = \argmin_{\mathbf{A} \in \mathbb{R}^{p \times p}} \lVert \mathbf{Y} - \mathbf{A} \mathbf{X} \rVert^2_{F} \text{ with } \rank{\mathbf{A}} \leq r,
\end{equation*}
using singular value decomposition, which results in extracting and learning dynamics $\mathbf{\Omega} \in \mathbb{C}^{r \times r}$ and spatial modes $\mathbf{\Phi} \in \mathbb{C}^{p \times r}$ directly from the eigendecomposition of the best-fit operator $\mathbf{\Tilde{A}}_r$. A low-rank reconstruction of the input is computed as
\begin{equation*}
	\mathbf{\widehat{x}}_{t+1} = \boldsymbol{\Phi} \exp{(\mathbf{\Omega} t)} \boldsymbol{\Phi}^{\dagger} \mathbf{x}_1, ~ \text{for } t \in \{ 0, 1, \dots, m-1 \}.
\end{equation*}
The DMD algorithm is described in Algorithm \ref{alg:dmd}. Note that DMD extracts pairs of modes and their associated dynamics: modes can be interpreted as coherent structures or spatial \textit{shapes} of the signal, and dynamics encode the temporal evolution of the related mode, with real part corresponding to their growth and imaginary part to their oscillating frequency. Computational complexity of DMD algorithm is described in Lemma \ref{prop:dmd_complexity} in Appendix \ref{sec:theoretical_analysis_complexity}.

\newpage
\subsection{Unsupervised hyperparameter selection}
\label{sec:appendix_model_selection}

We provide in Algorithm \ref{alg:model_selection} a pseudo-code for our proposed hyperparameter selection on burn-in period. 

\begin{algorithm}[H]
	\caption{$\theta^{\ast} = ( w, d, r ) \gets \texttt{HyperparameterSelection}(x_1, x_2, \dots, x_{T_0} \mid T_0)$}\label{alg:model_selection}
	
	\nonl \bluecode{/* Hyperparameter selection by average reconstruction error minimisation */}
	
	\KwIn{Sequence of observations on burn-in period $x_1, x_2, \dots, x_{T_0} \in \mathbb{R}^p$}
	
	\KwOut{Optimal parameters $\theta^{\ast} = ( w, d, r )$}
	
	\KwParam{Burn-in period $T_0 \in \mathbb{N}$}
	
	\bluecode{/* Grid of parameters generation and initalisation */} \\
	$\Theta \gets \texttt{grid}(T_0, p)$ \\
	$\varepsilon^{\ast} \gets \infty$ \\
	\bluecode{/* Explore possible models (in parallel) */} \\
	\For{$( w, d, r ) \in \Theta$}{
		$\varepsilon_t \gets 0$ for $t \in \{ w, w+1, \dots, T_0 \}$ \\
		\For{$t \in \{ w, w+1, \dots, T_0 \}$}{
			\bluecode{/* Sequential DMD on Hankel batches and reconstruction */} \\
			$X_t \gets [ x_{t-w+1}, x_{t-w+2}, \dots, x_t ]$ \\
			$\mathcal{X}_t \gets \texttt{Hankel}(X_t \mid d)$ \\
			$\{ \widehat{\mathcal{X}}_t, \mathbf{\Phi}_t, \mathbf{\Omega}_t \} \gets \texttt{DMD}(\mathcal{X}_t \mid r)$ \\
			$\widehat{X}_t \gets \texttt{Unroll}(\widehat{\mathcal{X}}_t \mid d)$ \\
			\bluecode{/* Reconstruction error */} \\
			$\varepsilon_t \gets \lVert X_t - \widehat{X}_t \rVert^2_F$ \\
		}
		\bluecode{/* Average reconstruction error */} \\
		$\varepsilon \gets \frac{1}{T_0 - w + 1} \sum_{t=w}^{T_0} \varepsilon_t$ \\
		\If{$\varepsilon < \varepsilon^{\ast}$}{
			\bluecode{/* Update best parameters and error */} \\
			$\varepsilon^{\ast} \gets \varepsilon$ \\
			$\theta^{\ast} \gets ( w, d, r )$
		}
	}
	\Return{$\theta^{\ast} = ( w, d, r )$}
	
\end{algorithm}

\paragraph{Classical hyperparameter selection methods.}
On the one hand, classical unsupervised selection techniques involve the definition of a criterion to optimise, such as AIC \cite{aic} or BIC \cite{bic}. Such criteria are often parametric: they compute a likelihood function and thus require distribution assumptions. On the other hand, some hyperparameter selection methods rely on heuristics, such as the elbow \cite{elbow} or silhouette \cite{silhouette} methods in clustering, but are not strict objective functions to optimise. Moreover, in contrast to iterative and sequential approaches such as Bayesian optimisation, our hyperparameter selection necessitates a streaming paradigm, where operations may be executed concurrently on different threads as suggested in \cite{bocpdms}, to enable real-time processing. 

\paragraph{Approximation vs. reconstruction.} Note that this hyperparameter selection procedure does not simply choose the most complex model on the grid of tested parameters: while at fixed window length $w$ and auto-regressive order $d$, the approximation error of DMD would decrease when increasing the rank of the decomposition $r$, it may not be the case for the average reconstruction error. This is the result of DMD reconstruction being based on learned modes and dynamics rather than simply multiplying the first lagged matrix by the DMD operator. Indeed, the approximation error may be seen as a one-step ahead forecasting optimisation, where each observed snapshot is used to forecast the next one. In the contrary, the reconstruction error could be considered as a multiple steps ahead forecasting, where only the first snapshot is used to recurrently predict the next ones from the learned modes and dynamics. The optimal window length $w$, the auto-regressive order $d$ and the rank $r$ are typically determined by the underlying dynamical system properties.

\paragraph{Practical implementation.} In practice, not all triplets of $\Theta$ can be explored and truncation of the parameters space is necessary due to computational constraints (e.g., the number of threads). We further discuss the grid generation process used in simulations and experiments in Appendix \ref{sec:experiments_details}. 


\newpage
\subsection{Multiple changepoints detection}
\label{sec:appendix_cpdmd}

We provide in Algorithm \ref{alg:multiple_changepoints_with_model_selection} a pseudo-code for our proposed method in the context of multiple changepoints, with parameters being auto-adjusted to the post-change distribution after each detected changepoint. 

\begin{algorithm}[H]
	\caption{$\boldsymbol{\widehat{\tau}} \gets \texttt{CPDMD}(x_1, x_2, \dots, x_T  \mid T_0, \lambda, L)$}\label{alg:multiple_changepoints_with_model_selection}
	\nonl \bluecode{/* Multiple changepoint detection via DMD with hyperparameter selection */}
	
	\KwIn{Sequence of observations $x_1, x_2, \dots, x_T \in \mathbb{R}^p$}
	
	\KwOut{List of detected changepoints in the sequence $\boldsymbol{\widehat{\tau}}$}
	
	\KwParam{Burn-in period $T_0 \in \mathbb{N}$, adaptive EWMA learning rate $\lambda \in [0, 1]$ and limit $L \in \mathbb{R}^{\ast}_{+}$}
	
	$\widehat{\tau} \gets 1$ \\
	$\boldsymbol{\widehat{\tau}} \gets \varnothing$ \\
	\While{$\widehat{\tau} \neq \mathtt{None}$}{
		\bluecode{/* Check that there are enough remaining samples */} \\
		\eIf{$T - \widehat{\tau} \geq T_0$}{
			\bluecode{/* Hyperparameter selection on new data distribution (Algorithm \ref{alg:model_selection}) */} \\ 
			$\theta^{\ast} = ( w, d, r ) \gets \texttt{HyperparameterSelection}(x_{\tau}, x_{\tau+1}, \dots, x_{T_0} \mid T_0)$ \\
			\bluecode{/* Find the next changepoint in the remaining sequence (Algorithm \ref{alg:single_changepoint}) */} \\ 
			$\widehat{\tau} \gets \texttt{SingleCP}(x_{\tau}, x_{\tau+1}, \dots, x_T \mid T_0, w, d, r, \lambda, L)$
		}{
			$\widehat{\tau} \gets \mathtt{None}$
		}
		\bluecode{/* Add the new detected changepoint */} \\
		$\boldsymbol{\widehat{\tau}} \gets \boldsymbol{\widehat{\tau}} \cup \{\widehat{\tau}\}$
	}
	\Return{$\boldsymbol{\widehat{\tau}}$} 
	
\end{algorithm}

\newpage
\section{Theoretical analysis}
\label{sec:theoretical_analysis}

\subsection{Spectral properties of the DMD operator}
\label{sec:theoretical_analysis_eig}

\subsubsection{Windowed incremental algorithms}

\begin{lemma}[Windowed incremental SVD algorithm, Proposition 2 from \citep{incremental_svd_dmd}]
	\label{th:update_svd}
	Consider a windowed dataset 
	\begin{equation*}
		\mathbf{X}_k = \left[ \begin{array}{cccc}
			\mathbf{x}_{k-w+1} & \mathbf{x}_{k-w+2} & \cdots & \mathbf{x}_{k}
		\end{array} \right] \in \mathbb{R}^{n \times w},
	\end{equation*}
	where $n$ is the dimension of the measurements and $w$ is the window length. Consider the other windowed dataset corresponding to the consecutive window
	\begin{equation*}
		\mathbf{X}_{k+1} = \left[ \begin{array}{cccc}
			\mathbf{x}_{k-w+2} & \mathbf{x}_{k-w+3} & \cdots & \mathbf{x}_{k+1}
		\end{array} \right] \in \mathbb{R}^{n \times w}. 
	\end{equation*}
	Assume the SVD of $\mathbf{X}_k$ is known: $\mathbf{X}_k =  \mathbf{U}_{\mathbf{X}_k} \mathbf{\Sigma}_{\mathbf{X}_k} \mathbf{V}_{\mathbf{X}_k}^{\ast}$. Then, the SVD of $\mathbf{X}_{k+1}$ can be computed from the SVD of $\mathbf{X}_k$ as follows. 
	
	Define the new dataset $\mathbf{\acute{X}_k}$ such that the first column of $\mathbf{X}_k$ is eliminated
	\begin{equation*}
		\mathbf{\acute{X}_k} =  \left[ \begin{array}{cccc}
			\mathbf{x}_{k-w+2} & \mathbf{x}_{k-w+3} & \cdots & \mathbf{x}_{k}
		\end{array} \right] \in \mathbb{R}^{n \times (w-1)}.
	\end{equation*}
	Then, the SVD of $\mathbf{\acute{X}_k}$ is given by $\mathbf{\acute{X}_k} = \mathbf{U}_{\mathbf{\acute{X}_k}} \mathbf{\Sigma}_{\mathbf{\acute{X}_k}} \mathbf{V}^{\ast}_{\mathbf{\acute{X}_k}}$ such that
	\begin{equation*}
		\mathbf{U}_{\mathbf{\acute{X}_k}} =  \mathbf{U}_{\mathbf{X}_k}  \mathbf{U}_{\mathbf{\acute{S}}_k}, ~ \mathbf{\Sigma}_{\mathbf{\acute{X}_k}} = \mathbf{\Sigma}_{\mathbf{\acute{S}_k}} \text{ and } \mathbf{V}_{\mathbf{\acute{X}_k}} =  \mathbf{V}_{\mathbf{X}_{k, 2}}  \mathbf{V}_{\mathbf{\acute{S}_k}},
	\end{equation*}
	where $\mathbf{V}_{\mathbf{X}_{k, 2}}$ is the submatrix of $\mathbf{V}_{\mathbf{X}_{k}}$ obtained after removing its first row, i.e. $\mathbf{V}_{\mathbf{X}_k} = \left[ \begin{array}{c}
		\mathbf{v}_{\mathbf{X}_{k,1}} \\
		\mathbf{V}_{\mathbf{X}_{k, 2}}
	\end{array} \right]$, and $\mathbf{U}_{\mathbf{\acute{S}}_k}$, $\mathbf{\Sigma}_{\mathbf{\acute{S}_k}}$ and $\mathbf{V}_{\mathbf{\acute{S}_k}}$ are obtained from the SVD of $\mathbf{\acute{S}}_k = \mathbf{U}_{\mathbf{\acute{S}}_k} \mathbf{\Sigma}_{\mathbf{\acute{S}_k}} \mathbf{V}_{\mathbf{\acute{S}_k}}^{\ast}$ where 
	\begin{equation*}
		\mathbf{\acute{S}}_k = \mathbf{\Sigma}_{\mathbf{X}_k} - \mathbf{U}_{\mathbf{X}_k}^{\ast} \mathbf{x}_{k-w+1} \mathbf{z}_1^{\intercal} \mathbf{V}_{\mathbf{X}_k},
	\end{equation*}
	and $\mathbf{z}_1 = \left[\begin{array}{cccc}
		1 & 0 & \cdots & 0
	\end{array}\right]^{\intercal} \in \mathbb{R}^w$. Therefore, the SVD of $\mathbf{X}_{k+1}$ is given by $\mathbf{X}_{k+1} = \mathbf{U}_{\mathbf{X}_{k+1}} \mathbf{\Sigma}_{\mathbf{X}_{k+1}} \mathbf{V}^{\ast}_{\mathbf{X}_{k+1}}$ where
	\begin{equation}
		\begin{aligned}
			\mathbf{U}_{\mathbf{X}_{k+1}} &= \mathbf{U}_{\mathbf{\acute{X}_k}} \mathbf{U}_{\mathbf{\hat{S}}_k} = \mathbf{U}_{\mathbf{X}_k} \mathbf{U}_{\mathbf{\acute{S}}_k} \mathbf{U}_{\mathbf{\hat{S}}_k}, \\
			\mathbf{\Sigma}_{\mathbf{X}_{k+1}} &= \mathbf{\Sigma}_{\mathbf{\hat{S}}_k}, \\
			\mathbf{V}_{\mathbf{X}_{k+1}} &= \left[ \begin{array}{c}
				\mathbf{V}_{\mathbf{\acute{X}_k}} \mathbf{V}_{\mathbf{\hat{S}}_{k, 1}} \\
				\mathbf{v}_{\mathbf{\hat{S}}_{k, 2}}
			\end{array} \right] = \left[ \begin{array}{c}
				\mathbf{V}_{\mathbf{X}_{k, 2}} \mathbf{V}_{\mathbf{\acute{S}}_k} \mathbf{V}_{\mathbf{\hat{S}}_{k, 1}} \\
				\mathbf{v}_{\mathbf{\hat{S}}_{k, 2}}
			\end{array} \right].
		\end{aligned}
		\label{eq:update_svd}
	\end{equation}
	where $\mathbf{\hat{S}}_k = \left[ \begin{array}{cc}
		\mathbf{\Sigma}_{\mathbf{\acute{X}}_k} & \mathbf{U}_{\mathbf{\acute{X}}_k}^{\ast} \mathbf{x}_{k+1}
	\end{array} \right] = \mathbf{U}_{\mathbf{\hat{S}}_k} \mathbf{\Sigma}_{\mathbf{\hat{S}}_k} \mathbf{V}^{\ast}_{\mathbf{\hat{S}}_k}$.
\end{lemma}

\begin{lemma}[Windowed incremental DMD algorithm, Theorem 2 from \cite{incremental_svd_dmd}]
	\label{th:update_dmd}
	Assume that we have access to the following windowed datasets at time $t_{k+1}$
	\begin{equation*}
		\mathbf{X}_k = \left[ \begin{array}{cccc}
			\mathbf{x}_{k-w+1} & \mathbf{x}_{k-w+2} & \cdots & \mathbf{x}_{k} 
		\end{array} \right], ~ \mathbf{Y}_k = \left[ \begin{array}{cccc}
			\mathbf{x}_{k-w+2} & \mathbf{x}_{k-w+3} & \cdots & \mathbf{x}_{k+1}
		\end{array} \right].
	\end{equation*}
	Suppose that at time $t_{k+1}$, the SVD of $\mathbf{X}_k = \mathbf{U}_{\mathbf{X}_k} \mathbf{\Sigma}_{\mathbf{X}_k} \mathbf{V}_{\mathbf{X}_k}^{\ast}$ is known, along with the corresponding DMD operator $\mathbf{A}_k$. 
	
	At time $t_{k+2}$, a new measurement $\mathbf{x}_{k+2}$ is observed, and one can construct the updated datasets
	\begin{equation*}
		\mathbf{X}_{k+1} = \left[ \begin{array}{cccc}
			\mathbf{x}_{k-w+2} & \mathbf{x}_{k-w+3} & \cdots & \mathbf{x}_{k+1} 
		\end{array} \right], ~ \mathbf{Y}_{k+1} = \left[ \begin{array}{cccc}
			\mathbf{x}_{k-w+3} & \mathbf{x}_{k-w+4} & \cdots & \mathbf{x}_{k+2} 
		\end{array} \right].
	\end{equation*}
	Then, at time $t_{k+2}$, the SVD of $\mathbf{X}_{k+1}$ can be incrementally computed from the known SVD of $\mathbf{X}_k$ using Lemma \ref{th:update_svd} and Equation~\eqref{eq:update_svd}. Thus, the DMD operator $\mathbf{A}_{k+1}$ is
	\begin{equation*}
		\mathbf{A}_{k+1} = \mathbf{A}_k + \left( \mathbf{x}_{k+2} - \mathbf{A}_k \mathbf{x}_{k+1} \right) \mathbf{v}_{\mathbf{\hat{S}}_{k, 2}} \mathbf{\Sigma}_{\mathbf{X}_{k+1}}^{-1} \mathbf{U}_{\mathbf{X}_{k+1}}^{\ast}.
	\end{equation*}
\end{lemma}

\subsubsection{Perturbation theory}

\begin{lemma}[Bauer-Fike Theorem \cite{bauer_fike}, Theorem 5.3 from \cite{approx_eig}]
	\label{th:perturbation_eigenvalues}
	Let $A \in \mathbb{C}^{n \times n}$ be a diagonalisable matrix and denote by $X = \left[ \mathbf{x}_1, \mathbf{x}_2, \dots, \mathbf{x}_n \right] \in \mathbb{C}^{n \times n}$ the matrix of its right eigenvectors, where $\mathbf{x}_k \in \mathbb{C}^n$ for $k\in~\{ 1, 2, \dots, n \}$, such that $D = X^{-1} A X = \diag (\lambda_1, \lambda_2, \dots, \lambda_n)$, $\lambda_i$ being eigenvalues of $A$, $i \in \{ 1, 2, \dots, n\}$. 
	
	Let $E \in \mathbb{C}^{n \times n}$ be a perturbation of $A$. Let $\mu$ be an eigenvalue of the matrix $A+E \in \mathbb{C}^{n \times n}$, then 
	\begin{equation*}
		\min_{\lambda \in \sigma(A)} \lvert \lambda - \mu \rvert \leq \kappa_p(X) \lVert E \rVert_p,
	\end{equation*}
	where $\lVert E \rVert_p$ is any matrix $p$-norm of $E$, and $\kappa_p(X) = \lVert X \rVert_p \lVert X^{-1} \rVert_p$ is called the condition number of the eigenvalue problem for matrix $A$.
\end{lemma}

\begin{lemma}[Property 5.5 from \cite{approx_eig}]
	\label{th:perturbation_eigenvectors}
	Let $A \in \mathbb{C}^{n \times n}$ be a diagonalisable matrix and denote by $X = \left[ \mathbf{x}_1, \mathbf{x}_2, \dots, \mathbf{x}_n \right] \in \mathbb{C}^{n \times n}$ the matrix of its right eigenvectors, where $\mathbf{x}_k \in \mathbb{C}^n$ for $k\in~\{ 1, 2, \dots, n \}$, such that $D = X^{-1} A X = \diag (\lambda_1, \lambda_2, \dots, \lambda_n)$, $\lambda_i$ being eigenvalues of $A$, $i \in \{ 1, 2, \dots, n\}$. 
	
	Let $E \in \mathbb{C}^{n \times n}$ be a perturbation of $A$ with $\lVert E \rVert_2 = 1$  and let $A(\varepsilon) = A + \varepsilon E$. Then, the eigenvectors $\mathbf{x}_k$ and $\mathbf{x}_k(\varepsilon)$ of matrices $A$ and $A(\varepsilon)$, with $\lVert \mathbf{x}_k(\varepsilon) \rVert_2 = \lVert \mathbf{x} \rVert_2 = 1$, for $k \in \{ 1, 2, \dots, n\}$ satisfy
	\begin{equation*}
		\lVert \mathbf{x}_k(\varepsilon) - \mathbf{x}_k \rVert_2 \leq \frac{\varepsilon}{\min_{j \neq k} \lvert \lambda_k - \lambda_j \rvert} + \mathcal{O}\left( \varepsilon^2 \right), ~ \forall k \in \{1, 2, \dots, n\}.
	\end{equation*}
\end{lemma}

\subsubsection{Main theoretical result}

We prove in the following section Theorem \ref{th:bound_eigendecomposition} which consists of the main theoretical result of our work. 

\begin{reptheorem}{th:bound_eigendecomposition}[Upper bound on perturbations of DMD operator eigenvalues and eigenvectors]
	Consider a $p$-dimensional stream of observations monitored via Algorithm \ref{alg:single_changepoint}. Let $\mathbf{A} \in \mathbb{C}^{pd \times pd}$ be the DMD operator at time $t$ and let $\mathbf{\Tilde{A}} \in \mathbb{C}^{pd \times pd}$ denote the DMD operator at time $t+1$. Then, there exists a perturbation matrix $\mathbf{E} \in \mathbb{C}^{pd \times pd}$ such that $\mathbf{\Tilde{A}} = \mathbf{A} + \mathbf{E}$, where the closed-form of $\mathbf{E}$ is detailed in Appendix \ref{sec:theoretical_analysis_eig}. 
	Furthermore, assuming that the DMD operator $\mathbf{A}$ is diagonalisable, then the following results hold:
	\begin{enumerate}[label=(\roman*)]
		\item Eigenvalues (dynamics): let $\Tilde{\lambda} \in \sigma(\mathbf{\Tilde{A}})$, i.e. $\Tilde{\lambda}$ is an eigenvalue of $\mathbf{\Tilde{A}}$, then 
		\begin{equation*}
			\min_{\lambda \in \sigma(\mathbf{A})} \lvert \Tilde{\lambda} - \lambda \rvert \leq \kappa_2(\mathbf{\Phi}) \lVert \mathbf{E} \rVert_2,
		\end{equation*} 
		where $\lVert \mathbf{E} \rVert_2$ denotes the matrix $2$-norm of $\mathbf{E}$, $\kappa_2(\mathbf{\Phi}) = \lVert \mathbf{\Phi} \rVert_2 \lVert \mathbf{\Phi}^{-1} \rVert_2$ is the condition number of $\mathbf{\Phi} \in \mathbb{C}^{pd \times pd}$ and $\mathbf{\Phi}$ is the eigenvectors matrix of $\mathbf{A}$. 
		\item Eigenvectors (modes): let $\boldsymbol{\Tilde{\phi}}_k \in \mathbb{C}^{pd}$ and $\boldsymbol{\phi}_k \in \mathbb{C}^{pd}$ denote normalised eigenvectors of $\mathbf{\Tilde{A}}$ and $\mathbf{A}$ respectively for $k \in \{ 1, 2, \dots, pd\}$, then
		\begin{equation*}
			\lVert \boldsymbol{\Tilde{\phi}}_k - \boldsymbol{\phi}_k \rVert_2 \leq \frac{\lVert \mathbf{E} \rVert_2}{\min_{j \neq k} \lvert \lambda_k - \lambda_j \rvert} + \mathcal{O}\left(\lVert \mathbf{E} \rVert_2^2\right), \forall k \in \{1, 2, \dots, pd\}.
		\end{equation*}
	\end{enumerate}
\end{reptheorem}

\begin{proof}
	At time $t$, the new datum $x_t \in \mathbb{R}^{p}$ is observed and one can construct the Hankel batch $\mathcal{X}_t \in \mathbb{R}^{pd \times (w-d+1)}$, with given window length $w \in \mathbb{N}$ and auto-regressive order $d \in \{ 1, 2, \dots, w\}$, such that
	\begin{equation*}
		\mathcal{X}_t = \begin{bNiceMatrix}[margin, first-row, last-row]
			\rule[-8pt]{0pt}{8pt}
			\Block{2-5}{} \\
			\Block[fill=blue1!40,rounded-corners]{4-5}{}
			x_{t-w+1}^{(1)} & 	x_{t-w+2}^{(1)} & \cdots & 	x_{t-d}^{(1)} & x_{t-d+1}^{(1)} \\
			x_{t-w+2}^{(1)} & 	x_{t-w+3}^{(1)} & \cdots & 	x_{t-d+1}^{(1)} & x_{t-d+2}^{(1)} \\
			\vdots & \vdots & \ddots & \vdots & \vdots \\
			x_{t-w+d}^{(1)} & 	x_{t-w+d+1}^{(1)} & \cdots & 	x_{t-1}^{(1)} & x_{t}^{(1)} \\
			\Block[fill=blue2!40,rounded-corners]{4-5}{}
			x_{t-w+1}^{(2)} & 	x_{t-w+2}^{(2)} & \cdots & 	x_{t-d}^{(2)} & x_{t-d+1}^{(2)} \\
			x_{t-w+2}^{(2)} & 	x_{t-w+3}^{(2)} & \cdots & 	x_{t-d+1}^{(2)} & x_{t-d+2}^{(2)} \\
			\vdots & \vdots & \ddots & \vdots & \vdots \\
			x_{t-w+d}^{(2)} & 	x_{t-w+d+1}^{(2)} & \cdots & 	x_{t-1}^{(2)} & x_{t}^{(2)} \\
			\vdots & \vdots & \vdots & \vdots & \vdots \\
			\Block[fill=blue3!40,rounded-corners]{4-5}{}
			x_{t-w+1}^{(p)} & 	x_{t-w+2}^{(p)} & \cdots & 	x_{t-d}^{(p)} & x_{t-d+1}^{(p)} \\
			x_{t-w+2}^{(p)} & 	x_{t-w+3}^{(p)} & \cdots & 	x_{t-d+1}^{(p)} & x_{t-d+2}^{(p)} \\
			\vdots & \vdots & \ddots & \vdots & \vdots \\
			x_{t-w+d}^{(p)} & 	x_{t-w+d+1}^{(p)} & \cdots & 	x_{t-1}^{(p)} & x_{t}^{(p)} \\
			\Block{1-4}{\scriptstyle\mathbf{X}_t \in \mathbb{R}^{pd \times (w-d)}} & 
			\rule[8pt]{0pt}{8pt}
			\CodeAfter
			\UnderBrace[yshift=1mm,shorten]{last-1}{last-4}{}
			\OverBrace[yshift=1mm,shorten]{1-2}{1-last}{\scriptstyle\mathbf{Y}_t \in \mathbb{R}^{pd \times (w-d)}} 
		\end{bNiceMatrix} \in \mathbb{R}^{pd \times (w-d+1)}.
	\end{equation*}
	Recall that our proposed method desribed in Algorithm \ref{alg:single_changepoint} consists in applying DMD to snapshots of $\mathcal{X}_t$. Rigorously, we may define the snapshots of $\mathcal{X}_t$ as its $(w-d+1)$ column vectors, i.e. for $k \in \{ t-w+d, t-w+d+1, \dots, t \}$
	\begin{equation}
		\mathbf{x}_k^{\intercal} = \begin{bNiceMatrix}[margin]
			\Block[fill=blue1!40,rounded-corners]{1-3}{}
			x_{k-d+1}^{(1)} & \cdots & x_{k}^{(1)}  & \Block[fill=blue2!40,rounded-corners]{1-3}{} x_{k-d+1}^{(2)} & \cdots & x_{k}^{(2)} & \cdots & \Block[fill=blue3!40,rounded-corners]{1-3}{} x_{k-d+1}^{(p)} & \cdots & x_{k}^{(p)} 
		\end{bNiceMatrix} \in \mathbb{R}^{pd}.
		\label{eq:snapshot_vectors}
	\end{equation}
	This leads to defining the lagged snapshots matrices $\mathbf{X}_t~\in~\mathbb{R}^{pd \times (w-d)}$ and $\mathbf{Y}_t \in \mathbb{R}^{pd \times (w-d)}$ such that
	\begin{equation*}
		\begin{aligned}
			\mathbf{X}_t &= \begin{bNiceMatrix}
				\mathbf{x}_{t-w+d} & \mathbf{x}_{t-w+d+1} & \cdots & \mathbf{x}_{t-1} 
			\end{bNiceMatrix} \in \mathbb{R}^{pd \times (w-d)}, \\
			\mathbf{Y}_t &= \begin{bNiceMatrix}
				\mathbf{x}_{t-w+d+1} & \mathbf{x}_{t-w+d+2} & \cdots & \mathbf{x}_{t} 
			\end{bNiceMatrix} \in \mathbb{R}^{pd \times (w-d)}.
		\end{aligned}
	\end{equation*}
	Let $\mathbf{A} \in \mathbb{C}^{pd \times pd}$ denote the DMD operator at time $t$ extracted from snapshots matrices $\mathbf{X}_t$ and $\mathbf{Y}_t$. 
	
	Similarly, at time $t+1$, the new datum $x_{t+1} \in \mathbb{R}^p$ is observed, one can construct $\mathcal{X}_{t+1}~\in~\mathbb{R}^{pd \times (w-d+1)}$ and using the same notations for the snapshots vectors as in Equation \eqref{eq:snapshot_vectors}, one may obtain 
	\begin{equation*}
		\begin{aligned}
			\mathbf{X}_{t+1} &= \begin{bNiceMatrix}
				\mathbf{x}_{t-w+d+1} & \mathbf{x}_{t-w+d+2} & \cdots & \mathbf{x}_{t} 
			\end{bNiceMatrix} \in \mathbb{R}^{pd \times (w-d)}, \\
			\mathbf{Y}_{t+1} &= \begin{bNiceMatrix}
				\mathbf{x}_{t-w+d+2} & \mathbf{x}_{t-w+d+3} & \cdots & \mathbf{x}_{t+1} 
			\end{bNiceMatrix} \in \mathbb{R}^{pd \times (w-d)},
		\end{aligned}
	\end{equation*}
	along with the corresponding DMD operator at time $t+1$, denoted by $\mathbf{\Tilde{A}} \in \mathbb{C}^{pd \times pd}$. 
	
	Applying Lemma \ref{th:update_svd} and Lemma \ref{th:update_dmd} leads to
	\begin{equation*}
		\begin{aligned}
			\mathbf{\Tilde{A}} &= \mathbf{A} + \left( \mathbf{x}_{t+1} - \mathbf{A} \mathbf{x}_t \right) \mathbf{v}_{\mathbf{\hat{S}}_{t, 2}} \mathbf{\Sigma}^{-1}_{\mathbf{X}_{t+1}} \mathbf{U}_{\mathbf{X}_{t+1}}^{\ast} \\	
			&= \mathbf{A} + \mathbf{E},
		\end{aligned}
	\end{equation*}
	where $\mathbf{E} =  \left( \mathbf{x}_{t+1} - \mathbf{A} \mathbf{x}_t \right) \mathbf{v}_{\mathbf{\hat{S}}_{t, 2}} \mathbf{\Sigma}^{-1}_{\mathbf{X}_{t+1}} \mathbf{U}_{\mathbf{X}_{t+1}}^{\ast} \in \mathbb{C}^{pd \times pd}$. The matrices $\mathbf{\Sigma}^{-1}_{\mathbf{X}_{t+1}}$, $ \mathbf{U}_{\mathbf{X}_{t+1}}$ and $\mathbf{v}_{\mathbf{S}_{t, 2}}$ are obtained from the SVD of $\mathbf{X}_{t+1} =  \mathbf{U}_{\mathbf{X}_{t+1}} \mathbf{\Sigma}_{\mathbf{X}_{t+1}} \mathbf{V}_{\mathbf{X}_{t+1}}^{\ast}$ using Equation \eqref{eq:update_svd}. 
	
	Under the usual assumption of diagonalisable DMD operators \cite{dmd_theory_applications}, let $\mathbf{\Phi}~=~\left[ \boldsymbol{\phi}_1, \boldsymbol{\phi}_2, \dots, \boldsymbol{\phi}_{pd} \right]~\in~\mathbb{C}^{pd \times pd}$ denote the matrix of the right normalised eigenvectors of $\mathbf{A}$ where $\boldsymbol{\phi}_k \in \mathbb{C}^{pd}$ for $k \in \{ 1, 2, \dots, pd \}$, then Lemma \ref{th:perturbation_eigenvalues} for the matrix $2$-norm implies, for any eigenvalue $\Tilde{\lambda}$ of $\mathbf{\Tilde{A}}$,
	\begin{equation}
		\min_{\lambda \in \sigma(\mathbf{A})} \lvert \Tilde{\lambda} - \lambda \rvert \leq \kappa_2(\mathbf{\Phi}) \lVert \mathbf{E} \rVert_2,
	\end{equation}
	where $\sigma(\mathbf{A})$ denotes the spectrum of $\mathbf{A}$ and $\kappa_2(\mathbf{\Phi}) = \lVert \mathbf{\Phi} \rVert_2 \lVert \mathbf{\Phi}^{-1} \rVert_2$ is the condition number of $\mathbf{\Phi} \in \mathbb{C}^{pd \times pd}$. Moreover, Lemma \ref{th:perturbation_eigenvectors} implies that given the normalised eigenvectors $\boldsymbol{\phi}_k \in \mathbb{C}^{pd}$ and $\boldsymbol{\Tilde{\phi}}_k \in \mathbb{C}^{pd}$ of $\mathbf{A}$ and $\mathbf{\Tilde{A}}$ respectively, $k \in \{ 1, 2, \dots, pd\}$, then
	\begin{equation*}
		\lVert \boldsymbol{\Tilde{\phi}}_k - \boldsymbol{\phi}_k \rVert_2 \leq \frac{\lVert \mathbf{E} \rVert_2}{\min_{j \neq k} \lvert \lambda_k - \lambda_j \rvert} + \mathcal{O}\left(\lVert \mathbf{E} \rVert_2^2\right), \forall k \in \{1, 2, \dots, pd\}.
	\end{equation*}
\end{proof}

\clearpage

\subsection{Computational complexity}
\label{sec:theoretical_analysis_complexity}

\begin{lemma}[Efficient update of the mean and variance \cite{welford}]
	\label{prop:efficient_update}
	Consider a univariate time series $y_1, y_2, \dots$. Let $\mu_t$ denote the streaming sample mean and $\sigma_t^2$ the streaming sample variance of observations up to time $t$, i.e. $y_1, y_2, \dots, y_t$. Then, given initial conditions $\mu_1$ and $\sigma_1^2$, for $t \geq 2$
	\begin{equation*}
		\begin{aligned}
			\mu_t &= \frac{t-1}{t} \mu_{t-1} + \frac{y_t}{t}, \\
			\sigma_t^2 &= \frac{t-1}{t} \sigma_{t-1}^2 + \frac{(y_t - \mu_t)(y_t - \mu_{t-1})}{t}.
		\end{aligned}
	\end{equation*}
	Therefore, the sample mean and variance can be sequentially updated with $\mathcal{O}(1)$ time complexity sequentially without having to store all previous observed values.
\end{lemma}

\begin{lemma}[Computational complexity of DMD defined in Algorithm \ref{alg:dmd}]
	\label{prop:dmd_complexity}
	Given a $n$-dimensional snapshots matrix $\mathbf{x} = [\mathbf{x}_1, \mathbf{x}_2, \dots, \mathbf{x}_m] \in \mathbb{R}^{n \times m}$ and a fixed SVD rank $r \leq \min{(m, n)}$, Algorithm \ref{alg:dmd} extracts the modes $\mathbf{\Phi} \in \mathbb{C}^{n \times r}$, the dynamics $\mathbf{\Omega} \in \mathbb{C}^{r \times r}$ and the low-rank reconstruction of the inputs $\mathbf{\widehat{x}} \in \mathbb{R}^{n \times m}$ with time complexity $\mathcal{O}(m n \cdot \min \{m, n\})$. 
\end{lemma}

\begin{proof}
	
	The computational complexity of DMD algorithm is primarily determined by the SVD step, which is performed on the first lagged matrix $\mathbf{X} \in \mathbb{R}^{n \times m}$. The SVD of $\mathbf{X}$ can be performed in two steps \cite{svd_complexity}: the first one is to reduce $\mathbf{X}$ to a bidiagonal matrix, which can be done in $\mathcal{O}(m n \cdot \min \{m, n\})$ operations; the second one consists in computing the SVD of this new bidiagonal matrix using QR algorithm, which can be done in $\mathcal{O}(\min \{n, m\})$ operations. Therefore, the total cost for the computation of the SVD of $\mathbf{X}$ is $\mathcal{O}(m n \cdot \min \{m, n\})$ since the first step dominates the complexity. 
	
	The eigendecomposition of the projected DMD operator $\mathbf{\Tilde{A}}_r \in \mathbb{C}^{r \times r}$ can be computed using QR algorithm to obtain all eigenpairs, with a time complexity of $\mathcal{O}(r^3)$. However, one can easily show that $r \leq \min{(m, n)}$ implies $r^3 \leq m n \cdot \min \{m, n\}$: taking without loss of generality $m \leq n$, i.e. $\min{(m, n)} = m$, we have $r \leq m$ and thus 
	\begin{equation*}
		r^3 \leq m^3 \leq m^2 n = m n \cdot \min{(m, n)}.
	\end{equation*}
	Therefore, this shows that the eigendecomposition of the projected DMD operator step is dominated by the earlier SVD computation of $\mathbf{X}$. 
	
	The calculation of the pseudo-inverse of $\mathbf{\Phi} \in \mathbb{C}^{n \times r}$ relies on the SVD of $\mathbf{\Phi}$: indeed, if $\mathbf{\Phi}~=~\mathbf{U}_{\mathbf{\Phi}} \mathbf{\Sigma}_{\mathbf{\Phi}} \mathbf{V}_{\mathbf{\Phi}}^{\ast}$ where $\mathbf{U}_{\mathbf{\Phi}} \in \mathbb{C}^{n \times n}$, $ \mathbf{\Sigma}_{\mathbf{\Phi}} \in \mathbb{C}^{n \times r}$ and $\mathbf{V}_{\mathbf{\Phi}} \in \mathbb{C}^{r \times r}$, then $\mathbf{\Phi}^{\dagger} =  \mathbf{V}_{\mathbf{\Phi}} \mathbf{\Sigma^{\dagger}}_{\mathbf{\Phi}} \mathbf{U}_{\mathbf{\Phi}}^{\ast}$ where $\mathbf{\Sigma^{\dagger}}_{\mathbf{\Phi}}$ is the pseudo-inverse of $\mathbf{\Sigma}_{\mathbf{\Phi}}$ obtained by replacing each non-zero diagonal entry by its reciprocal and transposing the resulting matrix. This shows that $\mathbf{\Phi}^{\dagger}$ can be estimated in $\mathcal{O}(n r \cdot \min \{ n, r \}) = \mathcal{O}(n r^2)$ operations, which is dominated by the complexity of the SVD computation of $\mathbf{X}$. 
	
	Moreover, the reconstruction steps consist of $m$ matrix multiplications, each making $\mathcal{O}(n r)$ operations, hence resulting in the final $\mathcal{O}(m n r)$ time complexity, again dominated by $\mathcal{O}(m n \cdot \min \{ m, n \})$.
	
	The adaptive EWMA update steps can be performed with $\mathcal{O}\left(1\right)$ complexity as shown in Lemma \ref{prop:efficient_update}. 
	
	Therefore, Algorithm \ref{alg:dmd} has $\mathcal{O}(m n \cdot \min \{m, n\})$ time complexity.
	
\end{proof} 

\begin{reptheorem}{th:complexity}
	The computational complexity of sequence processing steps at each time step $t$ in Algorithm \ref{alg:single_changepoint} in the context of single changepoint detection is $\mathcal{O}\left(pd(w-d) \cdot \min \{ pd, w-d \}\right)$, where $p$ is the dimension of the input stream, $w$ is the window length and $d$ denotes the auto-regressive order. 
\end{reptheorem}

\begin{proof}
	From Algorithm \ref{alg:single_changepoint}, one can notice that at each time step $t$, the new datum $x_t$ is received and processed to construct the new Hankel batch $\mathcal{X}_t \in \mathbb{R}^{pd \times w-d+1}$ where $p$ is the input stream dimension, $w$ is the window length and $d$ is the auto-regressive order. The time complexity of these formatting steps is dominated by the DMD step processing $\mathcal{X}_t$. Applying Lemma \ref{prop:dmd_complexity}, the time complexity of this DMD step on the new Hankel batch is $\mathcal{O}\left(pd (w-d+1) \min \{ pd, w-d+1 \}\right)$, with $w \in \mathbb{N}$ and $d \in \mathbb{N}, d \leq w$. 
	
	\begin{itemize}
		\item If $pd \leq w-d+1$, then 
		\begin{equation*}
			\mathcal{O}\left(pd (w-d+1) \min \{ pd, w-d+1 \}\right) = \mathcal{O}\left(p^2 d^2 (w-d) \right).
		\end{equation*}
		\item If $pd > w-d+1$, i.e. $w-d \leq pd$, then 
		\begin{equation*}
			\mathcal{O}\left(pd (w-d+1) \min \{ pd, w-d+1 \}\right) = \mathcal{O}\left(p d (w-d)^2 \right).
		\end{equation*}
	\end{itemize}
	
	Lemma \ref{prop:efficient_update} shows that the further adaptive EWMA statistics updates are dominated by the computational complexity of the previous DMD algorithm. Therefore, the time complexity of Algorithm \ref{alg:single_changepoint} is $\mathcal{O}\left(pd (w-d) \min \{ pd, w-d \}\right)$. 
	
\end{proof}

\begin{figure}[t]
	\centering
	\includegraphics[width=\linewidth]{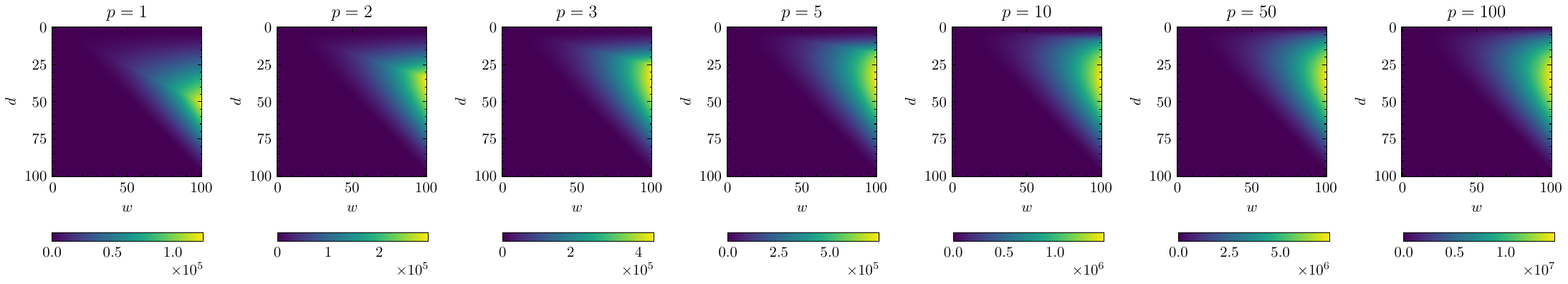}
	\caption{Illustration of $(d, w) \mapsto pd (w-d) \min \{ pd, w-d \}$ for different values of $p$ and $1~\leq~d~\leq~w~\leq~100$.}
	\label{fig:complexity}
\end{figure}

\begin{corollary}[Computational complexity of hyperparameter selection] \label{th:complexity_model_selection}
	During hyperparameter selection as described in Algorithm \ref{alg:model_selection}, for any considered model $\theta = \{w, d, r\} \in \Theta$ as defined in Equation~\eqref{eq:parameters_constraints}, the total number of operations is maximized when $d =\left[ \frac{w}{p+1} \right]$ or $d =\left[ \frac{w}{3} \right]$ and the worst case time complexity of Algorithm \ref{alg:model_selection} is $\mathcal{O}\left(p T_0^3\right)$.
\end{corollary}

\begin{proof}
	Assume one monitors a multivariate $p$-dimensional input stream. Recall that within the hyperparameter selection phase, given a burn-in period $T_0$, we consider the set of possible models $\Theta$ defined as
	\begin{equation*}
		\Theta \vcentcolon = \left\{ (w, d, r) \in \mathbb{N}^3 \mid w \leq T_0, d \leq w, r \leq \min{(pd, w-d+1)} \right\}.
	\end{equation*}
	Considering that competing model parameters $\theta = ( w, d, r ) \in \Theta$ may be run in parallel within the burn-period, the time complexity of the complete hyperparameter selection process is equal to the worst case time complexity among all considered models. One may notice that steps within the burn-in period defined in Algorithm \ref{alg:model_selection} have the same time complexity as the one stated in Theorem \ref{th:complexity}. Hence, considering a triplet $\theta = ( w, d, r ) \in \Theta$, 
	\begin{itemize}
		\item $(i)$: $pd \leq w-d$, then $pd (w-d) \cdot \min \{ pd, w-d \} = p^2 d^2 (w-d)$. Assuming fixed window length $w$, the function $d \mapsto p^2 d^2 (w - d)$ is maximised for $d = \left[ \frac{2w}{3} \right]$, i.e. the closest integer to $\frac{2w}{3}$. However, notice that $pd \leq w-d \Leftrightarrow d \leq \frac{w}{p+1}$. Thus, since the function $d \mapsto p^2 d^2 (w - d)$ is monotically increasing on $\{ 1, 2, \dots, \left[ \frac{2w}{3} \right] \}$, the closest integer on the considered domain of values is $d = \left[ \frac{w}{p+1} \right]$, and thus $p^2 d^2 (w-d) \lesssim w^3 \leq T_0^3$.
		\item $(ii)$: if $pd > w-d$, then $pd (w-d) \cdot \min \{ pd, w-d \} = pd (w-d)^2$. For any window length $w \leq T_0$, this cost is maximized for $d = \left[ \frac{w}{3} \right]$. However, notice that $pd > w-d \Leftrightarrow d > \frac{w}{p+1} \geq \frac{w}{2}$. If $p = 1$, then the closest integer on the domain is $\left[ \frac{w}{2} \right]$ as in $(i)$, otherwise $d = \left[ \frac{w}{3} \right]$, and thus $pd (w-d)^2 \lesssim p w^3 \leq p T_0^3$.
	\end{itemize}
	
\end{proof}

\paragraph{Consequences.}
Corollary \ref{th:complexity_model_selection} leverages Theorem \ref{th:complexity} result, assuming that competing models can be run in parallel on the burn-in period, to derive a maximum bound of the computational complexity of our proposed model selection procedure. 
Both results are useful to derive the complexity of our multiple changepoints detection algorithm with hyperparameter selection, described in Algorithm \ref{alg:multiple_changepoints_with_model_selection}, where parameters are adapted automatically after a change is detected with complexity $\mathcal{O}\left(p T_0^3\right)$, otherwise the data stream is processed with complexity $\mathcal{O}(pd(w-d) \cdot \min \{ pd, w-d \})$ on segments without changes.

\newpage
\section{Experiments and simulations}
\label{sec:experiments_details}

\subsection{Computing resources}
\label{sec:computing_resources}

The experiments were conducted on a high-performance computing cluster, which comprises dual Xeon CPUs with 4 to 10 cores per CPU and 16 to 256 GB of memory per node. The total computation time for the experiments was a couple of hours to 2-3 days for the longest simulations.

\subsection{Synthetic data}
\label{sec:experiment_details_synthetic}

\subsubsection{Data generation process}
\label{sec:data_generation_process}

We consider the following univariate time series generating process
\begin{equation*}
	x_t = \sum_{k=1}^{N} \alpha_k \sin{(\omega_k t)} + \beta t + \gamma + \mathcal{N}(0, \sigma^2), \text{ where}
\end{equation*}

\begin{itemize}
	\item $N$ denotes the number of periodicities in the dynamics,
	\item $\omega_k$ is the $k^{\text{th}}$ periodicity, $k \in \{ 1, 2, \dots, N \}$,
	\item $\alpha_k$ corresponds to the amplitude associated to the $k^{\text{th}}$ periodicity,
	\item $\beta$ denotes a linear trend coefficient,
	\item $\gamma$ is the stream location at initialisation, 
	\item $\sigma^2$ accounts for the noise variance.
\end{itemize}

For notational simplicity when considering a single periodicity ($N = 1$), we denote $\alpha_1$ as $\alpha$ and $\omega_1$ as $\omega$. 

The time series are generated for $t \in \{1, 2, \dots, 600\}$ with fixed parameters before and after the changepoint at time $t = \tau = 300$. Periodicities are chosen to avoid discontinuities at the changepoint location that would make the detection of subtle changes in periodicity or amplitude trivial.

We simulate a wide range of change types and magnitudes, covering diverse scenarios to thoroughly assess each algorithm's performance. The pre-change and post-change process parameters used in the simulations are presented, ensuring a comprehensive exploration of parameter spaces.

\paragraph{Change in periodicity.}
In this type of change, we consider a single periodicity ($N = 1$), and fix parameters such that $\alpha_1 = 1$, $\beta = 0$, $\gamma = 0$, $\sigma = 0.1$ and the only varying parameter is the periodicity $\omega$, which is chosen such that $\omega = \frac{6 \pi}{75}$ when $t < \tau = 300$ and $\omega \in \{ \frac{5 \pi}{75}, \frac{7 \pi}{75}, \frac{8 \pi}{75} \}$ when $t \geq \tau = 300$.

\paragraph{Change in location.}
In this type of change, we consider a single periodicity ($N = 1$), and fix parameters such that $\omega = \frac{4 \pi}{75}$, $\alpha = 1$, $\beta = 0$, $\sigma = 0.1$ and the only varying parameter is the location $\gamma$, which is chosen such that $\gamma = 0$ when $t < \tau = 300$ and $\gamma \in \{ -0.5, 0.5, 1 \}$ when $t \geq \tau = 300$.

\paragraph{Change in amplitude.}
In this type of change, we consider a single periodicity ($N = 1$), and fix parameters such that $\omega = \frac{13 \pi}{150}$, $\beta = 0$, $\gamma = 0$, $\sigma = 0.1$ and the only varying parameter is the amplitude $\alpha$, which is chosen such that $\alpha = 1$ when $t < \tau = 300$ and $\alpha \in \{ 0.5, 2, 3 \}$ when $t \geq \tau = 300$.

\paragraph{Change in trend.}
In this type of change, we consider a single periodicity ($N = 1$), and fix parameters such that $\omega = \frac{10 \pi}{75}$, $\alpha = 1$, $\sigma = 0.1$ and the only varying parameters are the trend $\beta$ and the location $\gamma$, which are chosen such that $\beta = \frac{1}{30}$ and $\gamma = 0$ when $t < \tau = 300$ and $\beta \in \{ - \frac{1}{30}, 0, \frac{2}{30} \}$ and $\gamma = 10$ when $t \geq \tau = 300$.

\paragraph{Change in mean.}
In this type of change, we consider no periodicity ($N = 0$), and fix parameters such that $\beta = 0$, $\sigma = 1$ and the only varying parameter is the location $\gamma$, which is chosen such that $\gamma = 0$ when $t < \tau = 300$ and $\gamma \in \{ -2, 3, 4 \}$ when $t \geq \tau = 300$.

\paragraph{Change in variance.}
In this type of change, we consider no periodicity ($N = 0$), and fix parameters such that $\beta = 0$, $\gamma = 0$ and the only varying parameter is the noise variance $\sigma$, which is chosen such that $\sigma = 0.1$ when $t < \tau = 300$ and $\sigma \in \{ 0.2, 0.3, 0.4 \}$ when $t \geq \tau = 300$.

\paragraph{Change in double periodicity.}
In this type of change, we consider two periodicities ($N = 2$), and fix parameters such that $\alpha_1 = \alpha_2 = 1$, $\beta = 0$, $\gamma = 0$, $\sigma = 0.1$ and the only varying parameters are the two periodicities $\omega_1$ and $\omega_2$, which are chosen such that $(\omega_1, \omega_2) = (\frac{9 \pi}{75}, \frac{6 \pi}{75})$ when $t < \tau = 300$ and $(\omega_1, \omega_2) \in \left\{ (\frac{3 \pi}{75}, \frac{5 \pi}{75}), (\frac{9 \pi}{75}, \frac{3 \pi}{75}), (\frac{9 \pi}{75}, \frac{4 \pi}{75}) \right\}$ when $t \geq \tau = 300$.

\subsubsection{Models parameters}
\label{sec:synthetic_models_parameters}

In this section, we detail the different methods and parameters that were explored through simulations on synthetic data in Section \ref{sec:experiments_synthetic}. We investigate a variety of parameters across all models, attempting to make the comparison as fair as possible. The burn-in periods (or equivalent parameters) are configured \textit{similarly} for each model to facilitate equivalent initialisations. 

\paragraph{CPDMD.}

\begin{itemize}
	
	\item \textbf{Burn-in period} $T_0 = 100$. 
	
	As discussed in Section \ref{sec:model_selection}, the selection of the burn-in period $T_0$ directly generates the set of possible parameters $\Theta$, defined as 
	\begin{equation*}
		\Theta = \left\{ (w, d, r) \in \mathbb{N}^3 \mid w \leq T_0, d \leq w, r \leq \min{(pd, w-d+1)} \right\},
	\end{equation*}
	where $p$ is the dimension of the input stream. Due to computational constraints, not all triplets can be explored on the burn-in period. We follow the following generation process for $\Theta$ which only requires the burn-in period $T_0$, and  restricts the explored parameters to
	\begin{itemize}
		\item Window length $w \in \{ 0.4, 0.6, 0.8 \} \cdot T_0$,
		\item Auto-regressive order $d \in \{ 0.05, 0.1, 0.2, 0.4 \} \cdot T_0$,
		\item Truncation rank $r \in 2 \cdot \{1, 2, \dots, p\} = \{2, 4\}$ since $p=1$.
	\end{itemize}
	\item \textbf{Learning rate} $\lambda = 0.05 \text{ (default)}, 0.10$ of the adaptive EWMA algorithm.
	\item \textbf{Control limit} $L = 1.5, 2.5, 3.5, 4.5 \text{ (default)}, 5.5$ of the adaptive EWMA algorithm.
\end{itemize}

\paragraph{mSSA and mSSA-MW.}
We use the Python implementation provided by the authors of \cite{mssa} and available via the GitHub repository \href{https://github.com/ArwaAlanqary/mSSA_cpd}{\small \texttt{https://github.com/ArwaAlanqary/mSSA\_cpd}}. We follow the authors guidelines and consider the following parameters:
\begin{itemize}
	\item \textbf{Window size} $= 90, 100 \text{ (default)}, 120$.
	\item \textbf{Rows} $= 3, 7, 10 \text{ (default)}, 13, 15$.
	\item \textbf{Distance threshold} $= 1, 5 \text{ (default)}, 10$.
	\item \textbf{Rank} $= 0.95$.
	\item \textbf{Training size} $= 0.9$.
\end{itemize}

\paragraph{RuLSIF.}
We use the MATLAB implementation provided by the authors of \cite{rulsif} and available via the GitHub repository \href{https://github.com/anewgithubname/change_detection}{\small \texttt{https://github.com/anewgithubname/change\_detection}}. While RuLSIF is introduced as a retrospective changepoint detection method, its implementation is analogous to online context, where a threshold needs to be set (and cannot be known in advance) to make online decisions. We follow the authors guidelines and consider the following parameters:
\begin{itemize}
	\item $\boldsymbol{n} = 25, 50 \text{ (default)}, 75$.
	\item $\boldsymbol{k} = 5, 10 \text{ (default)}, 15$.
	\item $\boldsymbol{\alpha} = 0.1$.
	\item \textbf{Threshold} $= 0.001, 0.01, 0.05, 1, 2 \text{ (default)}, 3, 4, 5$.
\end{itemize}

\paragraph{BOCPDMS.}
We use the Python implementation provided by the authors of \cite{bocpdms} and avaialble via the GitHub repository \href{https://github.com/alan-turing-institute/bocpdms}{\small \texttt{https://github.com/alan-turing-institute/bocpdms}}. We consider the same following grid of parameters as in \cite{mssa}
\begin{itemize}
	\item \textbf{Prior on $\boldsymbol{a}$} $= 0.01, 1 \text{ (default)}, 100$.
	\item \textbf{Prior on $\boldsymbol{b}$} $= 0.01, 1 \text{ (default)}, 100$.
	\item \textbf{Intensity} $= 50, 100 \text{ (default)}, 200$.
\end{itemize}

\paragraph{EWMA.}
We use a self-implementation of the EWMA algorithm introduced by \cite{ewma} in Python, and we consider the following grid of parameters:
\begin{itemize}
	\item \textbf{Burn-in period} $T_0 = 100$.
	\item \textbf{Learning rate} $r = 0.05 \text{ (default)}, 0.1$.
	\item \textbf{Control limit} $L = 1.5, 2.5 \text{ (default)}, 3.5, 4.5, 5.5$.
\end{itemize}

\paragraph{EWMVar.}
We use a self-implementation of the EWMVar algorithm introduced by \cite{ewmvar} in Python, and we consider the following grid of parameters:
\begin{itemize}
	\item \textbf{Burn-in period} $T_0 = 100$.
	\item \textbf{Learning rate for the variance} $r = 0.01, 0.05 \text{ (default)}$.
	\item \textbf{Learning rate for the mean} $\lambda = 0.2$.
	\item \textbf{Lower control limit} $C_7 = 0.42, 0.63, 0.68 \text{ (default)}, 0.91$.
	\item \textbf{Upper control limit} $C_8 = 1.08 \text{ (default)}, 1.12, 1.14, 1.70$.
\end{itemize}

Since all algorithms have different computational complexites, we may run each of them on a different number of sequences. We provide below in Table \ref{tab:simulations_sequence_samples} a summary of the number of sequences that are used for the simulations.

\clearpage

\begin{table}[h!]
	\centering
	\smaller
	\caption{Number of sequences used to estimate the performance metrics for each algorithm, change type, and change size.}
	\begin{tabular}{l c c}
		\toprule
		& \textbf{Sequences with a single changepoint} & \textbf{Sequences with no changepoint} \\
		\midrule
		\textbf{Sequence length} & $t \in [1, 600]$ & $t \in [1, 100000]$ \\
		\textbf{Metrics} & Precision, recall, $F_1$-Score, $\text{ARL}_1$ & $\text{ARL}_0$ \\
		\midrule
		EWMA & $1000$ & $100$ \\
		EWMVar & $1000$ & $100$ \\
		BOCPDMS & $50$ & N/A \\
		RuLSIF & $100$ & $10$ \\
		mSSA & $1000$ & $100$ \\
		mSSA-MW & $1000$ & $100$ \\
		CPDMD & $1000$ & $100$ \\
		\bottomrule
	\end{tabular}
	\label{tab:simulations_sequence_samples}
\end{table} 

\begin{table}[h!]
	\centering
	\smaller
	\caption{Summary of selected changepoint detection algorithms capabilities.}
	\begin{tabular}{l c c c c}
		\toprule
		\textbf{Algorithm} & \textbf{Online} & \textbf{Nonparametric} & \textbf{Multivariate data} & \textbf{Periodic data} \\
		\midrule
		EWMA & \cmark & \cmark & \xmark & \xmark \\
		EWMVar & \cmark & \cmark & \xmark & \xmark \\
		BOCPDMS & \cmark & \xmark & \cmark & \xmark \\
		RuLSIF & \xmark & \cmark & \cmark & \cmark \\
		mSSA & \cmark & \cmark & \cmark & \cmark \\
		mSSA-MW & \cmark & \cmark & \cmark & \cmark \\
		\textbf{CPDMD} & \cmark & \cmark & \cmark & \cmark \\
		\bottomrule
	\end{tabular}
	\label{tab:methods_capabilities}
\end{table}

\clearpage
\subsubsection{Global performance}
\label{sec:global_with_bocpdms}

We provide below additional results with BOCPDMS performance for simulations in the context of single changepoint detection. Note that Table \ref{tab:performance_synthetic_with_bocpdms} is thus the same as Table \ref{tab:performance_synthetic_without_bocpdms} in Section \ref{sec:experiments_simulations}, with only additional performance metrics for the BOCPDMS algorithm. 

\begin{table}[h!]
	\centering
	\smaller
	\caption{Evaluation of selected changepoint detection algorithms on synthetic data, comparing performance achieved with the parameter set yielding highest $F_1$-Score (\textit{best}) against default parametrisation (\textit{default}). We highlight the top two performance metrics for each set of parameters in bold.}
	\begin{tabular}{l l r@{\hskip\tabcolsep} >{\scriptsize}r@{\hskip\tabcolsep} r@{\hskip\tabcolsep} >{\scriptsize}r@{\hskip\tabcolsep} r@{\hskip\tabcolsep} >{\scriptsize}r@{\hskip\tabcolsep} r@{\hskip\tabcolsep} >{\scriptsize}r@{\hskip\tabcolsep} r@{\hskip\tabcolsep} >{\scriptsize}r@{\hskip\tabcolsep}}
		\toprule
		\textbf{Algorithm} & \textbf{Params.} & \multicolumn{2}{c}{$\boldsymbol{P}$} & \multicolumn{2}{c}{$\boldsymbol{R}$} & \multicolumn{2}{c}{$\boldsymbol{F_1}$} & \multicolumn{2}{c}{$\mathbf{ARL_1}$} & \multicolumn{2}{c}{$\mathbf{ARL_0}$} \\
		\midrule
		\multirow{2}{*}{EWMA} & Best & \textbf{.669} & (.004) & \textbf{.474}  & (.003) & \textbf{.555} & (.003) & 9.78 & (7.44) & \textbf{71085.25} & (41976.94) \\
		& Default & .188 & (.003) & .181 & (.003) & .185 & (.003) & \textbf{5.88} & (5.11) & 4261.15 & (18156.77) \\
		\hline
		\multirow{2}{*}{EWMVar} & Best & .275 & (.005) & .131 & (.002) & .178 & (.003) & 10.96 & (6.86) & \textbf{71202.77} & (45092.24) \\
		& Default & .452 & (.009) & .089 & (.002) & .148 & (.003) & 12.36 & (7.40) & \textbf{99900.00} & (0.00) \\
		\hline
		\multirow{2}{*}{BOCPDMS} & Best & .540 & (.020) & .319 & (.014) & .401 & (.016) & \textbf{4.50} & (5.83) & N/A & (N/A) \\
		& Default & \textbf{.492} & (.020) & .297 & (.013) & .371 & (.015) & \textbf{4.52} & (5.96) & N/A & (N/A) \\
		\hline
		\multirow{2}{*}{RuLSIF} & Best & .537 & (.013) & .414 & (.010) & .468 & (.011) & 19.01 & (6.80) & 41934.74 & (49177.96) \\
		& Default & .448 & (.011) & \textbf{.364} & (.010) & \textbf{.401} &(.010)  & 23.58 & (4.61) & 57147.14 & (49746.53) \\
		\hline
		\multirow{2}{*}{mSSA} & Best & .475 & (.005) & .266 & (.003) & .341 & (.004) & \textbf{6.91} & (6.09) & 43786.39 & (47153.54) \\
		& Default & .304 & (.003) & .210 & (.002) & .248 & (.003) & 6.28 & (5.80) & 27023.57 & (42587.84) \\
		\hline
		\multirow{2}{*}{mSSA-MW} & Best & .443 & (.005) & .313 & (.004)& .367 & (.004) & 8.79 & (5.87) & 13834.54 & (32965.67) \\
		& Default & .043 & (.001) & .041 & (.001) & .042 & (.001) & 9.11 & (7.29) & 8210.28 & (20999.07) \\
		\hline
		\multirow{2}{*}{\textbf{CPDMD (ours})} & Best & \textbf{.960} & (.001) & \textbf{.902} & (.002) & \textbf{.930} & (.002) & 7.11 & (6.14) & 4034.30 & (3997.10) \\
		& Default & \textbf{.978} & (.001) & \textbf{.807} & (.003) & \textbf{.884} & (.002) & 9.53 & (6.47) & \textbf{58341.28} & (38161.07) \\
		\bottomrule
	\end{tabular}
	\label{tab:performance_synthetic_with_bocpdms}
\end{table}

\begin{figure}[h!]
	\centering
	\includegraphics[width=\linewidth]{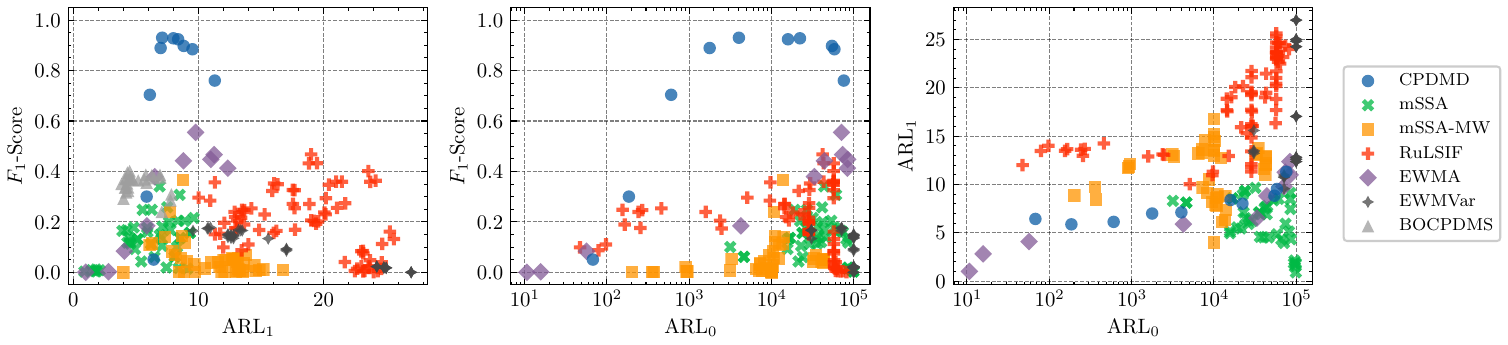}
	\caption{Performance comparison of changepoint detection algorithms on synthetic data across all scenarios. Each dot corresponds to an algorithm with a specific choice of parameter values.}
	\label{fig:simulations_global_with_bocpdms}
\end{figure}

\clearpage
\subsubsection{Change in periodicity}

The synthetic data generation process for this type of change is detailed in Appendix \ref{sec:data_generation_process}, and we provide in Figure \ref{fig:illustration_periodicity} an illustration of this change type and the considered change sizes. The successful detection region corresponds to the time region in which detected changepoints are counted as true positives. With respect to Definition \ref{def:f1}, we use a left margin $\mu_l = 0$ and a right margin $\mu_r = 30$. Performance of the grid of tested models is shown in Figure \ref{fig:simulations_periodicity} and a summary of performance is provided in Table \ref{tab:simulations_periodicity}.

\begin{figure}[h!]
	\centering
	\includegraphics[width=\linewidth]{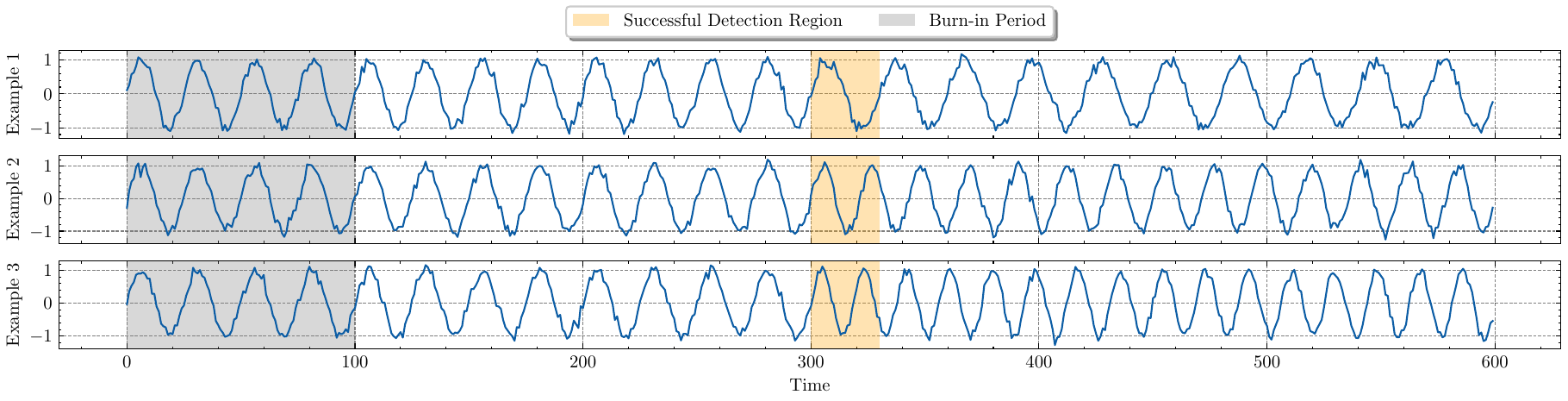}
	\caption{Illustration of the different change sizes considered in the periodicity change simulations.}
	\label{fig:illustration_periodicity}
\end{figure}

\begin{figure}[h!]
	\centering
	\includegraphics[width=\linewidth]{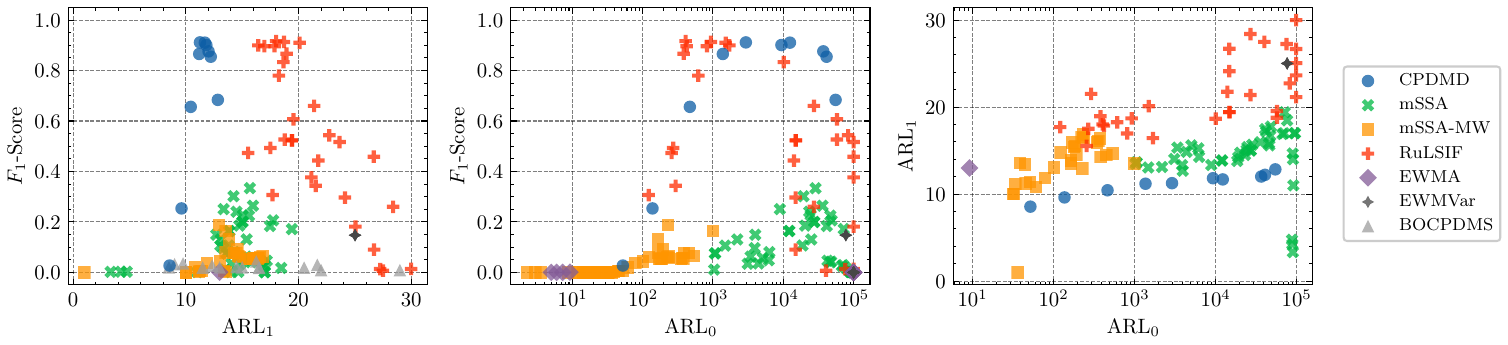}
	\caption{Performance comparison of changepoint detection algorithms on synthetic data for a change in periodicity. Each dot corresponds to an algorithm with a specific choice of parameter values.}
	\label{fig:simulations_periodicity}
\end{figure}

\begin{table}[h!]
	\centering
	\smaller
	\caption{Evaluation of selected changepoint detection algorithms on synthetic data, comparing performance achieved with the parameter set yielding highest $F_1$-Score (\textit{best}) against default parametrisation (\textit{default}). We highlight the top two performance metrics for each set of parameters in {bold}.}
	\begin{tabular}{l l r@{\hskip\tabcolsep} >{\scriptsize}r@{\hskip\tabcolsep} r@{\hskip\tabcolsep} >{\scriptsize}r@{\hskip\tabcolsep} r@{\hskip\tabcolsep} >{\scriptsize}r@{\hskip\tabcolsep} r@{\hskip\tabcolsep} >{\scriptsize}r@{\hskip\tabcolsep} r@{\hskip\tabcolsep} >{\scriptsize}r@{\hskip\tabcolsep}}
		\toprule
		\textbf{Algorithm} & \textbf{Params.} & \multicolumn{2}{c}{$\boldsymbol{P}$} & \multicolumn{2}{c}{$\boldsymbol{R}$} & \multicolumn{2}{c}{$\boldsymbol{F_1}$} & \multicolumn{2}{c}{$\mathbf{ARL_1}$} & \multicolumn{2}{c}{$\mathbf{ARL_0}$} \\
		\midrule
		\multirow{2}{*}{EWMA} & Best & .001 & (.001) & .001 & (.000) & .001 & (.000) & 13.00 & (0.00) & 9.27 & (0.58) \\
		& Default & .001 & (.001) & .001 & (.000) & .001 & (.000) & 13.00 & (0.00) & 9.27 & (0.58) \\
		\hline
		\multirow{2}{*}{EWMVar} & Best & .297 & (.014) & .098 & (.005) & .147 & (.008) & 25.02 & (6.45) & \textbf{77421.77} & (35052.88) \\
		& Default & .000 & (.000) & .000 & (.000) & .000 & (.000) & N/A & (N/A) & \textbf{99900.00} & (0.00) \\
		\hline
		\multirow{2}{*}{BOCPDMS} & Best & .121 & (.059) & .027 & (.014) & .044 & (.021) & 16.25 & (10.21) & N/A & (N/A) \\
		& Default & .041 & (.024) & .020 & (.012) & .027 & (.015) & 12.33 & (7.51) & N/A & (N/A) \\
		\hline
		\multirow{2}{*}{RuLSIF} & Best &\textbf{ .917} & (.016) & \textbf{.917} & (.016) & \textbf{.917} & (.016) & 18.02 & (4.64) & 409.30 & (411.67) \\
		& Default & .000 & (.000) & .000 & (.000) & .000 & (.000) & N/A & (N/A) & \textbf{99920.00} & (0.00) \\
		\hline
		\multirow{2}{*}{mSSA} & Best & .375 & (.010) & .301 & (.009) & .334 & (.009) & 15.68 & (5.62) & \textbf{28707.71} & (42667.33) \\
		& Default & \textbf{.288} & (.010) & \textbf{.188} & (.007) & \textbf{.227} & (.008) & 15.07 & (6.88) & 24466.87 & (40557.90) \\
		\hline
		\multirow{2}{*}{mSSA-MW} & Best & .198 & (.008) & .178 & (.007) & .187 & (.007) & \textbf{12.93} & (6.78) & 229.29 & (207.34) \\
		& Default & .037 & (.003) & .037 & (.003) & .037 & (.003) & \textbf{11.87} & (6.77) & 79.18 & (65.99) \\
		\hline
		\multirow{2}{*}{\textbf{CPDMD (ours)}} & Best & \textbf{.955} & (.004) & \textbf{.872} & (.006) & \textbf{.912} & (.005) & \textbf{11.25} & (2.74) & 2936.04 & (2642.37) \\
		& Default & \textbf{.994} & (.002) & \textbf{.750} & (.008) & \textbf{.855} & (.005) & \textbf{12.21} & (2.56) & 41186.52 & (35630.17) \\
		\bottomrule
	\end{tabular}
	\label{tab:simulations_periodicity}
\end{table}

\clearpage
\subsubsection{Change in location}

The synthetic data generation process for this type of change is detailed in Appendix \ref{sec:data_generation_process}, and we provide in Figure \ref{fig:illustration_location} an illustration of this change type and the considered change sizes. The successful detection region corresponds to the time region in which detected changepoints are counted as true positives. With respect to Definition \ref{def:f1}, we use a left margin $\mu_l = 0$ and a right margin $\mu_r = 30$. Performance of the grid of tested models is shown in Figure \ref{fig:simulations_location} and a summary of performance is provided in Table \ref{tab:simulations_location}.

\begin{figure}[h!]
	\centering
	\includegraphics[width=\linewidth]{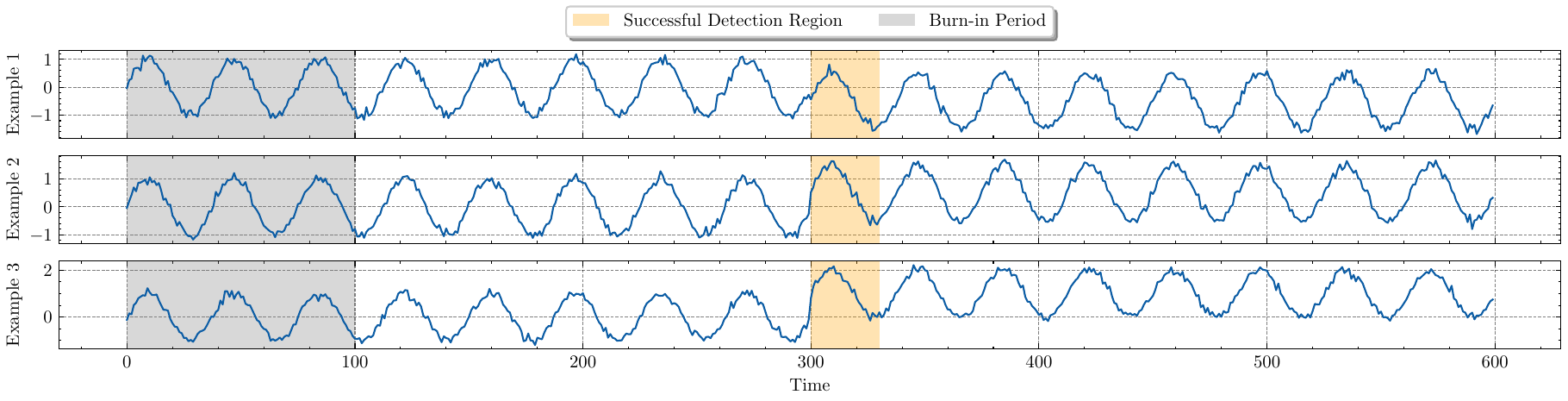}
	\caption{Illustration of the different change sizes considered in the location change simulations.}
	\label{fig:illustration_location}
\end{figure}

\begin{figure}[h!]
	\centering
	\includegraphics[width=\linewidth]{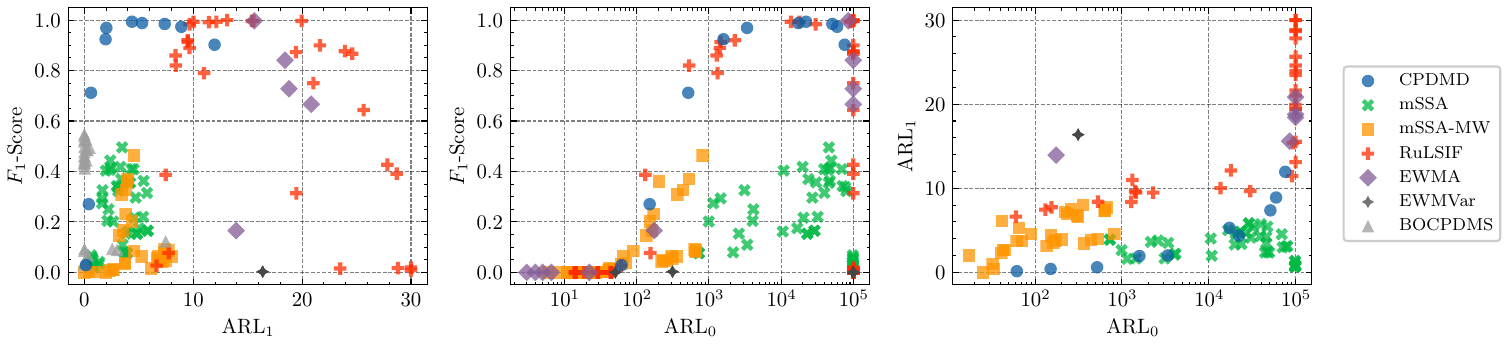}
	\caption{Performance comparison of changepoint detection algorithms on synthetic data for a change in location. Each dot corresponds to an algorithm with a specific choice of parameter values.}
	\label{fig:simulations_location}
\end{figure}

\begin{table}[h!]
	\centering
	\smaller
	\caption{Evaluation of selected changepoint detection algorithms on synthetic data, comparing performance achieved with the parameter set yielding highest $F_1$-Score (\textit{best}) against default parametrisation (\textit{default}). We highlight the top two performance metrics for each set of parameters in {bold}.}
	\begin{tabular}{l l r@{\hskip\tabcolsep} >{\scriptsize}r@{\hskip\tabcolsep} r@{\hskip\tabcolsep} >{\scriptsize}r@{\hskip\tabcolsep} r@{\hskip\tabcolsep} >{\scriptsize}r@{\hskip\tabcolsep} r@{\hskip\tabcolsep} >{\scriptsize}r@{\hskip\tabcolsep} r@{\hskip\tabcolsep} >{\scriptsize}r@{\hskip\tabcolsep}}
		\toprule
		\textbf{Algorithm} & \textbf{Params.} & \multicolumn{2}{c}{$\boldsymbol{P}$} & \multicolumn{2}{c}{$\boldsymbol{R}$} & \multicolumn{2}{c}{$\boldsymbol{F_1}$} & \multicolumn{2}{c}{$\mathbf{ARL_1}$} & \multicolumn{2}{c}{$\mathbf{ARL_0}$} \\
		\midrule
		\multirow{2}{*}{EWMA} & Best & \textbf{.997} & (.001) & \textbf{.997} & (.001) & \textbf{.997} & (.001) & 15.61 & (7.88) & \textbf{85058.20} & (29473.09) \\
		& Default & .000 & (.000) & .000 & (.000) & .000 & (.000) & N/A & (N/A) & 6.57 & (0.50) \\
		\hline
		\multirow{2}{*}{EWMVar} & Best & .003 & (.001) & .003 &(.001) & .003 & (.001) & 16.38 & (0.92) & 313.45 & (33.01) \\
		& Default & .000 & (.000) & .000 & (.000) & .000 & (.000) & N/A & (N/A) & \textbf{99900.00} & (0.00) \\
		\hline
		\multirow{2}{*}{BOCPDMS} & Best & .983 & (.017) & .380 & (.039) & .548 & (.042) & 0.00 & (0.00) & N/A & (N/A) \\
		& Default & \textbf{.578} & (.054) & .320 & (.039) & .412 & (.044) & 0.04 & (0.29) & N/A & (N/A) \\
		\hline
		\multirow{2}{*}{RuLSIF} & Best & \textbf{1.000} & (.000) & \textbf{1.000} & (.000) & \textbf{1.000} & (.000) & 13.11 & (2.57) & \textbf{99970.00} & (0.00) \\
		& Default & .877 & (.019) & \textbf{.877} & (.019) & \textbf{.877} & (.019) & 23.93 & (3.63) & \textbf{99920.00} & (0.00) \\
		\hline
		\multirow{2}{*}{mSSA} & Best & .771 & (.011) & .365 & (.009) & .496 & (.009) & \textbf{3.46} & (2.48) & 45940.81 & (46430.82) \\
		& Default & .538 & (.012) & .342 & (.009) & .418 & (.010) & \textbf{3.31} & (2.58) & 19904.84 & (36828.91) \\
		\hline
		\multirow{2}{*}{mSSA-MW} & Best & .607 & (.011) & .375 & (.009) & .464 & (.009) & 4.56 & (3.52) & 821.39 & (875.52) \\
		& Default & .038 & (.004) & .038 & (.004) & .038 & (.004) & \textbf{3.74} & (5.44) & 71.00 & (54.03) \\
		\hline
		\multirow{2}{*}{\textbf{CPDMD (ours})} & Best & \textbf{.997} & (.001) & .991 & (.002) & .994 & (.001) & \textbf{4.37} & (6.72) & 22104.95 & (20841.94) \\
		& Default & \textbf{.993} & (.002) & \textbf{.955} & (.004) & \textbf{.973} & (.002) & 8.90 & (8.06) & 59582.69 & (37379.71) \\
		\bottomrule
	\end{tabular}
	\label{tab:simulations_location}
\end{table}

\clearpage
\subsubsection{Change in amplitude}

The synthetic data generation process for this type of change is detailed in Appendix \ref{sec:data_generation_process}, and we provide in Figure \ref{fig:illustration_amplitude} an illustration of this change type and the considered change sizes. The successful detection region corresponds to the time region in which detected changepoints are counted as true positives. With respect to Definition \ref{def:f1}, we use a left margin $\mu_l = 0$ and a right margin $\mu_r = 30$. Performance of the grid of tested models is shown in Figure \ref{fig:simulations_amplitude} and a summary of performance is provided in Table \ref{tab:simulations_amplitude}.

\begin{figure}[h!]
	\centering
	\includegraphics[width=\linewidth]{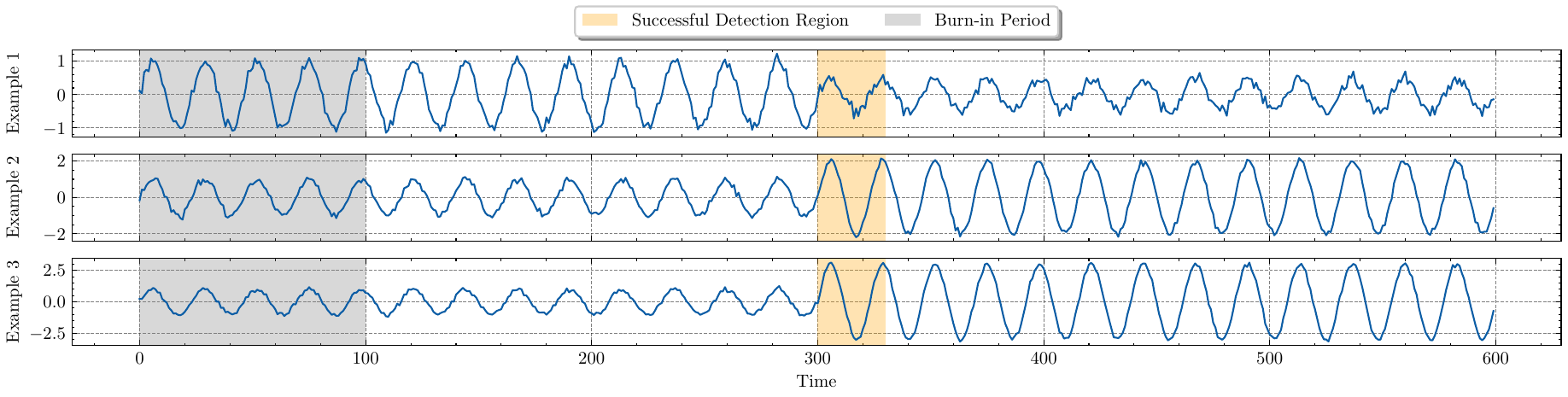}
	\caption{Illustration of the different change sizes considered in the amplitude change simulations.}
	\label{fig:illustration_amplitude}
\end{figure}

\begin{figure}[h!]
	\centering
	\includegraphics[width=\linewidth]{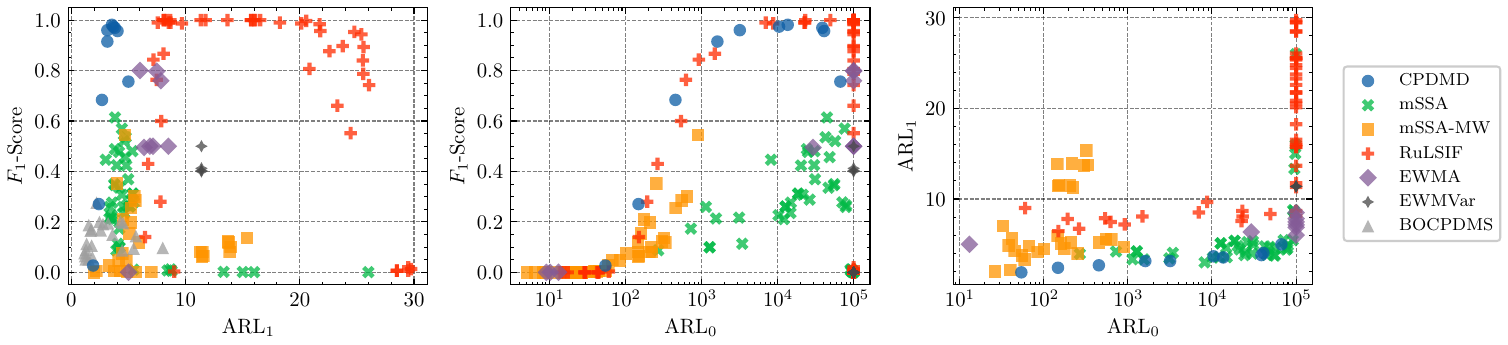}
	\caption{Performance comparison of changepoint detection algorithms on synthetic data for a change in amplitude. Each dot corresponds to an algorithm with a specific choice of parameter values.}
	\label{fig:simulations_amplitude}
\end{figure}

\begin{table}[h!]
	\centering
	\smaller
	\caption{Evaluation of selected changepoint detection algorithms on synthetic data, comparing performance achieved with the parameter set yielding highest $F_1$-Score (\textit{best}) against default parametrisation (\textit{default}). We highlight the top two performance metrics for each set of parameters in {bold}.}
	\begin{tabular}{l l r@{\hskip\tabcolsep} >{\scriptsize}r@{\hskip\tabcolsep} r@{\hskip\tabcolsep} >{\scriptsize}r@{\hskip\tabcolsep} r@{\hskip\tabcolsep} >{\scriptsize}r@{\hskip\tabcolsep} r@{\hskip\tabcolsep} >{\scriptsize}r@{\hskip\tabcolsep} r@{\hskip\tabcolsep} >{\scriptsize}r@{\hskip\tabcolsep}}
		\toprule
		\textbf{Algorithm} & \textbf{Params.} & \multicolumn{2}{c}{$\boldsymbol{P}$} & \multicolumn{2}{c}{$\boldsymbol{R}$} & \multicolumn{2}{c}{$\boldsymbol{F_1}$} & \multicolumn{2}{c}{$\mathbf{ARL_1}$} & \multicolumn{2}{c}{$\mathbf{ARL_0}$} \\
		\midrule
		\multirow{2}{*}{EWMA} & Best & \textbf{1.000} & (.000) & .667 & (.009) & .800 & (.006) & 6.01 & (1.01) & \textbf{99900.00} & (0.00) \\
		& Default & .567 & (.010) & .440 & (.009) & .495 & (.009) & 6.36 & (0.97) & 29236.98 & (39902.66) \\
		\hline
		\multirow{2}{*}{EWMVar} & Best & \textbf{1.000} & (.000) & .333 & (.008) & .500 & (.009) & 11.39 & (1.66) & \textbf{99900.00} & (0.00) \\
		& Default & .000 & (.000) & .000 & (.000) & .000 & (.000) & N/A & (N/A) & \textbf{99900.00} & (0.00) \\
		\hline
		\multirow{2}{*}{BOCPDMS} & Best & .371 & (.052) & .220 & (.035) & .276 & (.041) & \textbf{2.21} & (4.48) & N/A & (N/A) \\
		& Default & .206 & (.041) & .140 & (.029) & .167 & (.034) & \textbf{2.57} & (4.34) & N/A & (N/A) \\
		\hline
		\multirow{2}{*}{RuLSIF} & Best & \textbf{1.000} & (.000) & \textbf{1.000} & (.000) & \textbf{1.000} & (.000) & 8.06 & (2.66) & 23118.50 & (19393.84) \\
		& Default & \textbf{.943} & (.014) & \textbf{.943} & (.014) & \textbf{.943} & (.014) & 25.39 & (2.25) & \textbf{99920.00} & (0.00) \\
		\hline
		\multirow{2}{*}{mSSA} & Best & .813 & (.009) & .493 & (.009) & .614 & (.009) & 3.83 & (2.01) & 44491.10 & (45021.24) \\
		& Default & .571 & (.011) & .427 & (.009) & .489 & (.010) & \textbf{3.80} & (2.16) & 24336.95 & (38369.22) \\
		\hline
		\multirow{2}{*}{mSSA-MW} & Best & .620 & (.010) & .485 & (.009) & .544 & (.009) & 4.70 & (2.69) & 903.92 & (846.10) \\
		& Default & .046 & (.004) & .046 & (.004) & .046 & (.004) & 4.14 & (5.08) & 84.69 & (66.65) \\
		\hline
		\multirow{2}{*}{\textbf{CPDMD (ours})} & Best & .995 & (.001) & \textbf{.969} & (.003) & \textbf{.981} & (.002) & \textbf{3.55} & (1.84) & 13587.46 & (16750.13) \\
		& Default & \textbf{.997} & (.001) & \textbf{.920} & (.005) & \textbf{.957} & (.003) & 4.04 & (1.88) & 40865.57 & (35297.29) \\
		\bottomrule
	\end{tabular}
	\label{tab:simulations_amplitude}
\end{table}

\clearpage
\subsubsection{Change in trend}

The synthetic data generation process for this type of change is detailed in Appendix \ref{sec:data_generation_process}, and we provide in Figure \ref{fig:illustration_trend} an illustration of this change type and the considered change sizes. The successful detection region corresponds to the time region in which detected changepoints are counted as true positives. With respect to Definition \ref{def:f1}, we use a left margin $\mu_l = 0$ and a right margin $\mu_r = 30$. Performance of the grid of tested models is shown in Figure \ref{fig:simulations_trend} and a summary of performance is provided in Table \ref{tab:simulations_trend}.

\begin{figure}[h!]
	\centering
	\includegraphics[width=\linewidth]{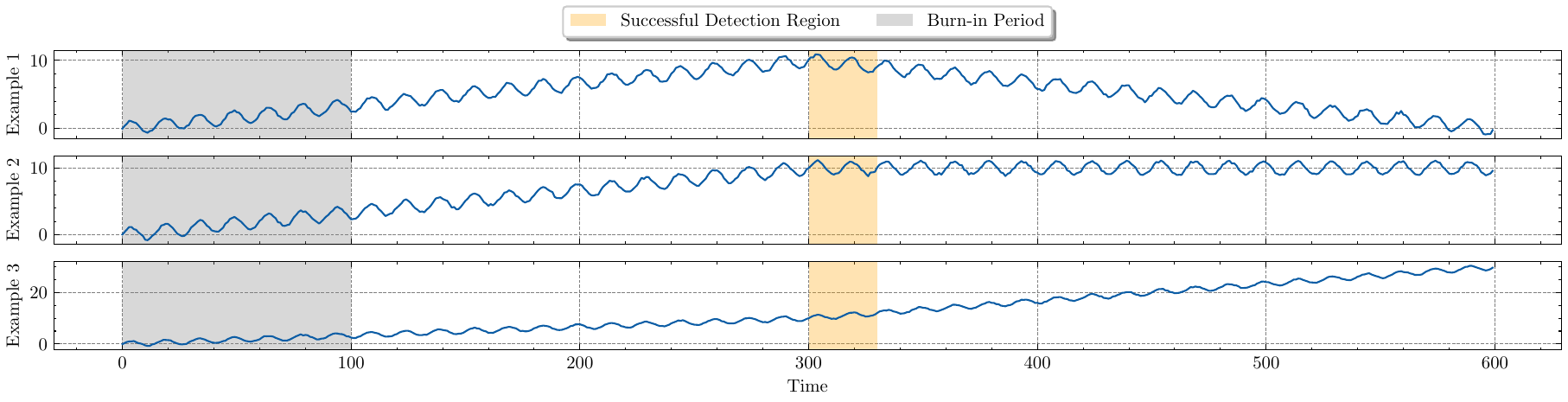}
	\caption{Illustration of the different change sizes considered in the trend change simulations.}
	\label{fig:illustration_trend}
\end{figure}

\begin{figure}[h!]
	\centering
	\includegraphics[width=\linewidth]{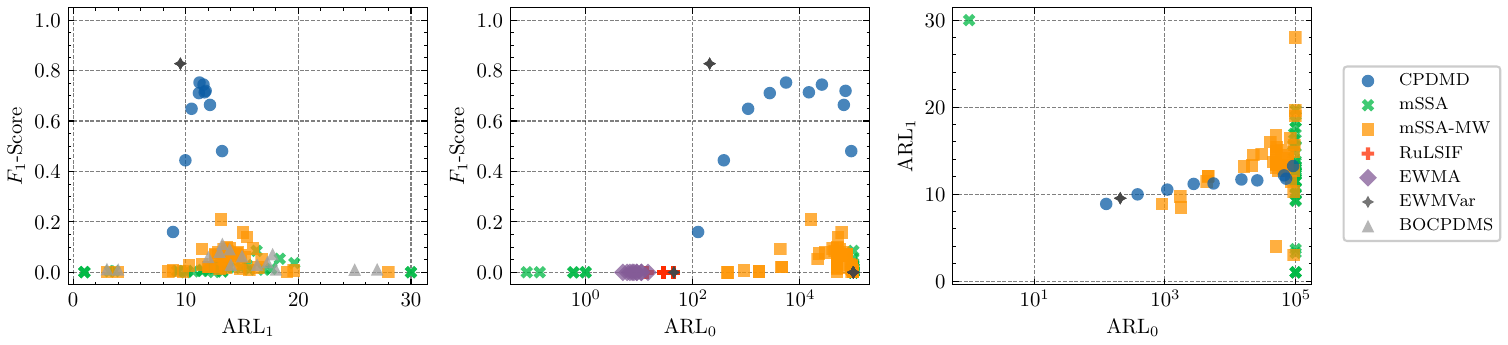}
	\caption{Performance comparison of changepoint detection algorithms on synthetic data for a change in trend. Each dot corresponds to an algorithm with a specific choice of parameter values.}
	\label{fig:simulations_trend}
\end{figure}

\begin{table}[h!]
	\centering
	\smaller
	\caption{Evaluation of selected changepoint detection algorithms on synthetic data, comparing performance achieved with the parameter set yielding highest $F_1$-Score (\textit{best}) against default parametrisation (\textit{default}). We highlight the top two performance metrics for each set of parameters in {bold}.}
	\begin{tabular}{l l r@{\hskip\tabcolsep} >{\scriptsize}r@{\hskip\tabcolsep} r@{\hskip\tabcolsep} >{\scriptsize}r@{\hskip\tabcolsep} r@{\hskip\tabcolsep} >{\scriptsize}r@{\hskip\tabcolsep} r@{\hskip\tabcolsep} >{\scriptsize}r@{\hskip\tabcolsep} r@{\hskip\tabcolsep} >{\scriptsize}r@{\hskip\tabcolsep}}
		\toprule
		\textbf{Algorithm} & \textbf{Params.} & \multicolumn{2}{c}{$\boldsymbol{P}$} & \multicolumn{2}{c}{$\boldsymbol{R}$} & \multicolumn{2}{c}{$\boldsymbol{F_1}$} & \multicolumn{2}{c}{$\mathbf{ARL_1}$} & \multicolumn{2}{c}{$\mathbf{ARL_0}$} \\
		\midrule
		\multirow{2}{*}{EWMA} & Best & .000 & (.000) & .000 & (.000) & .000 & (.000) & N/A & (N/A) & 14.67 & (1.15) \\
		& Default & .000 &(.000)  & .000 & (.000) & .000 & (.000) & N/A & (N/A) & 7.47 & (0.50) \\
		\hline
		\multirow{2}{*}{EWMVar} & Best & \textbf{.827} & (.007) & \textbf{.827} & (.007) & \textbf{.827} & (.007) & \textbf{9.52} & (5.38) & 207.76 & (6.63) \\
		& Default & .000 & (.000) & .000 & (.000) & .000 & (.000) & N/A & (N/A) & \textbf{99900.00} & (0.00) \\
		\hline
		\multirow{2}{*}{BOCPDMS} & Best & .194 & (.050) & .080 & (.022) & .113 & (.030) & 13.25 & (9.43) & N/A & (N/A) \\
		& Default & .000 & (.000) & .000 & (.000) & .000 & (.000) & N/A & (N/A) & N/A & (N/A) \\
		\hline
		\multirow{2}{*}{RuLSIF} & Best & .000 & (.000) & .000 & (.000) & .000 & (.000) & N/A & (N/A) & 44.00 & (0.00) \\
		& Default & .000 & (.000) & .000 & (.000) & .000 & (.000) & N/A & (N/A) & 29.00 & (0.00) \\
		\hline
		\multirow{2}{*}{mSSA} & Best & .102 & (.006) & .073 & (.005) & .085 & (.005) & 16.30 & (7.51) & \textbf{99880.00} & (0.00) \\
		& Default & .033 & (.011) & .003 & (.001) & .006 & (.002) & \textbf{11.44} & (8.19) & \textbf{99900.00} & (0.00) \\
		\hline
		\multirow{2}{*}{mSSA-MW} & Best & .215 & (.007) & .204 & (.007)& .210 & (.007) & 13.15 & (5.35) & \textbf{16158.87} & (13911.10) \\
		& Default & \textbf{.103} & (.006) & \textbf{.082} & (.005) & \textbf{.092} & (.005) & 13.36 & (8.73) & 56953.57 & (17415.25) \\
		\hline
		\multirow{2}{*}{\textbf{CPDMD (ours})} & Best & \textbf{.907} & (.006) & \textbf{.643} & (.009) & \textbf{.753} & (.007) & \textbf{11.22} & (3.73) & 5551.77 & (4942.42) \\
		& Default & \textbf{.980} & (.003) & \textbf{.502} & (.009) & \textbf{.664} & (.008) & \textbf{12.15} & (3.73) & 66510.56 & (34645.38) \\
		\bottomrule
	\end{tabular}
	\label{tab:simulations_trend}
\end{table}

\clearpage
\subsubsection{Change in mean}

The synthetic data generation process for this type of change is detailed in Appendix \ref{sec:data_generation_process}, and we provide in Figure \ref{fig:illustration_mean} an illustration of this change type and the considered change sizes. The successful detection region corresponds to the time region in which detected changepoints are counted as true positives. With respect to Definition \ref{def:f1}, we use a left margin $\mu_l = 0$ and a right margin $\mu_r = 30$. Performance of the grid of tested models is shown in Figure \ref{fig:simulations_mean} and a summary of performance is provided in Table \ref{tab:simulations_mean}.

\begin{figure}[h!]
	\centering
	\includegraphics[width=\linewidth]{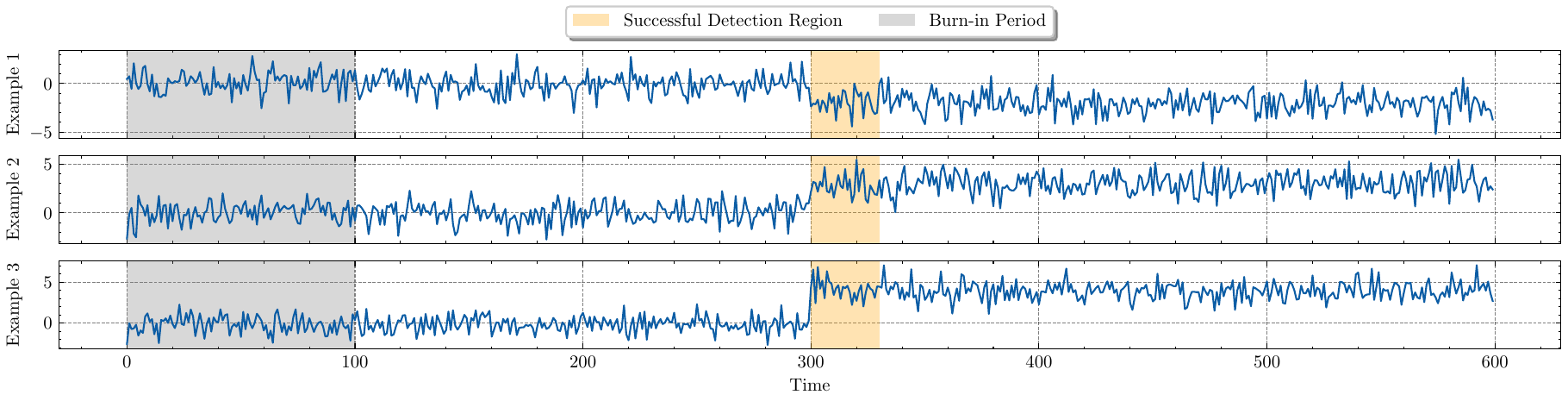}
	\caption{Illustration of the different change sizes considered in the mean change simulations.}
	\label{fig:illustration_mean}
\end{figure}

\begin{figure}[h!]
	\centering
	\includegraphics[width=\linewidth]{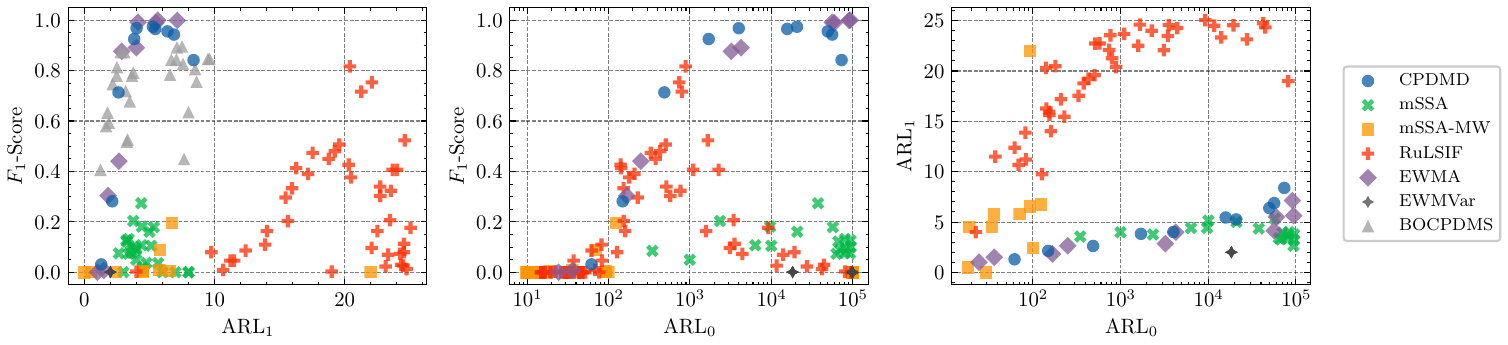}
	\caption{Performance comparison of changepoint detection algorithms on synthetic data for a change in mean. Each dot corresponds to an algorithm with a specific choice of parameter values.}
	\label{fig:simulations_mean}
\end{figure}

\begin{table}[h!]
	\centering
	\smaller
	\caption{Evaluation of selected changepoint detection algorithms on synthetic data, comparing performance achieved with the parameter set yielding highest $F_1$-Score (\textit{best}) against default parametrisation (\textit{default}). We highlight the top two performance metrics for each set of parameters in {bold}.}
	\begin{tabular}{l l r@{\hskip\tabcolsep} >{\scriptsize}r@{\hskip\tabcolsep} r@{\hskip\tabcolsep} >{\scriptsize}r@{\hskip\tabcolsep} r@{\hskip\tabcolsep} >{\scriptsize}r@{\hskip\tabcolsep} r@{\hskip\tabcolsep} >{\scriptsize}r@{\hskip\tabcolsep} r@{\hskip\tabcolsep} >{\scriptsize}r@{\hskip\tabcolsep}}
		\toprule
		\textbf{Algorithm} & \textbf{Params.} & \multicolumn{2}{c}{$\boldsymbol{P}$} & \multicolumn{2}{c}{$\boldsymbol{R}$} & \multicolumn{2}{c}{$\boldsymbol{F_1}$} & \multicolumn{2}{c}{$\mathbf{ARL_1}$} & \multicolumn{2}{c}{$\mathbf{ARL_0}$} \\
		\midrule
		\multirow{2}{*}{EWMA} & Best & \textbf{1.000}  & (.000) & \textbf{1.000} & (.000) & \textbf{1.000} & (.000) & 5.64 & (3.42) & \textbf{94051.07} & (21322.11) \\
		& Default & .441 & (.009)  & .441 & (.009) & .441  & (.009) & \textbf{2.65} & (1.78) & 250.47 & (290.29) \\
		\hline
		\multirow{2}{*}{EWMVar} & Best & .004 & (.003) & .001 & (.000) & .001 & (.001) & \textbf{2.00} & (1.41) & 18408.77 & (31222.62) \\
		& Default & .000  & (.000)  & .000 & (.000) & .000 & (.000) & N/A & (N/A) & \textbf{99900.00}  & (0.00) \\
		\hline
		\multirow{2}{*}{BOCPDMS} & Best & .905 & (.024) & .887 & (.026) & .896 & (.025) & 7.50 & (5.97) & N/A & (N/A) \\
		& Default & \textbf{.682} &(.040)  &\textbf{.673}  & (.040) & \textbf{.678} & (.040) & \textbf{3.48} & (3.00) & N/A & (N/A) \\
		\hline
		\multirow{2}{*}{RuLSIF} & Best & .817 & (.022) & .817 & (.022) & .817 & (.022) & 20.40 & (3.41) & 890.10  & (615.54) \\
		& Default & .427 &(.028)  & .427 & (.028) & .427 & (.028) & 20.30 & (5.80) & 141.60 & (97.35) \\
		\hline
		\multirow{2}{*}{mSSA} & Best & .480 & (.014) & .192 & (.007) & .274 & (.009) & \textbf{4.36} & (2.66) & \textbf{37678.13} & (45640.53) \\
		& Default & .219 &(.009)  & .155 & (.007) &.181  & (.007) & 4.41 & (3.04) & 9717.39 & (27538.01) \\
		\hline
		\multirow{2}{*}{mSSA-MW} & Best & .201 &(.008)  & .190 & (.007) &  .196 & (.007) & 6.73 & (6.35) & 123.82 & (101.40) \\
		& Default &.000  & (.000)  & .000 & (.000) & .000 & (.000) & N/A & (N/A) & 14.22 & (6.40) \\
		\hline
		\multirow{2}{*}{\textbf{CPDMD (ours})} & Best & \textbf{.995} & (0.001) & \textbf{.953} & (0.004) & \textbf{.974} & (0.002) & 5.26 & (4.43) &20833.05  & (21516.31) \\
		& Default & \textbf{.990} & (0.002)  & \textbf{.899} & (0.006) & \textbf{.943} & (0.003) & 6.88 & (4.76) & \textbf{56526.97} & (35946.23) \\
		\bottomrule
	\end{tabular}
	\label{tab:simulations_mean}
\end{table}

\clearpage
\subsubsection{Change in variance}

The synthetic data generation process for this type of change is detailed in Appendix \ref{sec:data_generation_process}, and we provide in Figure \ref{fig:illustration_variance} an illustration of this change type and the considered change sizes. The successful detection region corresponds to the time region in which detected changepoints are counted as true positives. With respect to Definition \ref{def:f1}, we use a left margin $\mu_l = 0$ and a right margin $\mu_r = 30$. Performance of the grid of tested models is shown in Figure \ref{fig:simulations_variance} and a summary of performance is provided in Table \ref{tab:simulations_variance}.

\begin{figure}[h!]
	\centering
	\includegraphics[width=\linewidth]{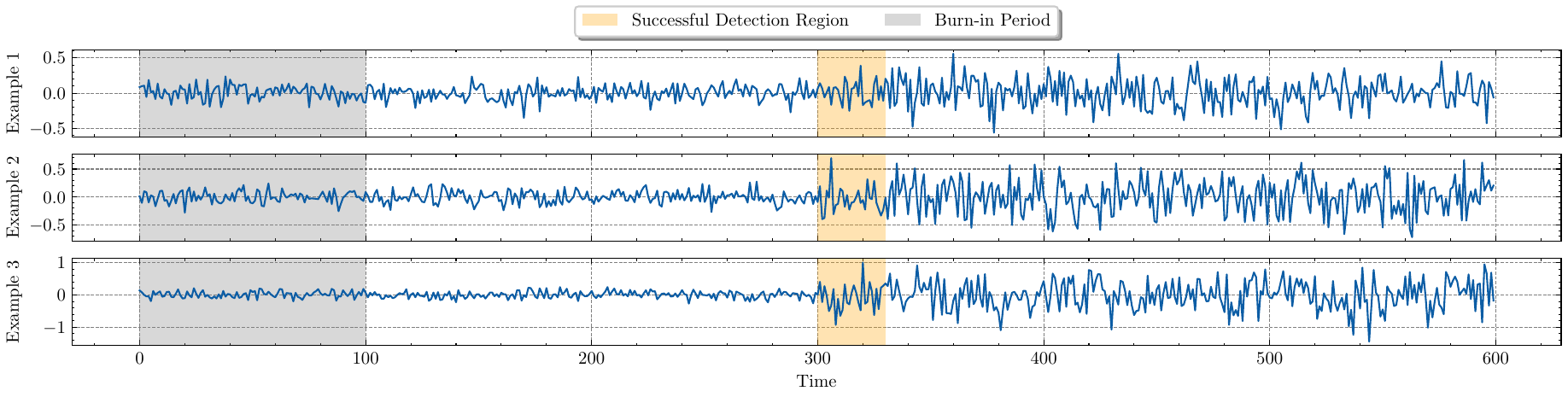}
	\caption{\small Illustration of the different change sizes considered in the variance change simulations.}
	\label{fig:illustration_variance}
\end{figure}

\begin{figure}[h!]
	\centering
	\includegraphics[width=\linewidth]{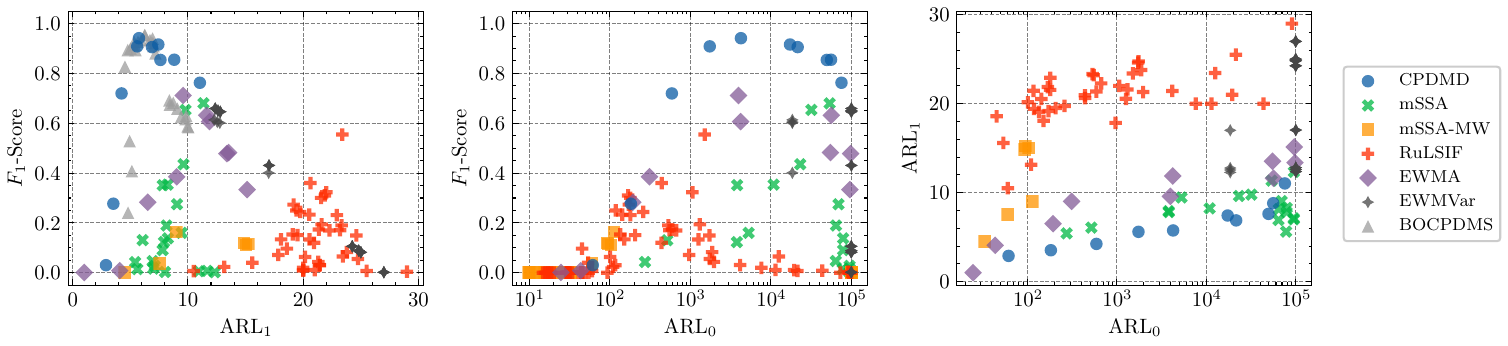}
	\caption{Performance comparison of changepoint detection algorithms on synthetic data for a change in variance. Each dot corresponds to an algorithm with a specific choice of parameter values.}
	\label{fig:simulations_variance}
\end{figure}

\begin{table}[h!]
	\centering
	\smaller
	\caption{Evaluation of selected changepoint detection algorithms on synthetic data, comparing performance achieved with the parameter set yielding highest $F_1$-Score (\textit{best}) against default parametrisation (\textit{default}). We highlight the top two performance metrics for each set of parameters in {bold}.}
	\begin{tabular}{l l r@{\hskip\tabcolsep} >{\scriptsize}r@{\hskip\tabcolsep} r@{\hskip\tabcolsep} >{\scriptsize}r@{\hskip\tabcolsep} r@{\hskip\tabcolsep} >{\scriptsize}r@{\hskip\tabcolsep} r@{\hskip\tabcolsep} >{\scriptsize}r@{\hskip\tabcolsep} r@{\hskip\tabcolsep} >{\scriptsize}r@{\hskip\tabcolsep}}
		\toprule
		\textbf{Algorithm} & \textbf{Params.} & \multicolumn{2}{c}{$\boldsymbol{P}$} & \multicolumn{2}{c}{$\boldsymbol{R}$} & \multicolumn{2}{c}{$\boldsymbol{F_1}$} & \multicolumn{2}{c}{$\mathbf{ARL_1}$} & \multicolumn{2}{c}{$\mathbf{ARL_0}$} \\
		\midrule
		\multirow{2}{*}{EWMA} & Best & .711 & (.009) & .711 & (.009) & .711 & (.009) & 9.58 & (7.90) & 3959.58 & (6248.10) \\
		& Default & .385 & (.009) & .385 & (.009) & .385 & (.009) & 9.02 & (7.68) & 310.69 & (373.74) \\
		\hline
		\multirow{2}{*}{EWMVar} & Best & .704 & (.009) & .620 & (.009) & .659 & (.009) & 12.36 & (7.40) & \textbf{99900.00} & (0.00) \\
		& Default & .704 & (.009) & .620 & (.009) & .659 & (.009) & 12.36 & (7.40) & \textbf{99900.00} & (0.00) \\
		\hline
		\multirow{2}{*}{BOCPDMS} & Best & \textbf{.953} & (.017) & \textbf{.953} & (.017)  & \textbf{.953} & (.017) & \textbf{6.24} & (7.10) & N/A & (N/A) \\
		& Default & \textbf{.926} & (.022) & \textbf{.920} & (.022) & \textbf{.923} &  (.022)& \textbf{6.96} & (7.38) & N/A & (N/A) \\
		\hline
		\multirow{2}{*}{RuLSIF} & Best & .622 & (.032) & .500 & (.029) & .555 &(.029) & 23.39 & (3.63) & 1506.40 & (1438.51) \\
		& Default & .301 & (.026) & .300 & (.026) & .301 & (.026) & 21.56 & (6.98) & 179.40 & (195.94) \\
		\hline
		\multirow{2}{*}{mSSA} & Best & .696 & (.009) & .665 & (.009) & .680 & (.008) & 11.35 & (6.85) & \textbf{53624.57} & (47612.65) \\
		& Default & .354 & (.009) & .354 & (.009) & .354 & (.009) & \textbf{8.23} & (5.97) & 10837.83 & (28977.50) \\
		\hline
		\multirow{2}{*}{mSSA-MW} & Best & .164 &(.007)  & .164 & (.007) & .164 & (.007) & 9.00 & (6.13) & 113.26 & (101.68) \\
		& Default & .000 & (.000) & .000 & (.000) & .000 & (.000) & N/A & (N/A) & 14.70 & (6.44) \\
		\hline
		\multirow{2}{*}{\textbf{CPDMD (ours})} & Best & \textbf{.942} & (.004) & \textbf{.941} & (.004) & \textbf{.941} & (.004) & \textbf{5.75} & (5.88) & 4244.30 & (4398.62) \\
		& Default & \textbf{.899} & (.006) & \textbf{.815} & (.007) & \textbf{.855} & (.006) & 8.81 & (6.51) & \textbf{55844.03} & (38654.98) \\
		\bottomrule
	\end{tabular}
	\label{tab:simulations_variance}
\end{table}

\clearpage
\subsubsection{Change in double periodicity}

The synthetic data generation process for this type of change is detailed in Appendix \ref{sec:data_generation_process}, and we provide in Figure \ref{fig:illustration_double} an illustration of this change type and the considered change sizes. The successful detection region corresponds to the time region in which detected changepoints are counted as true positives. With respect to Definition \ref{def:f1}, we use a left margin $\mu_l = 0$ and a right margin $\mu_r = 30$. Performance of the grid of tested models is shown in Figure \ref{fig:simulations_double} and a summary of performance is provided in Table \ref{tab:simulations_double}.

\begin{figure}[h!]
	\centering
	\includegraphics[width=\linewidth]{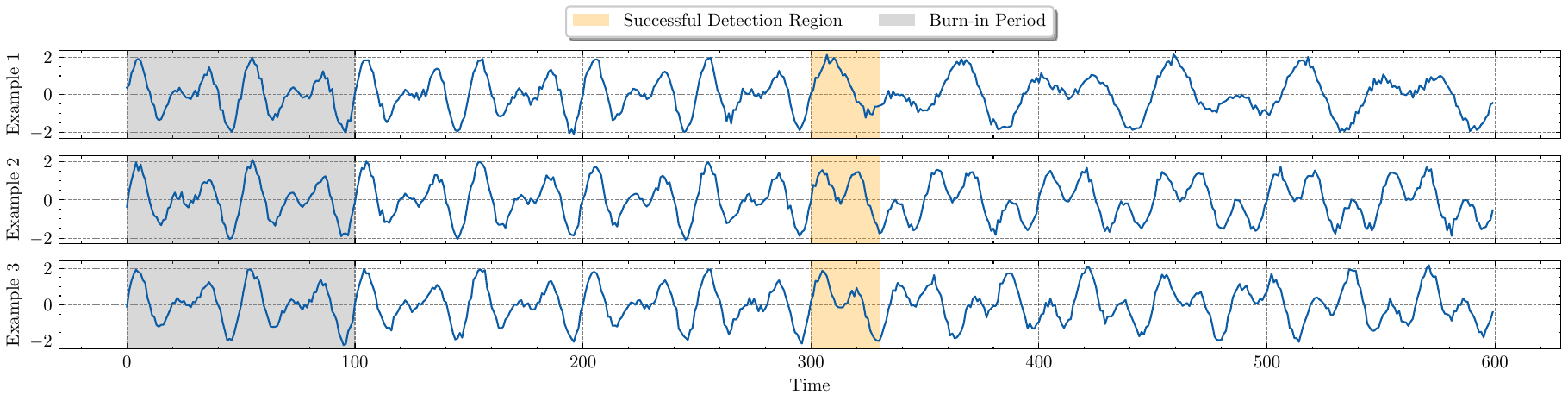}
	\caption{Illustration of the different change sizes considered in the double periodicity change simulations.}
	\label{fig:illustration_double}
\end{figure}

\begin{figure}[h!]
	\centering
	\includegraphics[width=\linewidth]{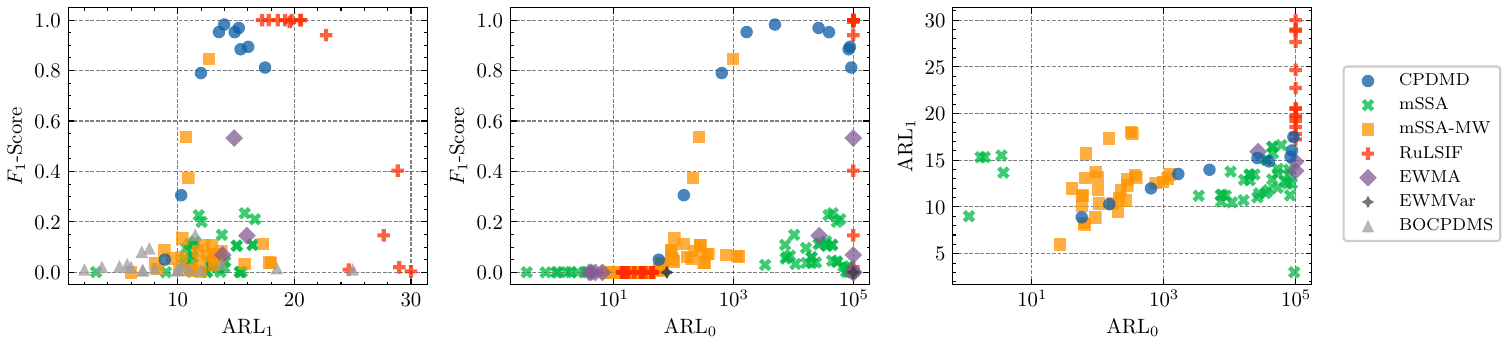}
	\caption{Performance comparison of changepoint detection algorithms on synthetic data for a change in double periodicity. Each dot corresponds to an algorithm with a specific choice of parameter values.}
	\label{fig:simulations_double}
\end{figure}

\begin{table}[h!]
	\centering
	\smaller
	\caption{Evaluation of selected changepoint detection algorithms on synthetic data, comparing performance achieved with the parameter set yielding highest $F_1$-Score (\textit{best}) against default parametrisation (\textit{default}). We highlight the top two performance metrics for each set of parameters in {bold}.}
	\begin{tabular}{l l r@{\hskip\tabcolsep} >{\scriptsize}r@{\hskip\tabcolsep} r@{\hskip\tabcolsep} >{\scriptsize}r@{\hskip\tabcolsep} r@{\hskip\tabcolsep} >{\scriptsize}r@{\hskip\tabcolsep} r@{\hskip\tabcolsep} >{\scriptsize}r@{\hskip\tabcolsep} r@{\hskip\tabcolsep} >{\scriptsize}r@{\hskip\tabcolsep}}
		\toprule
		\textbf{Algorithm} & \textbf{Params.} & \multicolumn{2}{c}{$\boldsymbol{P}$} & \multicolumn{2}{c}{$\boldsymbol{R}$} & \multicolumn{2}{c}{$\boldsymbol{F_1}$} & \multicolumn{2}{c}{$\mathbf{ARL_1}$} & \multicolumn{2}{c}{$\mathbf{ARL_0}$} \\
		\midrule
		\multirow{2}{*}{EWMA} & Best & .982 & (.004) & .365 & (.008) & .533 & (.009) & 14.84 & (2.83) & \textbf{99900.00} & (0.00) \\
		& Default & .000 & (.000) & .000 & (.000) & .000 & (.000) & N/A & (N/A) & 6.57 & (0.50) \\
		\hline
		\multirow{2}{*}{EWMVar} & Best & .000 & (.000) & .000 & (.000) & .000 & (.000) & N/A & (N/A) & \textbf{99900.00} & (0.00) \\
		& Default & .000 & (.000) & .000 & (.000) & .000 & (.000) & N/A & (N/A) & \textbf{99900.00} & (0.00) \\
		\hline
		\multirow{2}{*}{BOCPDMS} & Best & .152 & (.030) & .147 & (.029) & .149 & (.029) & \textbf{11.50} & (10.37) & N/A & (N/A) \\
		& Default & .017 & (.018) & .007 & (.007) & .010 & (.009) & 10.00 & (N/A) & N/A & (N/A) \\
		\hline
		\multirow{2}{*}{RuLSIF} & Best & \textbf{1.000} & (.000) & \textbf{1.000} & (.000) & \textbf{1.000} & (.000) & 17.24 & (1.55) & \textbf{99910.00} & (0.00) \\
		& Default & .000 & (.000) & .000 & (.000) & .000 & (.000) & N/A & (N/A) & \textbf{99920.00} & (0.00) \\
		\hline
		\multirow{2}{*}{mSSA} & Best & .342 & (.012) & .179 & (.007) & .235 & (.009)& 15.75 & (4.88) & 45381.45 & (47480.18) \\
		& Default & .000 & (.000) & .000 & (.000) & .000 & (.000)& \textbf{9.00} & (N/A) & 1.12 & (0.97) \\
		\hline
		\multirow{2}{*}{mSSA-MW} & Best & .846 &  (.007)& .846 & (.007) & .846 & (.007) & \textbf{12.67} & (4.51) & 981.60 & (983.34) \\
		& Default & \textbf{.086} & (.005) & \textbf{.086} & (.005) & \textbf{.086} & (.005) & \textbf{8.86} & (3.65) & 91.92 & (70.20) \\
		\hline
		\multirow{2}{*}{\textbf{CPDMD (ours})} & Best & \textbf{.990} & (.002) & \textbf{.975} & (.003) & \textbf{.983} & (.002) & 13.98 & (5.63) & 4938.27 & (4722.74) \\
		& Default & \textbf{.999} & (.001) & \textbf{.809} & (.007) & \textbf{.894} & (.004) & 16.05 & (5.71) & 87872.59 & (28593.59) \\
		\bottomrule
	\end{tabular}
	\label{tab:simulations_double}
\end{table}

\clearpage
\subsection{Real-world experiments}

\subsubsection{Datasets}
\label{sec:datasets}

In this section, we provide additional details regarding real-world datasets used in Section \ref{sec:experiments_real_world}.

\paragraph{HASC \cite{hasc}.}
Published by the Human Activity Sensing Consortium (HASC), this dataset consists of human activity data recorded via a wearable three-axis accelerometer. Each time series correspond to 120 seconds of acceleration measurements along the $x$, $y$, and $z$ axes. The changepoints indicate transitions between six distinct activities: stay, walk, jogging, skip, stair-up, and stair-down. The dataset contains $18$ $3$-dimensional sequences of length $\sim 12,000$. We evaluate performance of selected algorithms on a subset of the HASC dataset made of the first 10 recordings. This sample dataset was provided by the 2011 Human Activity Sensing Consortium challenge and has been previously analysed in changepoint detection studies \cite{mssa, density_ratio_cpd}. This dataset can be downloaded from \href{http://hasc.jp/hc2011/download-en.html}{\small \texttt{http://hasc.jp/hc2011/download-en.html}}. Following the definition of $F_1$-Score in Appendix \ref{sec:background}, and to take into account fluctuations in the human annotations for changepoints locations, we consider a left margin $\mu_l = 50$ and a right margin $\mu_r = 200$. 

\paragraph{Digits \cite{digits}.}
The \textit{Optical Recognition of Handwritten Digits Dataset} consists of $8 \times 8$ greyscale images (one channel) of a digit. Each pixel corresponds to an integer in the range $0$-$16$. In total, it contains $1,797$ samples, leading to $\sim 180$ digits per class. We generated sequences by concatenating piecewise-constant digit segments, where images were randomly sampled from the original dataset. Each sequence spans a length of $3,000$ with approximately $600$ timesteps between changepoints. The resulting sequence can be interpreted as a video displaying a single digit that abruptly transitions to a new digit at the changepoints locations. Ultimately, each image is flattened, yielding sequences of length $3,000$ with $64$ feature dimensions. The original dataset can be downloaded from \href{https://archive.ics.uci.edu/dataset/80/optical+recognition+of+handwritten+digits}{\small \texttt{https://archive.ics.uci.edu/dataset/80/optical+recognition+of+handwritten+digits}}. Following the definition of $F_1$-Score in Appendix \ref{sec:background}, while the images are extracted from a real-world dataset, we know exactly the changepoint locations and it can be assumed that there is no fluctuation in terms of annotation in this dataset, leading to a left margin $\mu_l = 0$ and a right margin $\mu_r = 50$.  

\paragraph{Yahoo \cite{yahoo}.}
The \textit{S5 - A Labeled Anomaly Detection Dataset} is provided as part of the Yahoo! Webscope program. More specifically, we use the A1 Benchmark which is based on the real production traffic to some of the Yahoo! properties, and consists of $67$ univariate time series of length $\sim 1,000$ with seasonalities and trends. Both anomalies and changepoints are labelled by human annotations. This dataset has already been used in the changepoint detection literature \cite{mssa, cpd_contrastive} to measure performance of algorithms in real-world web traffic data. The dataset can be accessed from \href{https://webscope.sandbox.yahoo.com/catalog.php?datatype=s&did=70}{\small \texttt{https://webscope.sandbox.yahoo.com/catalog.php?datatype=s\&did=70}}. Following the definition of $F_1$-Score in Appendix \ref{sec:background}, and to take into account fluctuations in the human annotations for changepoints locations, we consider a left margin $\mu_l = 50$ and a right margin $\mu_r = 50$. 

\subsubsection{Models parameters}
\label{sec:appendix_real_world_parameters}

In this section, we detail the different methods and parameters that were explored through simulations on real-world datasets in Section \ref{sec:experiments_real_world}. Since each real-world dataset has specific settings (i.e. dimension, sequence length, average inter-changepoint times), we evaluate the performance of each selected algorithm on a grid of parameters which is suited for a given dataset, while trying to make the comparisons as fair as possible and using similar burn-in periods (or equivalent parameters).

Since each considered dataset consists of multiple data streams, we evaluate algorithms' performance via the defined performance metrics ($F_1$-Score and covering metric), along with the standard deviation of these performance metrics per set of parameters across all sequences in each dataset. 

\paragraph{CPDMD.}
Given a burn-in period $T_0$, we use the same grid generation process as in synthetic data simulations detailed in Section \ref{sec:experiment_details_synthetic}, with explored ranks $r \in 10 \cdot \{ 1, 2, \dots, p \}$ due to the higher complexity of the generation process of real-world data compared to synthetic data. Note that while increasing the maximum rank, this does not modify the number of explored values for the rank but only consists in increasing the spacing between the explored rank values, and allowing to explore higher values, while avoiding to explore the complete set of possible rank values. 

\begin{itemize}
	
	\item \textbf{Learning rate} $\lambda = 0.05$.
	\item \textbf{Control limit} $L = 2.5, 3.5, 4.5 \text{ (default)}$.
	
	\item \textbf{HASC dataset}
	\begin{itemize}
		\item \textbf{Burn-in period} $T_0 = 300, 400 \text{ (default)}$.
	\end{itemize}
	
	\item \textbf{Digits dataset}
	\begin{itemize}
		\item \textbf{Burn-in period} $T_0 = 100, 150 \text{ (default)}$.
	\end{itemize}
	
	\item \textbf{Yahoo dataset}
	\begin{itemize}
		\item \textbf{Burn-in period} $T_0 = 100, 200 \text{ (default)}$.
	\end{itemize}
	
\end{itemize}

\paragraph{mSSA and mSSA-MW.}
We use the Python implementation provided by the authors of \cite{mssa} and available via the GitHub repository \href{https://github.com/ArwaAlanqary/mSSA_cpd}{\small \texttt{https://github.com/ArwaAlanqary/mSSA\_cpd}}. We follow the authors guidelines and consider the following parameters:

\begin{itemize}
	
	\item \textbf{Distance threshold} $= 5 \text{ (default)}, 10$.
	\item \textbf{Rank} $= 0.95$.
	\item \textbf{Training size} $= 0.9$.
	
	\item \textbf{HASC dataset}
	\begin{itemize}
		\item \textbf{Window size} $= 300, 400 \text{ (default)}$.
		\item \textbf{Rows} $= 17, 20 \text{ (default)}$.
	\end{itemize}
	
	\item \textbf{Digits dataset}
	\begin{itemize}
		\item \textbf{Window size} $= 100, 150 \text{ (default)}$.
		\item \textbf{Rows} $= 10, 12 \text{ (default)}$.
	\end{itemize}
	
	\item \textbf{HASC dataset}
	\begin{itemize}
		\item \textbf{Window size} $= 100, 200 \text{ (default)}$.
		\item \textbf{Rows} $= 10, 14 \text{ (default)}$.
	\end{itemize}
\end{itemize}

\paragraph{RuLSIF.}
We use the MATLAB implementation provided by the authors of \cite{rulsif} and available via the GitHub repository \href{https://github.com/anewgithubname/change_detection}{\small \texttt{https://github.com/anewgithubname/change\_detection}}. While RuLSIF is introduced as a retrospective changepoint detection method, its implementation is analogous to online context, where a threshold needs to be set (and cannot be known in advance) to make online decisions. We follow the authors guidelines and consider the following parameters:
\begin{itemize}
	\item $\boldsymbol{n} = 50 \text{ (default)}, 75$.
	\item $\boldsymbol{k} = 10 \text{ (default)}, 15$.
	\item $\boldsymbol{\alpha} = 0.1$.
	\item \textbf{Threshold} $= 7$.
\end{itemize}

\paragraph{BOCPDMS.}
We use the Python implementation provided by the authors of \cite{bocpdms} and available via the GitHub repository \href{https://github.com/alan-turing-institute/bocpdms}{\small \texttt{https://github.com/alan-turing-institute/bocpdms}}. We consider the same following grid of parameters as in \cite{mssa}

\begin{itemize}
	\item \textbf{Intensity} $= 100 \text{ (default)}, 200$.
	\item \textbf{Prior on $\boldsymbol{a}$} $=$ \textbf{Prior on $\boldsymbol{b}$} $= 0.01, 1 \text{ (default)}$.
	
\end{itemize}

\paragraph{EWMA.}
We use a self-implementation of the EWMA algorithm introduced by \cite{ewma} in Python, and we consider the following grid of parameters:
\begin{itemize}
	\item \textbf{Learning rate} $r = 0.05 \text{ (default)}, 0.1$.
	\item \textbf{Control limit} $L = 2.5 \text{ (default)}, 3.5, 4.5$.
	
	\item \textbf{HASC dataset}
	\begin{itemize}
		\item \textbf{Burn-in period} $T_0 = 300, 400 \text{ (default)}$.
	\end{itemize}
	
	\item \textbf{Digits dataset}
	\begin{itemize}
		\item \textbf{Burn-in period} $T_0 = 100, 150 \text{ (default)}$.
	\end{itemize}
	
	\item \textbf{Yahoo dataset}
	\begin{itemize}
		\item \textbf{Burn-in period} $T_0 = 100, 200 \text{ (default)}$.
	\end{itemize}
	
\end{itemize}

\paragraph{EWMVar.}
We use a self-implementation of the EWMVar algorithm introduced by \cite{ewmvar} in Python, and we consider the following grid of parameters:
\begin{itemize}
	\item \textbf{Learning rate for the variance} $r = 0.05$.
	\item \textbf{Learning rate for the mean} $\lambda = 0.2$.
	\item \textbf{Control limits} $(C_7, C_8) = (0.42, 1.7), (0.63, 1.14), (0.68, 1.08) \text{ (default)}, (0.91, 1.12)$.
	\item \textbf{HASC dataset}
	\begin{itemize}
		\item \textbf{Burn-in period} $T_0 = 300, 400 \text{ (default)}$.
	\end{itemize}
	
	\item \textbf{Digits dataset}
	\begin{itemize}
		\item \textbf{Burn-in period} $T_0 = 100, 150 \text{ (default)}$.
	\end{itemize}
	
	\item \textbf{Yahoo dataset}
	\begin{itemize}
		\item \textbf{Burn-in period} $T_0 = 100, 200 \text{ (default)}$.
	\end{itemize}
\end{itemize}

While algorithms like EWMA and EWMVar are designed for univariate changepoint detection, we adapt them to the multivariate datasets (HASC and Digits) by running separate instances per component and taking the union of detected changepoints across all components.

\clearpage

\subsubsection{Examples of real-world data segmentations}

We provide in Figures \ref{fig:real_world_hasc} and \ref{fig:real_world_yahoo} illustrations of CPDMD online changepoint detection results. 

\begin{figure}[h!]
	\centering
	\includegraphics[width=\linewidth]{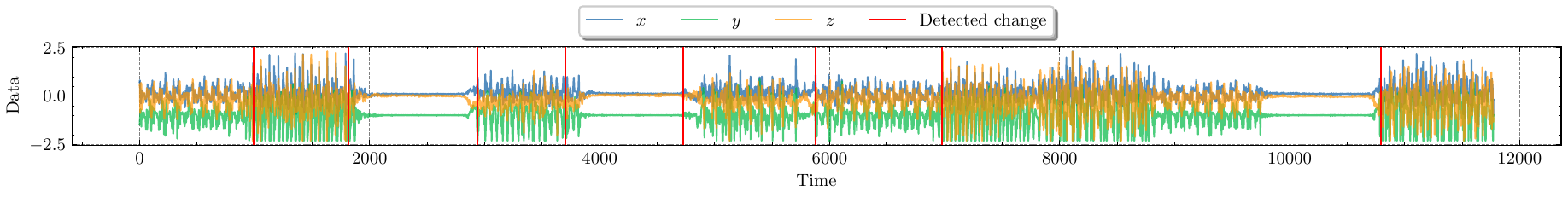}
	\caption{Example of CPDMD online segmentation on a multivariate sequence extracted from the HASC dataset.}
	\label{fig:real_world_hasc}
\end{figure}

\begin{figure}[h!]
	\centering
	\includegraphics[width=\linewidth]{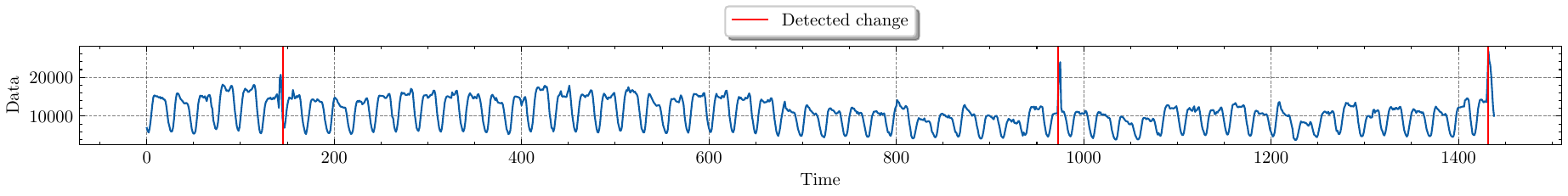}
	\caption{Example of CPDMD online segmentation on a univariate sequence extracted from the Yahoo dataset.}
	\label{fig:real_world_yahoo}
\end{figure}

\clearpage
\small
\bibliographystyle{plain}
\bibliography{refs}


\end{document}